\documentclass[12pt]{iopart}

\usepackage[pctex32]{graphics}
\usepackage{amssymb}
\usepackage{latexsym}
\usepackage{epsfig}
\usepackage{multirow}
\usepackage{pstricks}
\usepackage{float}
\usepackage{epstopdf}
\usepackage{mathrsfs}

\usepackage{graphicx}
\usepackage{amssymb}
\usepackage{relsize}
\usepackage{multirow}

\usepackage{graphicx}
\usepackage{multicol}

\usepackage{times}
\usepackage{epstopdf}
\usepackage{amsthm,verbatim,enumerate,lipsum}

\newtheorem{definition}{Definition}[section]

\newtheorem{lemma}[definition]{Lemma}

\newtheorem{proposition}[definition]{Proposition}

\newtheorem{remark}[definition]{Remark}
\newtheorem{assumptions}[definition]{Assumptions}
\newtheorem{example}[definition]{Example}

\definecolor{Red}{rgb}{1,0.,0.}
\definecolor{darkred}{rgb}{.7,0.,0.}

\usepackage{fancyhdr}
\usepackage{setspace}

\newcommand{\R}{{\mathbb R}}
\newcommand{\bA}{{\mathsf A}_\star}

\newcommand{\mGamma}{{\mathsf \Gamma}}

\newcommand{\mM}{{\mathsf M}}
\newcommand{\mC}{{\mathsf C}}
\newcommand{\mR}{{\mathsf R}}
\newcommand{\mT}{{\mathsf T}}
\newcommand{\mL}{{\mathsf L}}

\newcommand{\mA}{{\mathsf A}}
\newcommand{\mB}{{\mathsf B}}
\newcommand{\mD}{{\mathsf D}}
\newcommand{\mK}{{\mathsf K}}
\newcommand{\mI}{{\mathsf I}}
\newcommand{\mQ}{{\mathsf Q}}
\newcommand{\mV}{{\mathsf V}}
\newcommand{\mH}{{\mathsf H}}
\newcommand{\mSigma}{{\mathsf \Sigma}}

\newcommand{\eps}{\varepsilon}
\newcommand{\dee}{{d}}

\newcommand{\post}{\mathrm{post}}

\makeatletter
\newcommand{\pushright}[1]{\ifmeasuring@#1\else\omit\hfill$\displaystyle#1$\fi\ignorespaces}
\newcommand{\pushleft}[1]{\ifmeasuring@#1\else\omit$\displaystyle#1$\hfill\fi\ignorespaces}
\makeatother

\begin{document}

\title[Iterative Updating of Model Error for Bayesian Inversion]{Iterative Updating of Model Error for Bayesian Inversion}

\author{Daniela Calvetti$^1$, Matthew Dunlop$^2$, Erkki Somersalo$^1$, Andrew Stuart$^2$}

\address{$^1$Case Western Reserve University, Department of Mathematics, Applied Mathematics and Statistics, 10900 Euclid Avenue, Cleveland, OH 44106, USA}
\address{$^2$California Institute of Technology, Computing \& Mathematical Sciences, 1200 E California Boulevard, Pasadena, California, CA 91125, USA}
\ead{dxc57@case.edu,mdunlop@caltech.edu, ejs49@case.edu,astuart@caltech.edu}

\begin{abstract}
In computational inverse problems, it is common that a detailed and accurate forward model is approximated by a computationally less challenging substitute. The model reduction may be necessary to meet constraints in computing time when optimization algorithms are used to find a single estimate, or to speed up Markov chain Monte Carlo (MCMC) calculations in the Bayesian framework. The use of an approximate model introduces a discrepancy, or modeling error, that may have a detrimental effect on the solution of the ill-posed inverse problem, or it may severely distort the estimate of the posterior distribution. In the Bayesian paradigm, the modeling error can be considered as a random variable, and by using an estimate of the probability distribution of the unknown, one may estimate the probability distribution of the modeling error and incorporate it into the inversion. We introduce an algorithm which iterates this idea to
update the distribution of the model error, leading to a sequence of posterior distributions that 
are demonstrated empiricially to capture the underlying truth with 
increasing accuracy. Since the algorithm is not based on rejections, it requires only limited full model evaluations.

We show analytically that, in the linear Gaussian case, the algorithm converges geometrically fast with respect to the number of iterations. For more general models, we introduce particle approximations of the iteratively generated sequence of 
distributions;
we also prove that each element of the sequence converges in the large particle limit. We show numerically that, as in the linear case, rapid convergence 
occurs with respect to the number of iterations. Additionally, we show through computed examples that point estimates obtained from this iterative algorithm are superior to those obtained by neglecting the model error.
 
\end{abstract}

%
\vspace{2pc}

\section{Introduction}

The traditional way of describing an inverse problem is to define a forward map relating the unknown to an observed quantity, and to look for an estimate of the unknown when the data is corrupted by noise. In this description, it is often tacitly assumed that an underlying "truth" exists, and the noiseless data arises from applying the forward map on this true value. On the other hand, it is commonly acknowledged that a mathematical model does not coincide with the reality, and therefore part of the noise must be attributed to the model discrepancy, or the mismatch between the model and the reality. Modeling this discrepancy is an active research topic in statistics -- see
\cite{kennedy2001bayesian,brynjarsdottir,briol2015probabilistic} and the references
therein; it is also a closely related to the concept of the ``inverse crime'', a procedure of testing a computational method with data that has been generated by the same model that is used to solve the inverse problem \cite{Colton,KSart}. 

Common sources of modeling errors in inverse problems include:
\begin{enumerate}
\item model reduction -- a complex, computationally intensive model is replaced by a simpler, less demanding model;
\item parametric reduction -- in a model depending on poorly known parameters, some of them are frozen to fixed values, assuming that the solution is not sensitive to them;
\item unknown geometry -- a computational domain of unknown shape is approximated by a standard geometry.
\end{enumerate}

Including the modeling error into the computations in the traditional deterministic 
setting may not be straightforward. Recasting the inverse problem via Bayesian 
inference provides tools to carry this out in a natural statistical fashion. The present article introduces and analyzes
a Bayesian methodology for model error estimation which demonstrably leads to improved estimates of the true unknown function
generating the data.  

\subsection{Background}
\label{ssec:background}

In this article, we consider the problem of estimating an unknown quantity $u$ based on indirect observations. In the Bayesian framework, the prior belief about the quantity $u$ is encoded in the prior probability distribution $\mathbb{P}(u)$, and given the observed data $b$, the posterior 
distribution $\mathbb{P}(u\mid b)$ follows from Bayes' formula,
\[
\mathbb{P}(u\mid b) \propto \mathbb{P}(b\mid u)\mathbb{P}(u),
\]
where $\propto$ denotes proportionality up to a scaling constant depending on
$b$ but not on $u$, and  the distribution $\mathbb{P}(b\mid u)$  of $b$ is  the likelihood. To construct the likelihood, a forward model from $u$ to $b$ needs to be specified. A commonly used model, assuming additive observation noise that is independent of the unknown $u$, is
\[
b = F(u) + \eps
\]
where $\eps \sim \pi_{\mathrm{noise}}(\cdot)$ is a realization of random noise, 
and $F$ is a mapping defined typically on the parameter space for $u$,
which is often infinite-dimensional ($u$ is a function) or high dimensional
(a function has been discretized to obtain $u$). Under these assumptions, the likelihood is equal to the distribution of $\eps$ with mean shifted by $F(u)$, i.e.,
\[
\mathbb{P}(b\mid u) = \pi_{\mathrm{noise}}(b-F(u)).
\]
To estimate $u$, or derived quantities based on it, numerical approximations of integrals with respect to the posterior distribution are required, and a common approach is to use sampling methods such as Markov chain Monte Carlo (MCMC). This requires a large number of repeated evaluations of the forward map $F$, which is often expensive to evaluate numerically. A particular instance that we have in mind is the situation where
evaluation of $F$ requires numerical solution of a partial differential equation. If computational resources or time are an issue, an attractive approach is to trade off the accuracy of evaluations with the computational cost by adjusting the resolution of the mesh that the PDE is solved upon. Denoting by $f$ an approximation to $F$ on a coarse mesh, a model for the data can be written as
\[
b = f(u) + m + \eps
\]
where $m = F(u) - f(u)$ denotes the model error induced by moving from the accurate model $F$ to the approximate one. 
If we ignore the fact that $m$ depends on $u$, and instead model it as additive independent noise, the conditional likelihood $\mathbb{P}(b\mid u,m)$ is then given by
\[
\mathbb{P}(b\mid u,m) = \pi_{\mathrm{noise}}(b-f(u)-m);
\] 
evaluations of this map then only require evaluation of the approximate map $f$. Furthermore, the likelihood $\mathbb{P}(b\mid u)$ can be found by marginalizing out the model error $m$. However, the marginalization requires the distribution of $m$ which is not known. As suggested in \cite{KSbook}, the Bayesian approach provides a natural approximate solution to this problem: By using the prior distribution of $u$ and the model error mapping $M(u) = F(u)-f(u)$, one can generate a sample of model errors to estimate the model error distribution. This approach, referred to as the \emph{enhanced error model}, has been shown to produce more accurate point estimates than those that come from neglecting the model error (the \emph{conventional error model}), see, e.g., \cite{arridge2006approximation, Heino, KSart, banasiak2012improving} for static inverse problems, and \cite{huttunen2007approximation, huttunen2, huttunen2010importance} for extensions to dynamic inverse problems.

In \cite{Oliver}, the enhanced error model was developed further using the observation that the the posterior distribution based on the error model contains refined information about the unknown $u$ beyond the point estimate determined by the enhanced error model; as a consequence the model error distribution can be updated by pushing forward the distribution under the model error mapping $M$. When the data are particularly informative, posterior samples may differ significantly from prior samples, and this should produce a much better approximation to the distribution of the model error, potentially yielding a better approximation of the posterior distribution of $u$, and at the
very least providing point estimates of higher accuracy. The procedure can be iterated to produce a sequence of approximate posterior distributions that have the potential to yield a sequence of point estimates of improved accuracy;
they may also approximate the true posterior distribution with increasing accuracy. In this article, we address this approach in a systematic way, with particular focus on convergence of the iterative updating.

The effect of model error and model discrepancy in Bayesian inference is a widely studied topic. Early works focus primarily on code uncertainty -- the uncertainty that arises due to expense of forward model evaluations meaning that it is only practical
to compute outputs of the model for a limited finite number of inputs. A review of work in this direction is given in \cite{sacks1989design}, including the problems of optimal choice of inputs at which to evaluate the forward model, and how to predict the output for inputs away from the computed ones. In \cite{kennedy2001bayesian} the numerous sources of uncertainty within computational inverse problems are discussed, including those arising from model inadequacy and code uncertainty. The authors model this error as a function independent of the approximate model, which can be justified in certain cases. Hierarchical Gaussian priors are placed upon the model and the model error, and the model and error are then linked by imposing correlations between the hyperparameters. The technique has subsequently been developed further in \cite{brynjarsdottir}, and used in, for example, the context of model validation \cite{bayarri2007framework} and uncertainty quantification \cite{higdon2004combining}. 
More recent work in probabilistic numerical methods \cite{briol2015probabilistic}
provides a unifying perspective on this
body of work, linking it to earlier research connecting numerical algorithms to
Bayesian statistics \cite{diaconis1988bayesian}.

\subsection{Our Contribution}
\label{ssec:contribution}
\begin{itemize}
\item We develop further the iterative updating of the posterior probability densities based on repeated updating of the model error distribution, leading to an approximation of the posterior probability density. While the approximation error is defined
through the computationally expensive accurate model, the posterior approximation 
we introduce relies primarily on the computationally inexpensive approximate model,
and a limited number of evaluations of the accurate model.
\item In the case where the models are linear and the noise and prior distributions are Gaussian, we show that the means and covariances of the resulting sequence of posterior distributions converge to a limit geometrically fast.
\item For more general models and prior/noise distributions we introduce particle approximations to allow the algorithm to be implemented numerically, and show convergence of these approximations in the large particle limit.
\item We illustrate numerically the effectiveness of the algorithms in multiple different settings, showing the advantage over the conventional and enhanced error models.
\end{itemize}

\subsection{Outline}
\label{ssec:outline}
The iterative approach of updating the posterior distributions is introduced in Section \ref{sec:formulation}.
In Section \ref{sec:gaussian} we focus on the particular case where the forward model is linear, and the noise and prior distributions are Gaussian. The assumptions imply that the approximate posterior distributions are also Gaussian, and can therefore be characterized by their means and covariances. We identify conditions guaranteeing the convergence of the sequence of approximate posteriors to a non-degenerate limit as the number of iterations tends to infinity. In Section \ref{sec:sampling} we discuss different sampling methods which may be used to implement the algorithm in practice. In particular we focus on particle methods that require a finite number of full model evaluations in order to estimate the modeling error and posterior distribution, and show convergence 
to the correct sequence of approximate distributions in the large particle limit. Finally, in Section \ref{sec:numerics}, we provide numerical illustrations of the behavior and performance of the algorithm for three different forward models.
Section \ref{sec:conc} contains the conclusions and discussion.

\section{Problem Formulation}
\label{sec:formulation}

We start by introducing the main ingredients of the iterative algorithm: the accurate and approximate models are defined in Subsection \ref{ssec:models},  along with some examples whose details will be discussed later on. The enhanced error model \cite{KSbook} is reviewed in Subsection \ref{ssec:enhanced}, prompting the question of how to update the density of the modeling error. In subsection \ref{ssec:algorithm} we provide an iterative algorithm for doing this, in the case where all measures involved have Lebesgue densities; a more general algorithm is provided in the appendix for cases such as those
arising when the measures are defined on infinite-dimensional function spaces.

\subsection{Accurate vs. Approximate Model}
\label{ssec:models}
Let $X, Y$ be two Banach spaces representing the parameter and data spaces, respectively. Let
\[
 F:X\to Y, \quad u\mapsto b
\]
denote a reference forward model, referred to as the {\em accurate model}, and let the {\em approximate model} be denoted by
\[
 f:X\to Y, \quad u\mapsto b.
\]

We write the observation model using the accurate model,
\begin{equation}\label{obs model}
 b = F(u) + \varepsilon,
\end{equation}
and equivalently, using the approximate model, as
\begin{eqnarray}\label{appr model}
 b &=& f(u) + \big(F(u) - f(u)\big) + \varepsilon \nonumber \\
  &=& f(u) + m + \varepsilon,
\end{eqnarray}
where $m$ represents the modeling error,
\[
 m = F(u) - f(u) = M(u).
\]
In light of the above observation, we may view the data as coming from the approximate model,
with an error term which reflects both observational noise
$\varepsilon$ and modeling error $m$. The main problem 
addressed in this paper is how 
to model the probability distribution of the model error. We study this
question with the goal of providing computations to estimate the unknown $u$
from $b$ using only the less expensive model $f$, and not the true model $F$, 
without creating modeling artifacts. 

We conclude this subsection by giving two examples of approximate models, both
of which may be relevant in applications.

\begin{example}[Linearization] In electrical impedance tomography (EIT), the goal is to estimate the conductivity distribution inside a body from a finite set of current/voltage measurements at the boundary, as discussed in more detail in Subsection \ref{ssec:eit}.
We denote by $F$ the differentiable non-linear forward model, mapping the appropriately parametrized conductivity distribution to the voltage measurements. We define the approximate model through the linearization,
\[
 f(u) = F(u_0) + {\mathsf D}F(u_0)(u - u_0),
\]
where ${\mathsf D} F(u_0)$ is the Jacobian of the forward map, and $u_0$ is a fixed parameter value, representing, e.g., a constant conductivity background.
\end{example}

\begin{example}[Coarse Mesh] In the EIT model, the accurate model represents the forward model computed with a FEM grid fine enough to guarantee that the numerical solution approximates the solution of the underlying PDE within a required precision. To speed up computations, we introduce the reduced model $f$ as the forward model based on FEM built on a coarse grid. We assume that both computational grids are built on an underlying independent discretization of the conductivity $\sigma$, appropriately parametrized, and the FEM stiffness matrices that require integration over elements are computed by evaluating $\sigma$ in the Gauss points of the elements, respectively.
\end{example}

\subsection{The Enhanced Error Model}
\label{ssec:enhanced}
We start by reviewing the basic ideas of the enhanced error model and, for simplicity, assume here that  $X$ and $Y$ are Euclidean
spaces and that all probability distributions are expressible in terms of  Lebesgue densities. We assume that $u$ is an $X$-valued random variable with a given {\em a priori} density,
\[
 u \sim \pi_{\rm prior}(u).
\]
Furthermore, we assume that the additive noise $\varepsilon$ is a 
$Y$-valued random variable, independent of $u$, with the density
\[
 \varepsilon\sim \pi_{\rm noise}(\varepsilon).
\]
In view of the model (\ref{obs model}), we may write the likelihood as
\[
 \pi(b\mid u) = \pi_{\rm noise}(b - F(u)),
\]
and the posterior density is, according to Bayes' formula, given by
\[
 \pi(u\mid b) \propto \pi_{\rm prior}(u)\pi_{\rm noise}(b - F(u)).
\]
If, instead, we want to use the approximate model $f(u)$ for the data, the modeling error $m$ needs to be taken into account. The idea of the {\em enhanced error model} in \cite{KSbook} is the following: Given the prior distribution $\mu_0(\dee u) = \pi_{\rm prior}(u)\dee u$, with no other information about $u$, the probability distribution of $m=M(u)$ is obtained as a push-forward of the prior distribution:
\[
 m\sim M^\#\mu_0.
\]
To obtain a computationally efficient formalism, the distribution of $m$ is approximated by a Gaussian distribution sharing the mean and
covariance of the push-forward measure,
\[
 m\sim{\mathcal N}(\overline m,\mSigma),
\]
where in practice the mean and covariance may be estimated numerically by sampling the modeling error.
The Gaussian approximation is particularly convenient if the additive noise is Gaussian,
\[
 \varepsilon\sim{\mathcal N}(0,\mGamma),
\]
as it leads to the approximate likelihood model
\[
 \pi(b\mid u) \propto {\rm exp}\left(-\frac 12 (b - f(u) - \overline m)^\mT(\mGamma + \mSigma)^{-1}(b - f(u) - \overline m)\right),
\]
and, consequently, to the posterior model
\[
\pi(u\mid b) \propto \pi_{\rm prior}(u)\, {\rm exp}\left(-\frac 12 (b - f(u) - \overline m)^\mT(
\mGamma + \mSigma)^{-1}(b - f(u) - \overline m)\right).
\]
Note that the enhanced error model can be interpreted as a form of variance inflation: if the model error is neglected as in the conventional error model, then the covariance matrix in the likelihood would be smaller in the sense of positive-definite quadratic forms, since $\mSigma$ is non-negative definite. This is to be expected as the lack of accuracy in the model contributes additional uncertainty to the problem; we do not wish to be over-confident in an inaccurate model. In the next section we explain how this idea may be built upon to
produce the algorithms studied in this paper.

\subsection{The Iterative Algorithm}
\label{ssec:algorithm}

In this subsection we generalize the preceding enhanced error model
in two ways: (i) we iterate the construction of the model error, updating its probability distribution by pushing 
forward by $M$ the measure $\mu_{\ell}$ , the posterior distribution of $u$ when the model
error distribution is computed as the pushforward under
$M$ of the measure $\mu_{\ell-1};$ in this iterative method 
we choose $\mu_0$ to be the prior and so the first step is analogous to what is described
in the previous subsection;
(ii) we do not invoke a Gaussian approximation of
the model error, leaving open other possibilities for practical 
implementation.
We describe the resulting algorithm here in the case where Lebesgue
densities exist, and refer to the appendix for its formulation in a more abstract setting.

{\bf Algorithm (Lebesgue densities).} Let $\mu_\ell$ denote the posterior distribution at stage $\ell$, with density $\pi_\ell$, so that $\mu_\ell(\dee u) = \pi_\ell(u)\,\dee u$. Denote $\pi_{\ell}(b\mid u)$ the likelihood at stage $\ell$.
\begin{enumerate}[\hspace{0.2cm}1.]
\item Set $\pi_0(u) = \pi_{\rm prior}(u)$ and $\ell = 0$.
\item Given $\mu_\ell$, assume $m \sim M^\#\mu_\ell$. Assuming that $u$ and $m$ are mutually independent, we have
\[
\pi(b\mid u,m) = \pi_{\rm noise}(b - f(u) - m),
\]
and by marginalization,
\begin{eqnarray*}
\pi_{\ell+1}(b\mid u) &= \int_Y \pi_{\rm noise}(b - f(u) - m)(M^\#\mu_\ell)(\dee m)\\
&= \int_X \pi_{\rm noise}(b - f(u) - M(z))\pi_\ell(z)\dee z.
\end{eqnarray*}
Hence using Bayes' theorem, update the posterior distribution:
\begin{eqnarray}
\label{eq:update_leb}
\pi_{\ell+1}(u) \propto \pi_{\rm prior}(u)\int_X \pi_{\rm noise}(b - f(u) - M(z))\pi_\ell(z)\dee z.
\end{eqnarray}
\item Set $\ell \mapsto \ell+1$ and go to 2.\qed
\end{enumerate}

We can give an explicit expression for the above density $\pi_{\ell}(u)$:
\[
\fl \pi_\ell(u) \propto \pi_{\rm prior}(u)\int_X\cdots\int_X \left(\prod_{i=1}^\ell\pi_{\rm noise}(b - f(z_{i+1}) - M(z_i))\pi_{\rm prior}(z_i)\right)\,\dee z_1\ldots\dee z_\ell
\]
where we define $z_{\ell+1} = u$.

The above algorithm can be generalized to the case when no Lebesgue densities exist, such as will be the case on infinite dimensional Banach spaces, see the Appendix. 

\section{The Linear Gaussian Case}
\label{sec:gaussian}
We analyze the convergence of the general algorithm in the case where both the accurate model $F$ and the approximate model $f$ are linear, and the noise and prior distributions are Gaussian. With these assumptions the posterior distributions forming the sequence remain Gaussian, and are hence characterized by the sequences of means and covariances. Convergence properties of the iteration can be established by using the updating formulas for these sequences. Though the iteration is not immediately implementable, since the full matrix for the accurate model $F$ is required to calculate each covariance matrix, the explicit convergence results give insight into the implementable variants of the algorithm introduced in Section \ref{sec:sampling}, as well as for non-linear forward maps.

In subsection \ref{ssec:linear_evolution}, we first describe how the posterior density evolves, establishing the iterations that the means and covariances satisfy. In subsection \ref{ssec:conv_infinite} we show that if the model error is sufficiently small, the sequences of posterior means and covariances converge geometrically fast. Moreover, despite the repeated incorporation of the data into the posterior, the limiting distribution does not become singular. Additionally, in subsection \ref{ssec:conv_finite} we show that in finite dimensions the assumption of small model error is not needed in order to establish convergence of the covariance.

\subsection{Evolution of the Posterior Distribution}
\label{ssec:linear_evolution}
 When the noise distribution $\pi_{\rm noise}$ and prior distribution $\pi_{\rm prior}$ are Gaussian, the measure $\mu_\ell$ is Gaussian at each stage $\ell$, and we can write down expressions for the evolution of its mean and covariance.
Let $X,Y$ be separable Hilbert spaces and $\bA:X\rightarrow Y$ a linear operator. Assume that the data $b \in Y$ arise from $\bA$ via
\[
b = \bA u + \eps,\;\;\;\eps \sim \mathcal{N}(0,\mGamma)
\]
where $\mGamma:Y\rightarrow Y$ is a symmetric positive definite covariance operator. Let $\mA:X\rightarrow Y$ denote an approximation to $\bA$ so that the expression for $b$ may be rewritten
\begin{eqnarray*}
b &= \mA u + (\bA - \mA)u + \eps.
\end{eqnarray*}
We define the model error operator $\mM:X\rightarrow Y$ by $\mM = \bA - \mA$, so that $\mM u$ represents the (unknown) model error. We assume that the model error is Gaussian, with unknown mean $w$ and covariance $\mSigma$. Additionally we assume that it is independent of the observation noise. The data is now given by
\begin{eqnarray*}
b &= \mA u + \hat{\eps},\;\;\;\hat{\eps} \sim \mathcal{N}(w,\mSigma + \mGamma).
\end{eqnarray*}
Let $\mu_0 = N(m_0,\mC_0)$ denote the prior distribution on $u$. We first estimate $w$ and $\mSigma$ by pushing forward $\mu_0$ by the model error operator $\mM$:
\[
\mM^{\#}\mu_0  = \mathcal{N}(\mM m_0, \mM \mC_0 \mM^*) \equiv \mathcal{N}(\hat{w}_1,\hat{\mSigma}_1).
\]
Then, assuming for now that the model error has this distribution, the resulting posterior distribution on $u$ can be calculated as $\mu_1 = N(m_1,\mC_1)$, where
\begin{eqnarray*}
\mC_1 &= \big(\mA^*(\mGamma + \mM \mC_0 \mM^*)^{-1}\mA + \mC_0^{-1}\big)^{-1},\\
m_1 &= \mC_1\big(\mA^*(\mGamma + \mM\mC_0 \mM^*)^{-1}(b - \mM m_0) + \mC_0^{-1}m_0\big).
\end{eqnarray*}
As described previously, in order to obtain a better approximation of the model error, the above step can be repeated so that the measure pushed forward to approximate the model error is closer to the posterior distribution. 
Therefore, in the next step, we approximate the distribution of the model error by 
\[
\mM^{\#}\mu_1 = \mathcal{N}(\mM m_1,\mM\mC_1\mM^*) \equiv \mathcal{N}(\hat{w}_2,\hat{\mSigma}_2),
\]
and make a further posterior approximation. Iterating this process, we see that the posterior mean and covariance evolve via
\begin{eqnarray}
\label{eq:ck_standard}
\mC_{\ell+1} &= \big(\mA^*(\mGamma + \mM \mC_\ell \mM^*)^{-1}\mA + \mC_0^{-1}\big)^{-1},\\
m_{\ell+1} &= \mC_{\ell+1}\big(\mA^*(\mGamma + \mM\mC_\ell \mM^*)^{-1}(b - \mM m_\ell) + \mC_0^{-1}m_0\big)\nonumber.
\end{eqnarray}
We wish to show that these sequences are convergent.

We may write the above iteration in an equivalent form as 
\begin{eqnarray*}
\mC_{\ell+1} &= \mC_0 - \mC_0\mA^*(\mGamma + \mM\mC_\ell \mM^* + \mA\mC_0\mA^*)^{-1}\mA\mC_0,\\
m_{\ell+1} &= m_0 + \mC_0\mA^*(\mGamma + \mM\mC_\ell \mM^* + \mA\mC_0\mA^*)^{-1}(b - \mA m_0 - \mM m_\ell),
\end{eqnarray*}
using results from \cite{lps,mandelbaum}, assuming that $\mC_0$ is trace-class and $\mGamma$ is trace-class or white. This form has the advantage that the unbounded operator $\mC_0^{-1}$ does not appear, and so we need not worry about its domain of definition. In what follows we simply assume that this equivalent expression
for the evolution of the means and covariances is valid.

In the following subsections, we consider two different cases. In the first one, we limit the data into a finite dimensional space $Y$ but let the space $X$ to be a Hilbert space. The convergence of the algorithm is demonstrated under certain restrictive conditions: the modeling error needs to be small enough. In the second case, we also limit the unknown $u$ to a finite dimensional space, and show that in this case the convergence proof can be obtained without the restrictions needed in the former case. 
We emphasize that, although we establish convergence of the iteration in
various settings, the limiting distribution does not coincide with the
true posterior distribution found in the absence of model error. Nonetheless
our numerical experiments will show that the limiting distribution can
significantly improve point estimates of the underlying value used to generate the data.

\subsection{Convergence in Infinite Dimensions}
\label{ssec:conv_infinite}
We introduce a scalar parameter $\delta$ controlling the accuracy of the approximation $\mA$ of $\bA$, writing $\delta \mM$ in place of $\mM$. By writing explicitly the dependence of the mean and covariance on $\delta$, we have
\begin{eqnarray*}
\mC_{\ell+1}(\delta) &= \mC_0 - \mC_0 \mA^*\big(\mGamma + \delta^2 \mM\mC_\ell(\delta) \mM^* + \mA\mC_0 \mA^*\big)^{-1}\mA\mC_0,\\
m_{\ell+1}(\delta) &= m_0 + \mC_0\mA^*\big(\mGamma + \delta^2\mM\mC_\ell(\delta)\mM^* + \mA\mC_0\mA^*\big)^{-1}\big(b - \mA m_0 - \delta \mM m_\ell(\delta)\big).
\end{eqnarray*}
Let $\mathcal{L}(X)$ denote the space of bounded linear operators on $X$, equipped with the operator norm. Let $S_+(X)\subseteq\mathcal{L}(X)$ denote the set of positive bounded linear operators on $X$, and $\overline{S_+(X)}$ the set of non-negative bounded linear operators on $X$. We may write the iteration for $\mC_\ell(\delta)$ as
\begin{eqnarray}
\label{eq:ck}
\mC_{\ell+1}(\delta) = {\mathscr F}\big(\mC_\ell(\delta),\delta\big)
\end{eqnarray}
where ${\mathscr F}:\overline{S_+(X)}\times\mathbb{R}\rightarrow \overline{S_+(X)}$ is given by
\[
{\mathscr F}(\mB,\delta) = \mC_0 - \mC_0 \mA^*\big(\mGamma + \delta^2 \mM\mB\mM^* + \mA\mC_0 \mA^*\big)^{-1}\mA\mC_0.
\]
We show that under certain assumptions, for all $\delta$ sufficiently small, ${\mathscr F}(\,\cdot\,,\delta)$ has a unique stable fixed point $\mC(\delta)$. The assumptions we make are as follows.
\begin{assumptions}
\label{assump}
\begin{enumerate}[(i)]
\item $\mC_0 \in S_+(X)$ and is trace-class. 
\item $Y$ is finite dimensional, and $\mGamma \in S_+(Y)$.
\item $\bA,\,\mA:X\rightarrow Y$ are bounded.
\end{enumerate}
\end{assumptions}

We first establish the following result concerning convergence of the sequence of covariance operators: 
\begin{proposition}
\label{prop:ckconv}
Let Assumptions \ref{assump} hold. Then there is a $\beta > 0$ such that for all $\delta < 1/\beta$, a unique $\mC(\delta) \in S_+(X)$ exists with
\[
\mC(\delta) = {\mathscr F}\big(\mC(\delta),\delta\big).
\]
Moreover, $\mC(\delta)$ is a stable fixed point of ${\mathscr F}(\,\cdot\,,\delta)$, and there is a constant $\alpha_1$ such that
\[
\|\mC_\ell(\delta) - \mC(\delta)\|_{\mathcal{L}(X)} \leq \alpha_1(\beta\delta)^{2\ell}\;\;\;\mbox{for all }\ell \geq 1.
\]
In particular, for $\delta < 1/\beta$, the sequence $\{\mC_\ell(\delta)\}_{\ell\geq 1}$ converges geometrically fast.
\end{proposition}

From this geometric convergence of $\{\mC_\ell(\delta)\}_{\ell\geq 1}$ we can deduce that also the means $m_\ell(\delta)$ converge: Define the maps $G_\ell:X\times\mathbb{R}\rightarrow X$, $\ell \geq 0$, by 
\[
G_\ell(m,\delta) = m_0 + \mC_0\mA^*\big(\mGamma + \delta^2\mM\mC_\ell(\delta)\mM^* + \mA\mC_0\mA^*\big)^{-1}\big(b - \mA m_0 - \delta \mM m\big)
\]
so that the update for $m_\ell(\delta)$ is given by
\[
m_{\ell+1}(\delta) = G_\ell\big(m_\ell(\delta),\delta\big).
\]
Define also the limiting map $G:X\times\mathbb{R}\rightarrow X$ by
\[
G(m,\delta) = m_0 + \mC_0\mA^*\big(\mGamma + \delta^2\mM\mC(\delta)\mM^* + \mA\mC_0\mA^*\big)^{-1}\big(b - \mA m_0 - \delta \mM m\big).
\]
Then we have the following result:

\begin{proposition}
\label{prop:mkconv}
Let Assumptions \ref{assump} hold, and let $\beta$ be as in Proposition \ref{prop:ckconv}. Then for all $\delta < 1/\beta$, there exists a unique $m(\delta) \in X$ with
\[
m(\delta) = G\big(m(\delta),\delta\big).
\]
Moreover, there is an $\alpha_2 > 0$ such that
\[
\|m_\ell(\delta) - m(\delta)\|_X \leq \alpha_2(\beta\delta)^\ell\;\;\;\mbox{ for all }\ell \geq 1.
\]
Hence, for $\delta < 1/\beta$, the sequence $\{m_\ell(\delta)\}_{\ell\geq 1}$ converges geometrically fast.
\end{proposition}

To prove the above propositions we first prove the following lemma. In the proof, the following notation is used: Given symmetric non-negative linear operators $\mB_1,\mB_2 \in \mathcal{L}(X)$, we write $\mB_1 \leq \mB_2$ to mean that $\mB_2 - \mB_1$ is non-negative.
\begin{lemma}
\label{lem:ckbound}
Let Assumptions \ref{assump} hold. Then the family of operators $\mK(\mB,\delta):Y\rightarrow X$ given by
\[
\mK(\mB,\delta) = \mC_0\mA^*(\mGamma + \delta^2\mM\mB\mM^* + \mA\mC_0\mA^*)^{-1}
\]
is bounded uniformly over $\mB \in \overline{S_+(X)}$ and $\delta \in \mathbb{R}$.
\end{lemma}

\begin{proof}
We have that
\[
\mGamma + \mA\mC_0\mA^* \leq \mGamma + \mQ + \mA\mC_0 \mA^*,
\]
for any $\mQ \in \overline{S_+(Y)}$, which implies that
\[
(\mGamma + \mQ + \mA\mC_0 \mA^*)^{-1} \leq (\mGamma + \mA\mC_0A^*)^{-1},
\]
and consequently
\[
\|(\mGamma + \mQ + \mA\mC_0 \mA^*)^{-1}\|_{\mathcal{L}(Y)} \leq \|(\mGamma + \mA\mC_0\mA^*)^{-1}\|_{\mathcal{L}(Y)}.
\]
By choosing $\mQ = \delta^2\mM\mB\mM^*$ the claim follows.
\end{proof}
In what follows we will denote $K_{\max} = \sup\big\{\|\mK(\mB,\delta)\|_{\mathcal{L}(Y,X)}\;|\;\mB \in \overline{S_+(X)}, \delta \in \mathbb{R}\big\}$. Furthermore, we define the parameter $\beta = K_{\max}\| \mM\|_{{\mathcal L}(Y)}$.

\begin{proof}[Proof of Proposition \ref{prop:ckconv}]
We first show that for $\delta$ sufficiently small, the map ${\mathscr F}(\,\cdot\,,\delta)$ is a contraction on $\overline{S_+(X)}$. To do this, we look at the Fr\'{e}chet derivative of the map, which may be calculated explicitly. For  $B \in \overline{S_+(X)}$ and $\mV \in \mathcal{L}(X)$, we have
\begin{eqnarray*}
\fl {\mathsf D}_\mB{\mathscr F}&(\mB,\delta)\mV\\
\fl &= \delta^2\mC_0\mA^*(\mGamma + \delta^2\mM\mB\mM^* + \mA\mC_0\mA^*)^{-1}\mM\mV\mM^*(\mGamma + \delta^2\mM\mB\mM^*+\mA\mC_0\mA^*)^{-1}\mA\mC_0\\
\fl &= \delta^2\mK(\mB,\delta)\mM\mV\mM^*\mK(\mB,\delta)^*
\end{eqnarray*}
where $\mK(\mB,\delta)$ is as defined in Lemma \ref{lem:ckbound}. The norm of the derivative can be estimated as
\begin{eqnarray*}
\big\|\mD_\mB{\mathscr F}\big(\mB,\delta\big)\big\|_{\mathcal{L}(X)\rightarrow\mathcal{L}(X)} &= \sup_{\|\mV\|_{\mathcal{L}(X)} = 1} \big\|\mD_\mB{\mathscr F}\big(\mB,\delta\big)\mV\big\|_{\mathcal{L}(X)}\\
&\leq \sup_{\|\mV\|_{\mathcal{L}(X)} = 1} \delta^2\|\mK(\mB,\delta)\|_{{\mathcal L}(Y,X)}^2\|\mM\|_{\mathcal{L}(Y)}^2 \|\mV\|_{\mathcal{L}(X)}\\
&\leq (\beta\delta)^2
\end{eqnarray*}
by the estimate of Lemma \ref{lem:ckbound}. Since the above bound is uniform in $\mB$, we may use the mean value theorem to deduce that for all $\mB_1,\mB_2 \in \overline{S_+(X)}$,
\[
\|{\mathscr F}(\mB_1,\delta) - {\mathscr F}(\mB_2,\delta)\|_{\mathcal{L}(X)} \leq (\beta\delta)^2\|\mB_1-\mB_2\|_{\mathcal{L}(X)}
\]
and so ${\mathscr F}(\,\cdot\,,\delta)$ is a contraction for $\delta < 1/\beta$. The set $\overline{S_+(X)}$ is a complete subset of the space $\mathcal{L}(X)$, and so by the Banach fixed-point theorem there exists a unique $\mC(\delta) \in \overline{S_+(X)}$ such that
\[
\mC(\delta) = {\mathscr F}(\mC(\delta),\delta),
\]
and we have $\mC_\ell(\delta)\rightarrow \mC(\delta)$. Moreover, 
\begin{eqnarray*}
 \|\mC_\ell(\delta) - \mC(\delta)\|_{\mathcal{L}(X)}  &=& \|{\mathscr F}(\mC_{\ell-1}(\delta),\delta) - {\mathscr F}(\mC(\delta),\delta)\|_{\mathcal{L}(X)}  \\
 &\leq& (\beta\delta)^2\|\mC_{\ell-1}(\delta) - \mC(\delta)\|_{\mathcal{L}(X)},
\end{eqnarray*}
and recursively,
\[
\|\mC_\ell(\delta) - \mC(\delta)\|_{\mathcal{L}(X)}  \leq (\beta\delta)^{2\ell}  \|\mC_1(\delta) - \mC_0\|_{{\mathcal L}(X)} = \alpha_1(\beta\delta)^{2\ell}.
\]

We finally show that we actually have $\mC(\delta) \in S_+(X)$ and so the covariance does not become degenerate in the limit. We denote $\mC_{\post} = {\mathscr F}(\mC_{\post},0)$ the exact posterior covariance in the absence of model error, noting that $\mC_{\post} \in S_+(X)$ as we assume $\mC_0 \in S_+(X)$ and $\mGamma \in S_+(Y)$. From a similar argument as in the proof of Lemma \ref{lem:ckbound}, we have
\begin{eqnarray*}
0 < \mC_{\post} &= \mC_0 - \mC_0\mA^*(\mGamma + \mA\mC_0\mA^*)^{-1}\mA\mC_0\\
&\leq \mC_0 - \mC_0\mA^*(\mGamma + \delta^2\mM\mC(\delta)\mM^* + \mA\mC_0\mA^*)^{-1}\mA\mC_0\\
&= \mC(\delta),
\end{eqnarray*}
which gives the result.
\end{proof}

\begin{proof}[Proof of Proposition \ref{prop:mkconv}]
We may express $G_\ell$ in the form of an affine mapping,
\[
G_\ell(m,\delta) = \mH_\ell(\delta)m + g_\ell(\delta)
\]
where $\mH_\ell(\delta)$ and $g_\ell(\delta)$ are given by
\begin{eqnarray*}
\mH_\ell(\delta) &= -\delta \mC_0\mA^*\big(\mGamma + \delta^2\mM\mC_\ell(\delta)\mM^* + \mA\mC_0\mA^*\big)^{-1}\mM  \\
&= -\delta\mK(\mC_\ell(\delta),\delta)\mM,\\
g_\ell(\delta) &= m_0 + \mC_0\mA^*\big(\mGamma + \delta^2\mM\mC_\ell(\delta)\mM^* + \mA\mC_0\mA^*\big)^{-1}(b - \mA m_0) \\
 &= m_0 + \mK(\mC_\ell(\delta),\delta)(b - \mA m_0),
\end{eqnarray*}
respectively. From the estimates of Lemma \ref{lem:ckbound}, we obtain the uniform bounds
\begin{eqnarray*}
\|\mH_\ell(\delta)\|_{\mathcal{L}(X)} &\leq  \beta\delta<1,\\
\|g_\ell(\delta)\|_X &\leq \|m_0\|_X + \big\|\mK(\mC_\ell(\delta),\delta)\big\|_{\mathcal{L}(Y)}\|b - \mA m_0\|_X\\
&= \|m_0\|_X +\beta \|b - \mA m_0\|_X =  L.
\end{eqnarray*}

From the convergence of the sequence $\{\mC_\ell(\delta)\}_{\ell\geq 1}$ in the previous proposition, we see that $\{\mH_\ell(\delta)\}_{\ell\geq 1}$ and $\{g_\ell(\delta)\}_{\ell\geq 1}$ also converge, the limits being denoted by $\mH(\delta)$ and $g(\delta)$, respectively. Explicitly,
\begin{eqnarray*}
\mH(\delta) &=  -\delta \mK(\mC(\delta),\delta)\mM,\\
g(\delta) &= m_0 +  \mK(\mC(\delta),\delta)(b - \mA m_0).
\end{eqnarray*}
Moreover, since $\mB\mapsto \mK(\mB,\delta)$ is Fr\'{e}chet differentiable, this convergence occurs at the same rate as the convergence of $\{\mC_\ell(\delta)\}_{\ell\geq 1}$. 

Next we show that $\{m_\ell(\delta)\}_{\ell \geq 1}$ remains bounded for sufficiently small $\delta$. From the bounds above, we have
\begin{eqnarray*}
\|m_\ell(\delta)\|_X &\leq \|\mH_{\ell-1}(\delta)\|_{\mathcal{L}(X)}\|m_{\ell-1}(\delta)\|_X + \|g_{\ell-1}(\delta)\|_X\\
&\leq \beta\delta\|m_{\ell-1}\|_X + L,
\end{eqnarray*}
and therefore, by repeatedly applying the estimate, we obtain
\begin{eqnarray*}
\|m_\ell(\delta)\|_X &\leq (\beta\delta)^\ell\|m_0\|_X + L\sum_{j=0}^{\ell-1}(\beta\delta)^j\\
&\leq (\beta\delta)^\ell\|m_0\|_X + \frac{L}{1-\beta\delta},
\end{eqnarray*}
which provides a uniform bound for $\delta < 1/\beta$.

To prove the convergence, we write first for $i \geq 1$ the estimate
\begin{eqnarray*}
\|m_{i+1}(\delta) - m_i(\delta)\|_X &= \|\mH_i(\delta)m_i(\delta) + g_i(\delta) - \mH_{i-1}(\delta)m_{i-1}(\delta) - g_{i-1}(\delta)\|_X\\
&\leq \|\mH_i(\delta)m_i(\delta) - \mH_i(\delta)m_{i-1}(\delta)\|_X\\
&\hspace{1cm} + \|\mH_i(\delta)m_{i-1}(\delta) - \mH_{i-1}(\delta)m_{i-1}(\delta)\|_X\\
&\hspace{1cm} + \|g_i(\delta) - g_{i-1}(\delta)\|_X\\
&\leq \|\mH_i(\delta)\|_{\mathcal{L}(X)}\|m_i(\delta) - m_{i-1}(\delta)\|_X\\
&\hspace{1cm} + \|\mH_i(\delta) - \mH_{i-1}(\delta)\|_{\mathcal{L}(X)}\|m_{i-1}(\delta)\|_X\\
&\hspace{1cm} + \|g_i(\delta) - g_{i-1}(\delta)\|_X,
\end{eqnarray*}
and further, by the geometric convergence of the sequences $\{\mH_i(\delta)\}$ and $\{g_{i}(\delta)\}$, and the uniform boundedness, for some $\gamma>0$,
\[
\|m_{i+1}(\delta) - m_i(\delta)\|_X \leq \beta\delta\|m_i(\delta) - m_{i-1}(\delta)\|_X + \gamma(\beta\delta)^{2i}.
\]
By by repeatedly applying the estimate, we arrive at
\begin{eqnarray*}
\|m_{i+1}(\delta) - m_i(\delta)\|_X&\leq (\beta\delta)^{i+1}\|m_0\|_X + \gamma\sum_{j=0}^i(\beta\delta)^j\big((\beta\delta)^2\big)^{i-j}\\
&= (\beta\delta)^{i+1}\|m_0\|_X + \gamma\frac{\big((\beta\delta)^2\big)^{i+1} - (\beta\delta)^{i+1}}{(\beta\delta)^2 - \beta\delta},
\end{eqnarray*}
From this bound, it follows that $\{m_\ell(\delta)\}_{\ell \geq 1}$ is a Cauchy sequence: For $k > \ell$, we have
\begin{eqnarray*}
\|m_k(\delta) - m_\ell(\delta)\|_X &\leq \sum_{i=\ell}^{k-1} \|m_{i+1}(\delta) - m_i(\delta)\|_X\\
&\leq \|m_0\|_X\sum_{i=\ell}^{k-1}(\beta\delta)^{i+1} + \gamma\sum_{i=\ell}^{k-1}\frac{\big((\beta\delta)^2\big)^{i+1} - (\beta\delta)^{i+1}}{(\beta\delta)^2 - \beta\delta}\\
&= \|m_0\|_X\frac{(\beta\delta)^{\ell+1} - (\beta\delta)^{k+1}}{1-\beta\delta} + \gamma\frac{(\beta\delta)^{\ell+1}-(\beta\delta)^{k+1}}{(\beta\delta - 1)((\beta\delta)^2 - \beta\delta)}\\
&\hspace{1cm} + \gamma\frac{\big((\beta\delta)^2\big)^{\ell+1} - \big((\beta\delta)^2\big)^{k+1}}{(\delta^2 - 1)(\delta^2 - \beta\delta)}
\end{eqnarray*}
which tends to zero as $k,\ell\rightarrow\infty$, provided $\delta$ is small enough. Thus the sequence $\{m_\ell(\delta)\}$ converges, and we denote the limit by $m(\delta)$. Taking the limit as $k\rightarrow\infty$ in the above inequality, we have
\begin{eqnarray*}
\fl\|m_\ell(\delta) - m(\delta)\|_X &\leq \frac{(\beta\delta)^{\ell+1}}{\beta\delta - 1} + \gamma\frac{(\beta\delta)^{\ell+1}}{(\beta\delta - 1)((\beta\delta)^2 - \beta\delta)} + \gamma\frac{\big((\beta\delta)^2\big)^{\ell+1}}{(\delta^2 - 1)((\beta\delta)^2 - \beta\delta)}\\
&= \mathcal{O}\big((\beta\delta)^\ell\big)
\end{eqnarray*}
for all $\ell \geq 1$, and it follows that $\{m_\ell(\delta)\}_{\ell\geq 1}$ converges geometrically with rate $\beta\delta$.

To show that the limit $m(\delta)$ is indeed a fixed point of $G(\cdot,\delta)$, we first we note that
\begin{eqnarray*}
\fl\|G_\ell(m_\ell(\delta),\delta) - G(m(\delta),\delta)\|_X &= \|\mH_\ell(\delta)m_\ell(\delta) + g_\ell(\delta) - \mH(\delta)m(\delta) - g(\delta)\|_X\\
&\leq \|\mH_\ell(\delta) - \mH(\delta)\|_{\mathcal{L}(X)}\|m_\ell(\delta)\|_X \\
&\hspace{1cm} + \|\mH(\delta)\|_{\mathcal{L}(X)}\|m_\ell(\delta) - m(\delta)\|_X + \|g_\ell(\delta) - g(\delta)\|_X\\
&\rightarrow 0,
\end{eqnarray*}
as $\ell\rightarrow \infty$, and so it follows that
\[
m(\delta) = \lim_{\ell\rightarrow\infty} m_{\ell+1}(\delta)
= \lim_{\ell\rightarrow\infty} G_\ell(m_\ell(\delta),\delta)
= G(m(\delta),\delta).
\]
All that remains is to check that the fixed point is unique. Supposing that $h(\delta)$ is another fixed point,  it follows that
\begin{eqnarray*}
\|m(\delta) - h(\delta)\|_X &= \|G(m(\delta),\delta) - G(h(\delta),\delta)\|_X\\
&\leq \|\mH(\delta)\|_{\mathcal{L}(X)}\|m(\delta) - h(\delta)\|_X\\
&\leq \beta\delta\|m(\delta) - h(\delta)\|_X.
\end{eqnarray*}
Hence, for $\delta < 1/\beta$ we must have that $m(\delta) = h(\delta)$, and the result follows.
\end{proof}

\subsection{Convergence in Finite Dimensions}
\label{ssec:conv_finite}
Above we had to assume that the prior distribution was sufficiently close to the posterior in order to guarantee convergence; in finite dimensions we may drop this assumption and still get convergence of the covariances. Here we assume that $X$ and $Y$ are Euclidian spaces,  $X = \mathbb{R}^N$ and $Y = \mathbb{R}^J$. 

In the following, we use the following convention: For symmetric matrices 
$\mB_1,\mB_2$, we write $\mB_1 \geq \mB_2$, or $\mB_1 > \mB_2$ to indicate that $\mB_1 - \mB_2$ is non-negative definite, or positive definite, respectively. We denote by $S_+^N \subseteq \mathbb{R}^{N\times N}$  the set of positive definite $N\times N$ matrices. In this section, the adjoint operators are denoted as transposes.

We start by showing that the iterative updating formula (\ref{eq:ck_standard}) gives a convergent sequence. However, instead of the covariance matrices, it is more convenient to work with the precision matrices, defined as $\mB_\ell := \mC_\ell^{-1}$. Observe that formula (\ref{eq:ck_standard}) can be written as
\[
\mC_{\ell+1}^{-1} = \mA^\mT(\mGamma + \mM \mC_\ell \mM^\mT)^{-1}\mA + \mC_0^{-1},
\]
which motivates the following result.

\begin{proposition}
Let $\mB_0 = \mC_0^{-1} \in S^N_+$ be a positive definite precision matrix,  and let the sequence $\{\mB_{\ell}\}_{\ell \geq 0}$ be generated iteratively by the formula $\mB_{\ell+1} = R(\mB_\ell)$, where $R:S^N_+\rightarrow S^N_+$ is given by
\[
R(\mB) = \mA^\mT(\mGamma + \mM \mB^{-1} \mM^\mT)^{-1}\mA + \mB_0.
\]
Then the sequence $\{\mB_\ell\}_{\ell\geq 0}$ is increasing in the sense of quadratic forms, and there exists a positive definite $\mB \in S^N_{+}$ such that $\mB_\ell\uparrow \mB$.

Consequently, the sequence of covariances $\{\mC_\ell\}_{\ell\geq 0}$ defined by (\ref{eq:ck_standard}) satisfies $\mC_\ell\downarrow \mC := \mB^{-1}$.
\end{proposition}

\begin{proof}
We first show that $\mB_{\ell+1} \geq \mB_\ell$ for all $\ell$ using induction. Write
\[
r(\mB) = \mA^\mT(\mGamma + \mM \mB^{-1} \mM^\mT)^{-1}\mA,
\]
so that $R(\mB) = r(\mB) + \mB_0$. If $\mB > 0$, then $r(\mB) \geq 0$, proving that $\mB_1 = r(\mB_0) + \mB_0 \geq \mB_0$.

Now assume that $\mB_{\ell} \geq \mB_{\ell-1}$. Then we have
\begin{eqnarray*}
\mB_{\ell+1} - \mB_\ell &= R(\mB_\ell) - R(\mB_{\ell-1})\\
&= r(\mB_\ell) - r(\mB_{\ell-1})\\
&= \mA^\mT\Big(\big(\mGamma + \mM \mB_\ell^{-1} \mM^\mT\big)^{-1} - \big(\mGamma + \mM \mB_{\ell-1}^{-1}\mM^\mT\big)^{-1}\Big)\mA.
\end{eqnarray*}
To prove the claim, it suffices to show that the bracketed difference is non-negative definite.  Consider the difference
\[
 \big(\mGamma + \mM\mB_{\ell-1}^{-1}\mM^\mT\big) - \big(\mGamma + \mM\mB_\ell^{-1}\mM^\mT\big) 
 = \mM\big(\mB_{\ell-1}^{-1} - \mB_\ell^{-1}\big)\mM^\mT \geq 0
\]
by the induction assumption. Therefore,
\[
 \big(\mGamma + \mM\mB_{\ell-1}^{-1}\mM^\mT\big)^{-1} \leq  \big(\mGamma + \mM\mB_\ell^{-1}\mM^\mT\big)^{-1},
\]
which implies the desired non-negative definiteness.

To prove that the sequence $\{\mB_\ell\}_{\ell\geq 0}$ is bounded as quadratic forms, denote the quadratic form as
\[
 \mQ_\ell(u,v) = u^\mT \mB_\ell v,\quad u,v\in\R^N.
\]
Since
\[
 \big(\mGamma + \mM\mB_{\ell-1}^{-1}\mM^\mT\big)^{-1} \leq \mGamma^{-1},
\]
we have 
\[
 \mQ_\ell(u,u) \leq (\mA u)^\mT \mGamma^{-1} \mA u + u^\mT \mB_0 u,
\]
proving the boundedness of the sequence.

In particular, it follows that for each $u\in\R^N$, 
\[
 \mQ_\ell(u,u)\rightarrow u^\mT \mB u, \quad\mbox{as $\ell\to\infty$}
\]
for some symmetric positive definite matrix $\mB\in \R^{N\times N}$. The fact that the matrix entries of $\mB_\ell$ converge to the corresponding entries of $\mB$ follows from the polar identity,
\[
 u^\mT\mB_\ell v = \frac 14\left(\mQ_\ell(u+v) - \mQ_\ell(u-v)\right),
\]
with $u,v$ being the canonical basis vectors.  This completes the proof.
\end{proof}

\section{The General Case}
\label{sec:sampling}

In general, the sequence of distributions $\pi_\ell$  are not Gaussian, and so it 
is considerably harder to analyze the convergence as we did in the previous section. In this section we consider how the algorithm may be implemented in practice, and in particular, how to produce an approximation to the posterior distribution using a finite number of full model evaluations. This approximate distribution can be used for generating samples using only approximate model evaluations,
leading to a significantly lower computational cost over sampling using the true posterior based on the full model. 

In subsection \ref{ssec:update_densities} we outline the general framework for sampling from the approximate posterior sequence and updating the density, making use of particle approximations. In subsection \ref{ssec:iteration_op} we reformulate the iteration (\ref{eq:update_leb}) in terms of operators on the set of probability measures, and provide results on properties of these operators. Convergence in the large particle limit is shown, using the new formulation of the update. In subsection \ref{ssec:gauss_sampling} a particular rejection sampling method, based on a Gaussian mixture proposal, 
is studied. Importance sampling is then considered in subsection \ref{ssec:importance} and similar convergence is shown. 

\subsection{Updating the Densities via Sampling}
\label{ssec:update_densities}

We consider the algorithm in Section~\ref{ssec:algorithm}, and in particular, address the question of how to generate a sequence of approximate samples from the 
iteratively defined densities  $\mu_\ell$ given by (\ref{eq:update_leb}). We shall use particle approximations to do this. Assume that
\[
 {\mathscr S}_\ell = \big\{(u_\ell^1,w_\ell^1),(u_\ell^2,w_\ell^2),\ldots,(u_\ell^N,w_\ell^N)\big\},\quad \ell = 0,1,2,\ldots
\]
is the current approximate sample of the unknowns with relative weights $w_\ell^j$. For $\ell=0$, the sample is obtained by independent sampling from the prior, and $w_0^j = 1/N$. We then compute the modeling error sample,
\[
 {\mathscr M}_\ell = \big\{m_\ell^1,m_\ell^2,\ldots,m_\ell^N\big\},\quad \ell = 0,1,2,\ldots
\]
by defining
\[
 m_\ell^j = M(u_\ell^j).
\]
Consider now the model (\ref{appr model}). Assuming that the modeling error is independent of the unknown $u$, we may write a conditional likelihood model,
\[
 \pi(b\mid u,m) \propto \pi_{\rm noise}(b - f(u) - m).
\]
Let $\nu_\ell(m)$ denote the probability density of the modeling error based on our current information. Then, the updated likelihood model based on the approximate model  is
\[
 \pi_{\ell+1}(b\mid u) = \int \pi(b\mid u,m)\nu_\ell(m) dm,
\]
and, using a Monte Carlo integral approximation, postulating that the realizations $m_\ell^j$ inherit the weights of the sample points $u_\ell^j$, we obtain
\[
 \pi_{\ell+1}(b\mid u) \approx \sum_{j=1}^Nw_\ell^j \pi(b\mid u,m_\ell^j).
\]
The current approximation for the posterior density is
\[
 \pi_{\ell+1}(u\mid b) \propto \pi_{\rm prior}(u)\sum_{j=1}^Nw_\ell^j \pi(b\mid u,m_\ell^j),
\]
suggesting an updating scheme for ${\mathscr S}_\ell\rightarrow {\mathscr S}_{\ell+1}$: 
\begin{itemize}
\item[(a)] Draw  an index $k_j$ by replacement from $\{1,2,\ldots,N\}$, using the probabilities $w_\ell^j$;
\item[(b)] Draw the sample point
\begin{eqnarray}
\label{eq:update_draw}
 u_{\ell+1}^j\sim \pi_{\rm prior}(u)\pi(b\mid u,m_\ell^{k_j}).
\end{eqnarray}
\end{itemize}
Part (b) above is straightforward, in particular, if the model is Gaussian and $f$ is linear, such as in the linearized model for EIT, since the measure (\ref{eq:update_draw}) is then a Gaussian. We will demonstrate the effectiveness of this approach in Section~\ref{ssec:eit}. Otherwise we may consider other sampling methods such as importance sampling; this is what is done in the following subsections.

\subsection{A convergence result for particle approximations}
\label{ssec:iteration_op}
In this section, we rewrite the updating formula in terms of mappings of measures, and analyze the convergence of the particle approximation under certain limited conditions.

Let $\mu_\ell$ denote the current approximation of the posterior density for $u$. The updated likelihood based on the modeling error is
\[
 \pi_{\ell+1}( b\mid u) \propto \int_X \pi_{\rm noise}(b - f(u) - M(z))\mu_{\ell}(dz),
\]
and therefore, the updating, by Bayes' formula, is given by
\begin{equation}\label{def of P}
 \mu_{\ell+1}(du) \propto \mu_{\rm prior}(du) \int_X \pi_{\rm noise}(b - f(u) - M(z))\mu_{\ell}(dz) = P\mu_{\ell}(du).
\end{equation}
Furthermore, we write the normalization formally as an operator,
\[
 L\mu = \frac{\mu}{\mu(1)},  \quad \mu(1) = \int_X\mu(du).  
 \]
The model updating algorithm can therefore be written concisely as
\[
 \mu_{\ell+1} = LP\mu_\ell,\quad \mu_0 = \mu_{\rm prior}.
\]    

Let $\mathsf{M}(X)$ denote the set of finite measures on $X$.  Denote by $\mathsf{P}(X)$ the set of probability measures on $X$, and for $p \in (0,1)$ denote by
$\mathsf{M}_p(X)$ the set of finite measures with total mass lying in the interval $[p,p^{-1}]$.

Let $\mu$ and $\nu$ denote two random $\mathsf{M}(X)$-valued measures, i.e., $\mu^\omega,\nu^\omega\in \mathsf{M}(X)$ for $\omega\in\Omega$, where $\Omega$ is a probability space. Denoting by ${\mathbb E}$ the expectation, we define
the distance between random measures through
\[
 d(\mu,\nu)^2 = \sup_{\|\varphi\|_\infty=1} {\mathbb E} |\mu(\varphi)-\nu(\varphi)|^2,
\]
where the functions $\varphi$ are continuous over $X$. For non-random measures, the definition coincides with the total variation distance.

In the following two lemmas, which we need for the large particle convergence
result that follows them, 
we make a restrictive assumption about the noise distribution. 
\begin{assumptions}
\label{ass:particle}
There exists $\kappa \in (0,1)$ such that for all $\eps \in Y$, $\kappa \leq \pi_{\mathrm{noise}}(\eps) \leq \kappa^{-1}$.
\end{assumptions}

Under this assumption, we show the following results concerning the mappings $P$ and $L$ defined before.
\begin{lemma}
Let Assumptions \ref{ass:particle} hold. Then  $P:\mathsf{P}(X)\rightarrow\mathsf{M}_\kappa(X)$, and
\[
d(P\mu,P\nu) \leq \kappa^{-1} d(\mu,\nu).
\]
\end{lemma}
\begin{proof}
First note that $\kappa \leq \pi_{\mathrm{noise}} <\kappa^{-1}$ implies that $\kappa \leq (P\mu)(1) \leq \kappa^{-1}$, and so $P$ does indeed map ${\mathsf P}(X)$  into $\mathsf{M}_\kappa(X)$. Exchanging the order of integration, we see for any bounded measurable $\varphi$,
\[
(P\mu)(\varphi) = \int_X\underbrace{\left(\int_X \pi_{\mathrm{prior}}(u)\pi_{\mathrm{noise}}(b-f(u) - M(z))\varphi(u)\dee u\right)}_{=:\psi(z)}\mu(\dee z)
\]
and so
\[
|(P\mu)(\varphi)-(P\nu)(\varphi)|^2 = |\mu(\psi) - \nu(\psi)|^2.
\]
Using that $\pi_{\mathrm{noise}} <\kappa^{-1}$, we see that $\|\varphi\|_\infty \leq 1$ implies that $\|\psi\|_\infty \leq \kappa^{-1}$, and so
\begin{eqnarray*}
d(P\mu,P\nu)^2 &\leq \sup_{\|\psi\|_\infty \leq \kappa^{-1}} {\mathbb E}|\mu(\psi) - \nu(\psi)|^2\\
&\leq \sup_{\|\psi\|_\infty \leq 1} \kappa^{-2}{\mathbb E} |\mu(\psi) - \nu(\psi)|^2\\
&= \kappa^{-2} d(\mu,\nu)^2,
\end{eqnarray*}
implying the claim.
\end{proof}

A similar result for the mapping $L$ can be obtained. 

\begin{lemma}
Let Assumptions \ref{ass:particle} hold. Then it follows that $L:\mathsf{M}_\kappa(X)\rightarrow\mathsf{P}(X)$, and furthermore,
for $\mu,\nu \in \mathsf{M}_\kappa(X)$, we have
\[
d(L\mu,L\nu) \leq 2\kappa^{-2}d(\mu,\nu).
\]
\end{lemma}
\begin{proof}
The proof is essentially identical to that of Lemma 5.17 in \cite{lecturenotes}, with $1$ in place of $g$. We skip the details here.
\end{proof}

We use the above results to analyze the convergence of particle approximations of the measures. We  introduce the sampling operator $S^N:\mathsf{P}(X)\rightarrow\mathsf{P}(X)$,
\[
S^N\mu = \frac{1}{N}\sum_{j=1}^N \delta_{u^j},\quad u^1,\ldots,u^N \sim \mu\;\mathrm{i.i.d.}
\]
and we have
\[
(S^N\mu)(\varphi) = \frac{1}{N}\sum_{j=1}^N \varphi(u^j),\quad u^1,\ldots,u^N \sim \mu\;\mathrm{i.i.d.}
\]
Observe that $S^N\mu$ is a random measure, as it depends on the sample.
It is shown in \cite{lecturenotes}, Lemma 5.15, that the operator $S^N$ satisfies
\[
\sup_{\mu \in \mathsf{P}(X)} d(S^N\mu,\mu) \leq \frac{1}{\sqrt{N}}.
\]

Define the sequence of particle approximations $\{\mu_\ell^N\}_{\ell\geq 0}$ to $\{\mu_\ell\}_{\ell\geq 0}$ by
\begin{eqnarray}
\label{eq:iteration_exact}
\nonumber\mu_0^N &= S^N \mu_0,\\
\mu_{\ell+1}^N &= S^NLP\mu_\ell^N.
\end{eqnarray}

in light of the previous lemmas, now we prove the following result regarding convergence of this approximation as $N\rightarrow\infty$:

\begin{proposition}
\label{prop:exact_sampling}
Let Assumptions \ref{ass:particle} hold. Define $\{\mu_\ell\}_{\ell\geq 0}$, $\{\mu_\ell^N\}_{\ell\geq 0}$ as above. Then, for each $\ell$,
\[
d(\mu_\ell^N,\mu_\ell) \leq \frac{1}{\sqrt{N}}\sum_{k=0}^\ell (2\kappa^{-3})^k.
\]
In particular, $d(\mu_\ell^N,\mu_\ell)\rightarrow 0$ as $N\rightarrow\infty$.
\end{proposition}

\begin{proof}
The triangle inequality for $d(\cdot,\cdot)$ yields
\begin{eqnarray*}
e_\ell := d(\mu_\ell^N,\mu_\ell) \leq d(S^NLP\mu_{\ell-1}^N,LP\mu_{\ell-1}^N) + d(LP\mu_{\ell-1}^N,LP\mu_{\ell-1}),
\end{eqnarray*}
and applying the bounds given by the previous lemmas, we obtain
\begin{eqnarray*}
e_\ell &\leq \frac{1}{\sqrt{N}} + 2\kappa^{-2}d(P\mu_{\ell-1}^N,P\mu_{\ell-1})\\
&\leq \frac{1}{\sqrt{N}} + 2\kappa^{-3} d(\mu_{\ell-1}^N,\mu_{\ell-1})\\
&= \frac{1}{\sqrt{N}} + 2\kappa^{-3} e_{\ell-1}.
\end{eqnarray*}
The result follows since $e_0 = d(S^N\mu_0,\mu_0) \leq 1/\sqrt{N}$.
\end{proof}

\subsection{Particle approximation with Gaussian densities}
\label{ssec:gauss_sampling}

In this section, we consider the particle approximation when the approximate model is linear, while the accurate model need not be. This is the situation in the computed examples that will be discussed later.

Suppose that the approximate model $f$ is linear, $f(u) = \mA u$, and the noise and prior distributions are Gaussian,
\[
\pi_{\mathrm{noise}} = \mathcal{N}(0,\mGamma),\quad \pi_{\mathrm{prior}} = \mathcal{N}(m_0,\mC_0).
\]
Then the measure (\ref{eq:update_draw}) is a Gaussian mixture:
\[
\pi_{\mathrm{prior}}(u)\pi(b\mid u,m_\ell^j) = \frac{1}{N}\sum_{j=1}^N \mathcal{N}(u\mid p_\ell^j,\mC)
\]
where the means and covariance are given by
\begin{eqnarray*}
\mC &= (\mA^\mT\mGamma^{-1}\mA + \mC_0)^{-1},\\
p_\ell^j &= \mC(\mA^\mT \mGamma^{-1}(b-m_\ell^j) + \mC_0^{-1}m_0).
\end{eqnarray*}
The collection of samples $\mathscr{S}_\ell$ can then be evolved via the following algorithm.

{\bf Algorithm (Linear Approximate Model):}
\begin{enumerate}[\hspace{0.2cm}1.]
\item Set $\ell = 0$. Define the covariance operator $\mC = (\mA^\mT \mGamma^{-1}\mA + \mC_0)^{-1}$. Draw an initial ensemble of particles $\{u_\ell^j\}_{j=1}^N$ from the prior measure $\mu_0(\dee u) = \pi_{\mathrm{prior}}(u)\dee u$, and define the collection $\mathscr{S}_\ell = \{u_\ell^1,u_\ell^2,\ldots,u_\ell^N\}$. 

\item Define the means $p_\ell^j = \mC(\mA^\mT \mGamma^{-1}(b-M(u_\ell^j)) + \mC_0^{-1}m_0)$, $j=1,\ldots,N$.
\item For each $j=1,\ldots,N$
\begin{enumerate}[(i)]
\item Sample $k_j$ uniformly from the set $\{1,\ldots,N\}$
\item Sample $u_{\ell+1}^j \sim N(p_\ell^{k_j},\mC)$
\end{enumerate}

\item Set $\mathscr{S}_{\ell+1} = \{u_{\ell+1}^1,u_{\ell+1}^2,\ldots,u_{\ell+1}^N\}$.
\item Set $\ell \mapsto \ell+1$ and go to 2.
\end{enumerate}

\begin{remark}
For more general models, one could use a method such as rejection sampling in order to produce exact samples from the measure (\ref{eq:update_draw}). A suitable proposal distribution for this rejection sampling could be, for example, a Gaussian mixture with appropriately chosen means and covariances \cite{CSbook}.

Two natural candidates for non-Gaussian priors, that retain some of the simplicity of the Gaussian models without being as limited, are:
\begin{enumerate}
\item Hierarchical, conditionally Gaussian prior models,
\[
 \pi_{\rm prior}(u\mid\theta) \sim {\mathcal N}(\mu_\theta,\mC_\theta),
\]
where the mean and covariance depend on a hyperparameter vector $\theta$ that follows a hyperprior distribution,
\[
 \theta\sim\pi_{\rm hyper}.
\]
The hypermodels allow the introduction of sparsity promoting priors, similar to total variation; \cite{hypermodels,MEG}.
\item Gaussian mixtures, which allow a fast sampling from non-Gaussian distributions through a local approximation by Gaussian or other simple distributions \cite{west}.
\end{enumerate}
\end{remark}

\begin{remark}
We point out here that the approximation result of Proposition~\ref{prop:exact_sampling} does not apply to the Gaussian likelihood, as the noise density is not bounded from below.
\end{remark}

\subsection{Importance Sampling and Convergence}
\label{ssec:importance}

In this section we consider an approximate sampling based updating scheme of the probability densities using importance sampling.
This method effectively turns a collection of prior samples into samples from the posterior by weighting them appropriately, using the fact that the posterior is absolutely continuous with respect to the prior.

Assume that at stage $\ell$ of the approximation scheme, we have a collection of $N$ particles and the corresponding weights, 
${\mathscr S}_\ell = \{(u_\ell^j,w_\ell^j)\}_{j=1}^N$. The associated particle approximation $\mu_\ell^N$ of the probability distribution acting on a test function $\varphi$ is
\[
\mu_\ell^N(\varphi) = \sum_{j=1}^N w_\ell^j\varphi(u_\ell^j).
\]
We evolve this distribution by acting on it with $P$ and $L$. By the definition (\ref{def of P}) of $P$, we first get an approximation
\begin{eqnarray*}
P\mu_\ell^N(du) & = \left(\sum_{j=1}^N w_\ell^j\pi_{\mathrm{noise}}(b-f(u)-M(u_\ell^j))\right)\mu_{\rm prior}(\dee u)  \\
&=: g_\ell(u)\mu_{\rm prior}(\dee u).
\end{eqnarray*}
To generate an updated sample based on this approximation, we 
use independent sampling to draw $u_{\ell+1}^j \sim \mu_0$, $j=1,\ldots,N$, and define the particle approximation by
\[
\mu_{\ell+1}^N(\varphi) = \sum_{j=1}^N w_{\ell+1}^j\varphi(u_{\ell+1}^j),\;\;\;w_{\ell+1}^j = \frac{g_\ell(u_{\ell+1}^j)}{\sum_{j=1}^N g_\ell(u_{\ell+1}^j)}.
\] 
Denoting by $T^N:\mathsf{P}(X)\rightarrow\mathsf{P}(X)$ the importance sampling step, consisting of independent sampling and weighting, we may define an iterative algorithm symbolically as 
\begin{eqnarray}
\label{eq:iteration_importance}
\nonumber\mu_0^N &= T^N \mu_0,\\
\mu_{\ell+1}^N &= T^NLP\mu_\ell^N.
\end{eqnarray}

Explicitly, the algorithm can be described as follows.

{\bf Algorithm (Importance sampling)}
\begin{enumerate}[\hspace{0.2cm}1.]
\item Set $\ell = 0$. Draw an initial ensemble of particles $\{u_\ell^j\}_{j=1}^N$ from the prior measure $\mu_0(\dee u) = \pi_{\mathrm{prior}}(u)\dee u$, and initialize the weights $w_\ell^j = 1/N$ for each $j=1,\ldots,N$. Define the collection $\mathscr{S}_\ell = \{(u_\ell^1,w_\ell^1),(u_\ell^2,w_\ell^2),\ldots,(u_\ell^N,w_\ell^N)\}$.
\item Define 
\[
g_\ell(u) = \sum_{j=1}^N w_\ell^j\pi_{\mathrm{noise}}(b-f(u)-M(u_\ell^j)).
\]
\item Sample $u_{\ell+1}^j \sim \mu_0$, $j=1,\ldots,N$ i.i.d. and define the weights
\[
w_{\ell+1}^j = \frac{g_\ell(u_{\ell+1}^j)}{\sum_{j=1}^N g_\ell(u_{\ell+1}^j)},\quad j=1,\ldots,N.
\]
\item Set $\mathscr{S}_{\ell+1} = \{(u_{\ell+1}^1,w_{\ell+1}^1),(u_{\ell+1}^2,w_{\ell+1}^2),\ldots,(u_{\ell+1}^N,w_{\ell+1}^N)\}$.
\item Set $\ell \mapsto \ell+1$ and go to 2.
\end{enumerate}

As in the previous section, we establish a convergence result for $N\to\infty$ only under the restrictive condition of Assumption~\ref{ass:particle}.
We recall the following result from \cite{intrinsicdim}:
\begin{lemma}
Let $\mu \in \mathsf{P}(X)$ be absolutely continuous with respect to the prior measure $\mu_0$,
\[
\mu(\dee u) \propto g(u)\mu_0(\dee u),
\]
where $\mu_0(g^2) < \infty$. Define the quantity $\rho \geq 1$ by $\rho = \mu_0(g^2)/\mu_0(g)^2$. Then
\[
d(T^N\mu,\mu) \leq 2\sqrt{\frac{\rho}{N}}.
\]
\end{lemma}

By Assumption \ref{ass:particle}, there exists $\kappa \in (0,1)$ such that $\kappa \leq g_\ell(\eps) \leq \kappa^{-1}$ for all $\eps \in Y$,  implying that
\[
\frac{\mu_0(g_\ell^2)}{\mu_0(g_\ell)^2} \leq \kappa^{-4}.
\]
In particular, by applying the above lemma to the measure $\hat{\mu}_{\ell+1} =LP\mu_{\ell}$, we see that
\[
d(T^N LP \mu_{\ell},LP{\mu}_{\ell}) \leq \frac{2\kappa^{-2}}{\sqrt{N}}.
\]
We are ready to prove the following proposition establishing the convergence of the particle approximations as $N\rightarrow\infty$:
\begin{proposition}
Let Assumptions \ref{ass:particle} hold for the noise distribution, and let $\{\mu_\ell\}_{\ell\geq 0}$ be the sequence of the model error approximations of the posterior, and  $\{\mu_\ell^N\}_{\ell\geq 0}$ a sequence of importance sampling approximations obtained as above. Then, for each $\ell$,
\[
d(\mu_\ell^N,\mu_\ell) \leq \frac{2\kappa^{-2}}{\sqrt{N}}\sum_{k=0}^{\ell-1} (2\kappa^{-3})^k + \frac{(2\kappa^{-3})^\ell}{\sqrt{N}}.
\]
In particular, $d(\mu_\ell^N,\mu_\ell)\rightarrow 0$ as $N\rightarrow\infty$.
\end{proposition}

\begin{proof}
From the triangle inequality for $d(\,\cdot\,,\,\cdot\,)$, we have
\begin{eqnarray*}
e_\ell := d(\mu_\ell^N,\mu_\ell) \leq d(T^NLP\mu_{\ell-1}^N,LP\mu_{\ell-1}^N) + d(LP\mu_{\ell-1}^N,LP\mu_{\ell-1}).
\end{eqnarray*}
The bounds derived above yield
\begin{eqnarray*}
e_\ell &\leq \frac{2\kappa^{-2}}{\sqrt{N}} + 2\kappa^{-2}d(P\mu_{\ell-1}^N,P\mu_{\ell-1})\\
&\leq \frac{2\kappa^{-2}}{\sqrt{N}} + 2\kappa^{-3} d(\mu_{\ell-1}^N,\mu_{\ell-1})\\
&= \frac{2\kappa^{-2}}{\sqrt{N}} + 2\kappa^{-3}e_{\ell-1}.
\end{eqnarray*}
Since we have $e_0 = d(T^N\mu_0,\mu_0) = d(S^N\mu_0,\mu_0) \leq 1/\sqrt{N}$, the result follows.
\end{proof}

\begin{remark}
In theory, the importance sampling method described above can be used with very weak assumptions on the forward maps and prior/noise distributions. However in practice it may be ineffective if the posterior is significantly far from the prior, such as when the size of the observational noise is small. To overcome this issue, one could instead consider Sequential Monte Carlo or Sequential Importance Sampling methods to evolve prior samples into posterior samples by introducing a sequence of intermediate measures \cite{BJMS15,maceachern1999sequential}.
\end{remark}

\section{Numerical Illustrations}
\label{sec:numerics}
In this section, we demonstrate the convergence properties established in the
preceding sections by means of computed examples. Furthermore, we demonstrate
the enhanced reconstructions obtained by modelling error as advocated
in this paper. The first example is a linear inverse source problem, elucidating the geometric convergence in the linear Gaussian case. The second example is the EIT problem with linearized approximate model with a coarse FEM mesh, allowing for straightforward particle updates. In the last example we consider the problem of recovering the permeability field in the steady state Darcy flow model, again with a linearized approximate model.

\subsection{Inverse Source Problem}
\label{ssec:num_inv}
As a proof of concept, we start by considering a simple one-dimensional inverse source problem.
Let $\Omega = (0,1)$ and define $X = L^2(\Omega)$. Given $u \in X$, let $p = P(u) \in H^1_0(\Omega)$  be the solution to the Laplace equation,
\[
\cases{\pushright{- p'' = u}&$x\in\Omega$\\
\pushright{p=0}&$x\in\partial\Omega$
}
\]
The inverse problem is to estimate the source $u$ from pointwise observations of $p$. Therefore, define the observation operator $\mathcal{O}:H^1_0(\Omega)\rightarrow\mathbb{R}^J$ by
\[
\mathcal{O}(u) = \big(u(q_1),\ldots,u(q_J)\big)
\]
for some set of points $\{q_1,\ldots,q_J\} \subseteq \Omega$. We define the exact forward operator $\bA^{\rm exact} = \mathcal{O}\circ P$. For numerical simulations, the exact forward model $\bA^{\rm exact}$ is approximated by a high fidelity proxy, $\bA$, obtained by approximating the solution $p$ through a finite difference solution on a fine mesh. The coarse mesh approximation of $\bA^{\rm exact}$, used in the inverse model, is denoted by $\mA = \mA_n$. In our computed example, we use $2^{10} - 1=1\,023$ equally spaced interior points for $\bA$, while the coarse mesh model $\mA_n$ is computed with of $2^n - 1$ equally spaced interior points, $n < 10$.

We let $q_j = j/16$, $j=1,\ldots,15=J$ be equally spaced observation points, and to generate the simulated data, we corrupt the high fidelity data with a small amount of white noise, $\eps \sim \mathcal{N}(0,\mGamma)$,  where we set $\mGamma = 10^{-8}\mI_J$. The prior is chosen to be a standard Brownian motion, specifically we take $\mu_0 = \mathcal{N}(0,\mC_0)$ with 
\begin{eqnarray*}
\mC_0 &= (-\Delta)^{-1},\quad \Delta = \frac{\dee^2}{\dee x^2},\\
\mathcal{D}(-\Delta) &= \bigg\{u \in H^2(\Omega)\,\bigg|\,u(0) = 0,\quad\frac{\dee u}{\dee x}(1) = 0\bigg\},
\end{eqnarray*}
and the true source used for data generation is drawn from the prior. Numerically the precision operator $\mC_0^{-1}$ is implemented as the finite difference Laplacian matrix. We perform $L = 30$ iterations in each simulation.

The posterior mean and covariance, $m_{\post}, \mC_{\post}$, corresponding to the high fidelity model, and the corresponding mean and covariance, $m_{\post}^n, \mC_{\post}^n$, based on the approximate model are given by
\begin{eqnarray*}
m_{\post} &= m_0 + \mC_0\bA^\mT(\mGamma + \bA\mC_0\bA^\mT)^{-1}(b-\bA m_0),\\
\mC_{\post} &= \mC_0 - \mC_0\bA^\mT(\mGamma + \bA\mC_0\bA^\mT)^{-1}\bA\mC_0,\\
m_{\post}^n &= m_0 + \mC_0\mA_n^\mT(\mGamma + \mA_n\mC_0\mA_n^\mT)^{-1}(b-\mA_n m_0),\\
\mC_{\post}^n &= \mC_0 - \mC_0\mA_n^\mT(\mGamma + \mA_n\mC_0\mA_n^\mT)^{-1}\mA_n\mC_0,
\end{eqnarray*}
respectively. The approximate posterior mean and covariances (\ref{eq:ck_standard}) obtained by the modeling error approach, after $\ell$ iterations, are denoted by $m_\ell$ and $\mC_\ell$, respectively.

Table \ref{tab:table_pde} shows the approximation errors arising from both approximations of $\mC_\post$ and $m_\post$ with different discretization levels. The table shows that the modeling error approach produces a better approximation of the posterior mean than the model ignoring the modeling error, while the approximate covariances are 
slightly less accurate as approximations of the posterior covariance than those
found without allowing for the modeling error correction. These experiments
confirm our assertion at the start of the paper, namely that allowing for
model error can result in improved point estimates (here the posterior mean)
but that the iteration introduced does not converge to the true posterior
distribution (as evidenced by the error in the covariance at fixed $n$ and large $L$.)

To demonstrate the convergence rate, Figure \ref{fig:errors_pde} shows the 
mean and covariance errors for various approximation levels as functions of the number of iterations. The plots, as well as the tabulated values, Table \ref{tab:table_pde_slope}, of the logarithmic slopes of the approximation errors verify the geometric convergence rates, with their dependence on the approximation level. Observe that the logarithm of the convergence rate for the covariance, a quadratic quantity, is twice that of the mean.

\begin{table}
\begin{center}
\caption{The approximation errors of the approximate posterior means and covariances for various approximation levels with and without the inclusion of the modeling error correction. The matrix norms are the Frobenius norms.}
\label{tab:table_pde}
\begin{tabular}{l|c|c|c|c}
$n$ & $\|m_L - m_\post\|$ & $\|m_\post^n - m_\post\|$ & $\|\mC_L - \mC_\post\|$ & $\|\mC_\post^n - \mC_\post\|$\\
 \hline
4 & 0.1906 & 0.2986 & 0.0676 & 0.0381\\
5 & 0.0455 & 0.0739 & 0.0170 & 0.0095\\
6 & 0.0111 & 0.0182 & 0.0043 & 0.0024 \\
7 & 0.0028 & 0.0045 & 0.0011 & 0.0006 \\
8 & 0.0007 & 0.0011 & 0.00026 & 0.00015 \\
9 & 0.0002 & 0.0003 & 0.00006 & 0.00004 \\
\end{tabular}
\end{center}
\end{table}

\begin{table}
\begin{center}
\caption{The convergence rates quantified in terms of the slope of the logarithmic plot of mean and covariance deviation from the limit values $m_L$ and $\mC_L$ before the deviations plateau, indicating that the algorithm has converged.}
\label{tab:table_pde_slope}
\begin{tabular}{l|c|c|c}
$n$ & Slope of $\log\|m_\ell - m_L\|$ & Slope of $\log\|\mC_\ell - \mC_L\|$ & $\|\mA_n-\mA_\star\|$\\
 \hline
4 & -1.80 & -3.60 & 0.101\\
5 & -3.13 & -6.26 & 0.0706\\
6 & -4.57 & -9.00 & 0.0491\\
7 & -6.30 & -11.9 & 0.0335\\
8 & -7.73 & -14.8 & 0.0219\\
9 & -9.34 & -14.9 & 0.0127
\end{tabular}
\end{center}
\end{table}

\begin{figure}
\centerline{
\includegraphics[width=5.2cm]{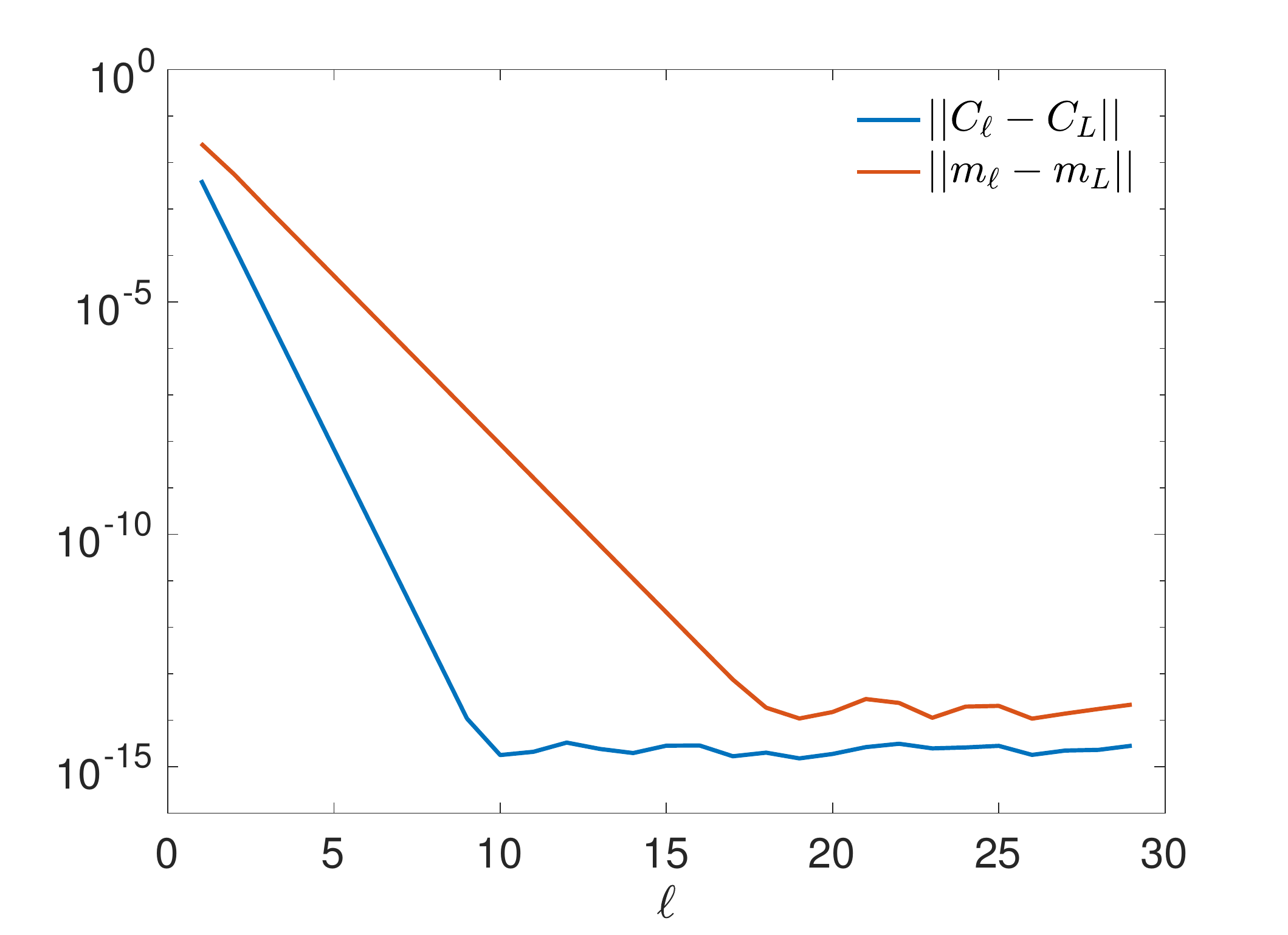}
\includegraphics[width=5.2cm]{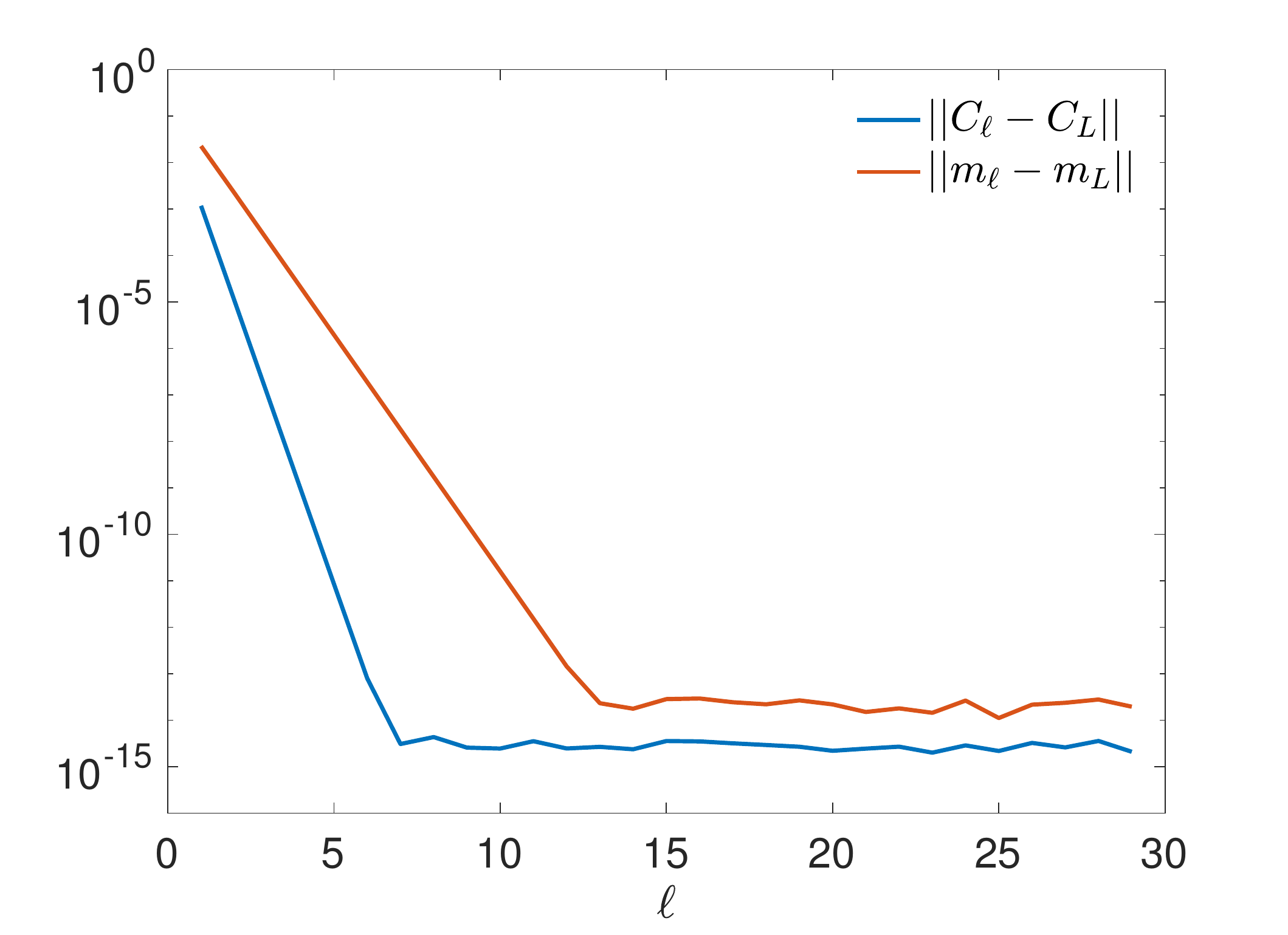}
\includegraphics[width=5.2cm]{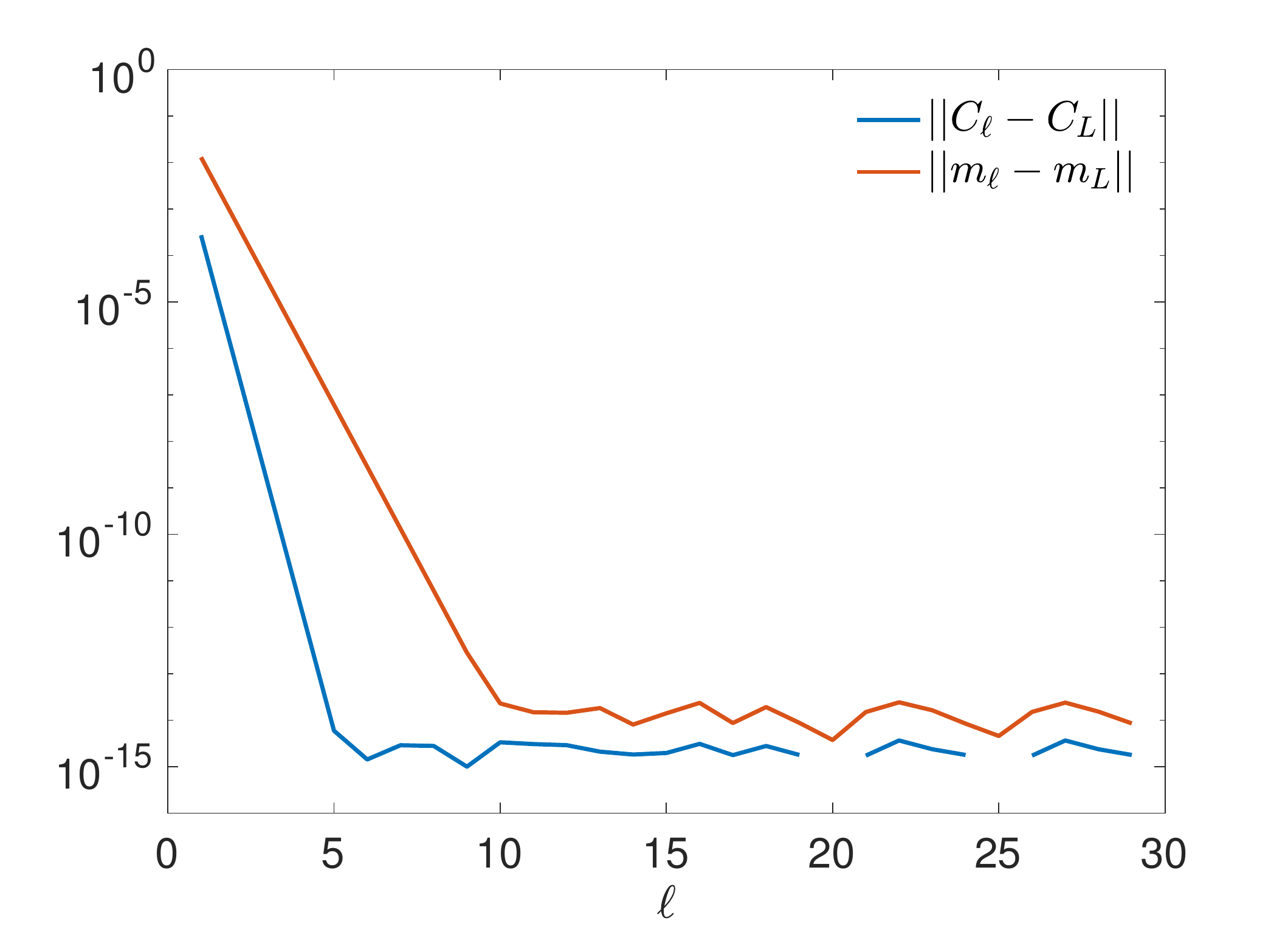}
}
\centerline{
\includegraphics[width=5.2cm]{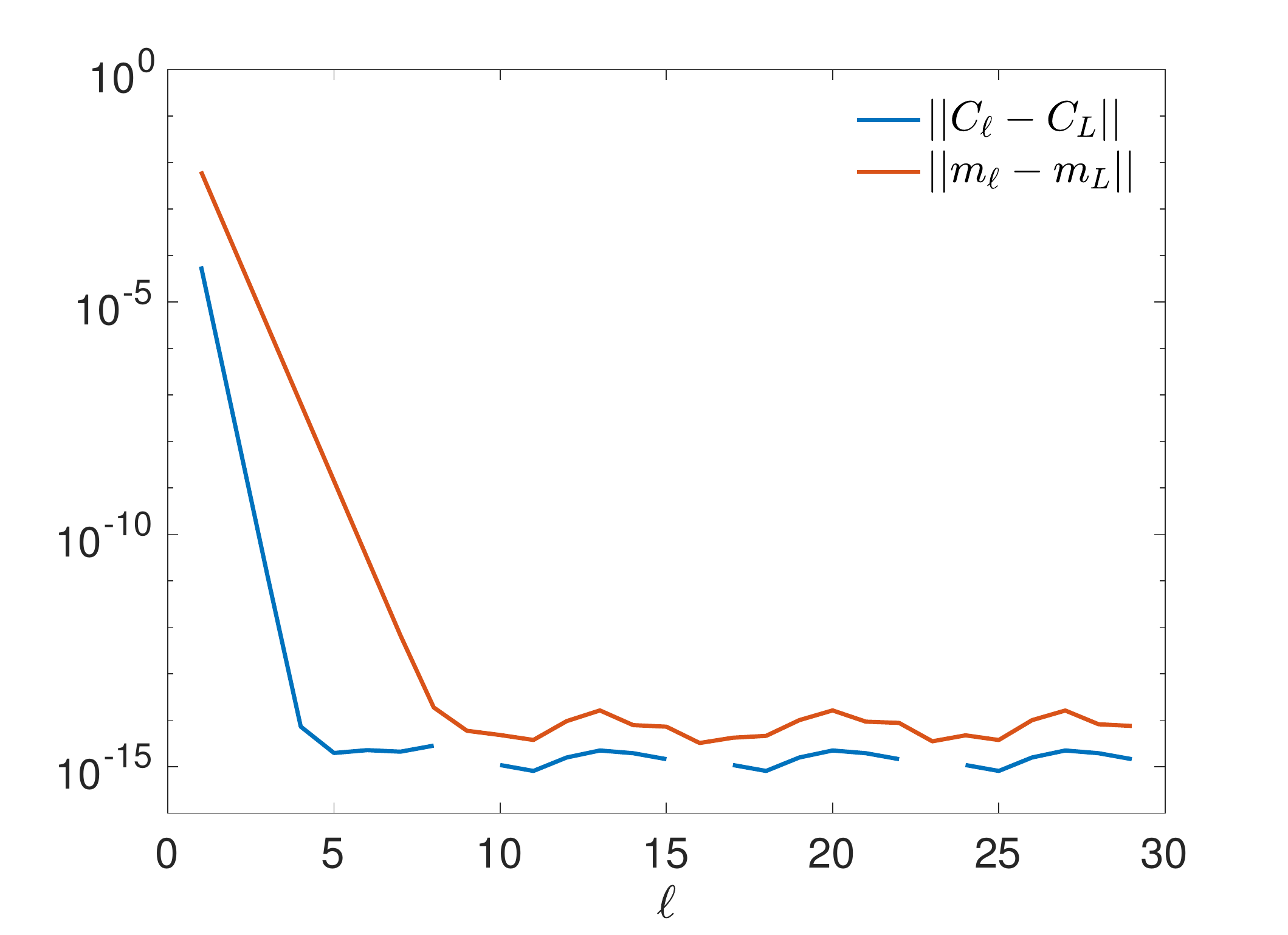}
\includegraphics[width=5.2cm]{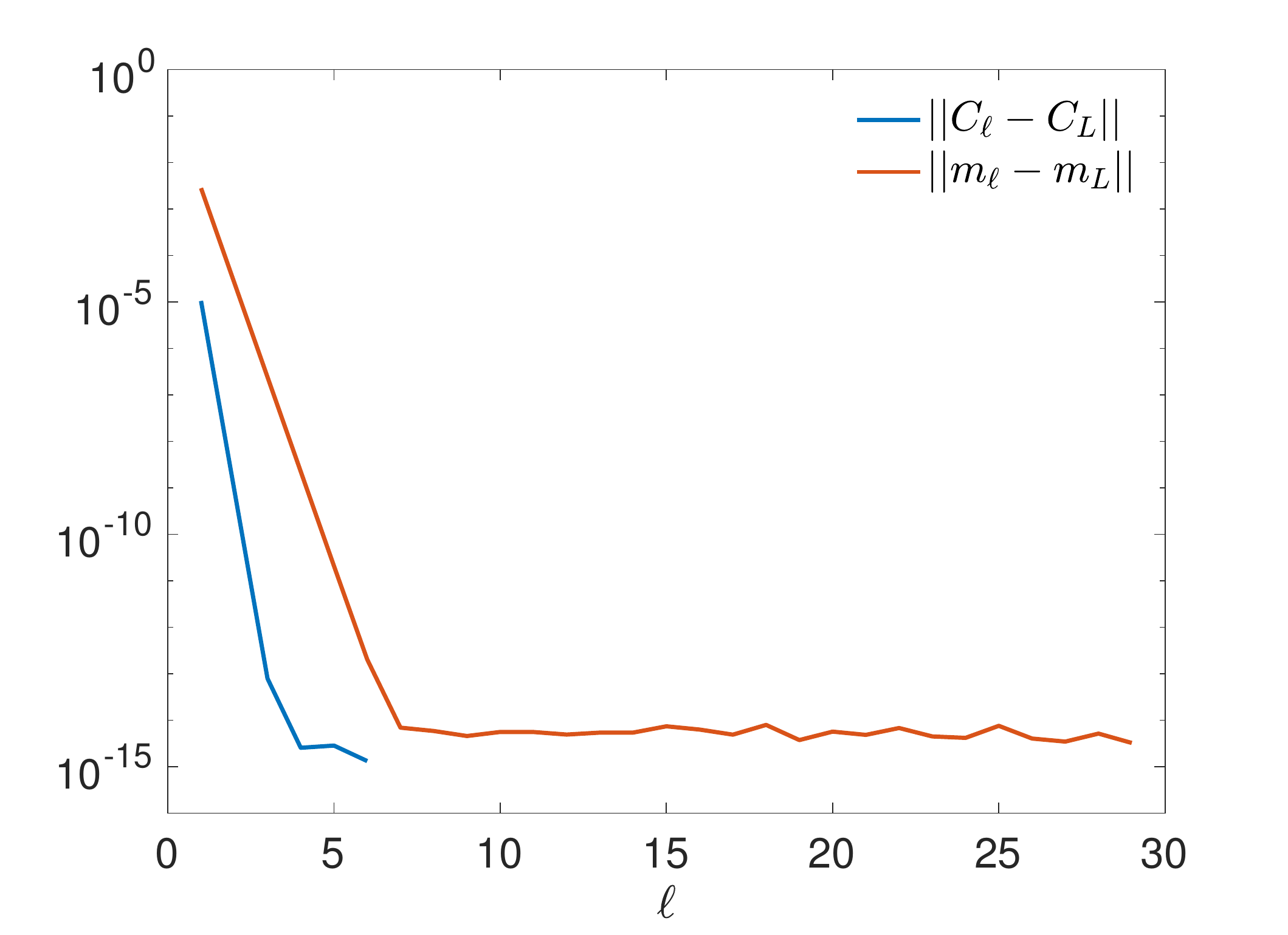}
\includegraphics[width=5.2cm]{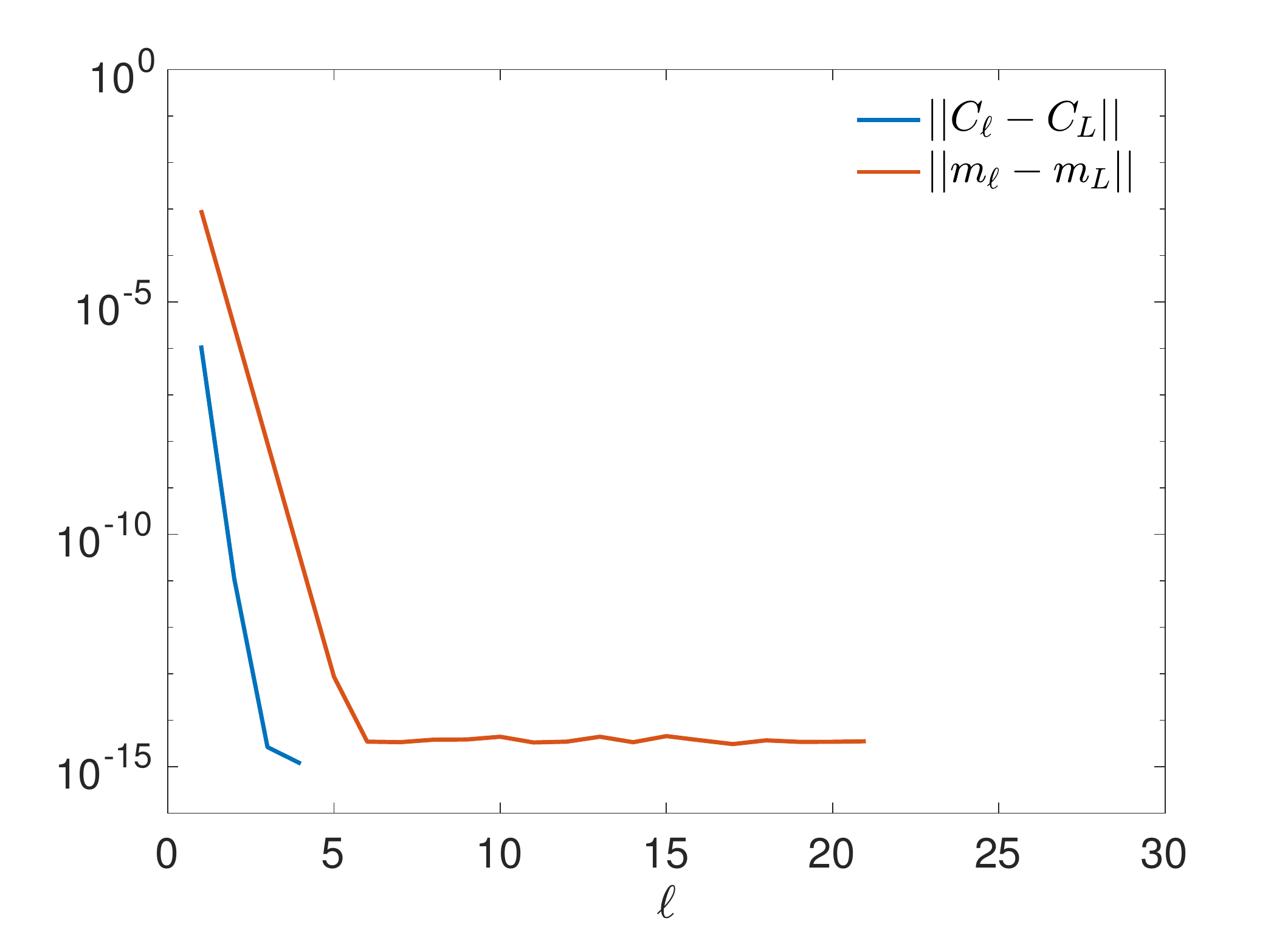}
}
\caption{The trace of the errors $\|m_\ell - m_L\|$ and $\|\mC_\ell - \mC_L\|$, illustrating their convergence rates. From left to right, top to bottom, the approximation level $n$ is increased from 4 to 9.}
\label{fig:errors_pde}
\end{figure}

\subsection{Electrical Impedance Tomography (EIT)}
\label{ssec:eit}

In this section, we revisit the modeling error due to coarse discretization of a PDE model in the context of Electrical Impedance Tomography (EIT).  Let $\Omega\subset \R^d$, $d=2,3$, denote a bounded connected set with boundary $\partial\Omega$, and let $\sigma:\Omega\to \R$ be a function modeling the electric conductivity in $\Omega$, $0<\sigma_m\leq\sigma\leq\sigma_M<\infty$. We assume that $S$ electrodes are attached to the boundary $\partial\Omega$, and we model them as open disjoint surface patches $e_s\subset\partial\Omega$, $1\leq s\leq S$.
Assuming that an electric current $I_s$ is injected through $e_s$ into the body modeled by $\Omega$, the electric voltage potential $v$ in $\Omega$, and the electrode voltages $V_s$ on the electrodes can be found as a solution of the Complete Electrode Model (CEM) boundary value problem \cite{SIC},
\begin{eqnarray*}
\cases{\pushright{\nabla\cdot\big(\sigma\nabla v\big) = 0} & $x \in \Omega$\\
\pushright{\sigma\frac{\partial v}{\partial n} = 0} & $x \in \partial\Omega\setminus\bigcup_{s=1}^S e_s$\\
\pushright{v+z_s \sigma\frac{\partial v}{\partial n} = V_s} & $x \in e_s$, $1\leq s\leq S$\\
\pushright{\int_{e_s} \sigma\frac{\partial v}{\partial n} dS = I_s} & $1\leq s\leq S$.}
\end{eqnarray*}
Here, the parameters $z_s>0$ are the presumably known contact impedances, and the currents satisfy the Kirchhoff's law, or conservation of charge condition,
\[
 \sum_{s=1}^S I_s \in\R^S_0 =\bigg\{V\in\R^L\,\bigg|\, \sum_{s=1}^S V_s = 0\bigg\}.
\] 
The solution of the boundary value problem is the unique solution $(v,V)\in H^1(\Omega)\times \R^S_0$ of the weak form variational problem
\[
 {\mathscr B}((w,W),(v,V)) = \sum_{s=1}^S I_s W_s = \langle (w,W),b_I\rangle,
 \quad\mbox{for all $(w,W)\in H^1(\Omega)\times\R^S_0$,}
\]
where $b_I = (0,I)\in H^1(\Omega)\times\R^S_0$, and
\[
 {\mathscr B}((w,W),(v,V)) = \int_\Omega \sigma\nabla w \cdot \nabla v dx + \sum_{s=1}^S \frac 1{z_s}\int_{e_s}(w - W_s)(v-V_s) dS.
\] 
To discretize the problem, assume that $\Omega$ is approximated by the union of triangular or tetrahedral elements, the mesh containing $n_{\rm f}$ nodes (`f' for fine), and let $\{\psi_j\}_{j=1}^{n_{\rm f}}$ denote a nodal-based piecewise polynomial Lagrange basis. Further, let $\{\phi_s\}_{s=1}^{S-1}$ denote a basis of $\R^S_0$. We define the basis functions $\overline \psi_j\in H^1(\Omega)\times \R^S_0$ as
\[
 \overline \psi_j = (\psi_j,0),\quad 1\leq j\leq n_{\rm f},\quad \overline\psi_{n_{\rm f}+s} = (0,\phi_s),\quad 1\leq s\leq S-1.
\]
We approximate the potential-voltage pair $(v,V)$ as
\[
 (v,V) = \sum_{j=1}^{n_{\rm f}+S-1} \alpha_j\overline \psi_j,
\]
and discretize the forward problem by choosing $(w,W) = \overline\psi_k$, to arrive at the Galerkin approximation,
\begin{equation}\label{eq for alpha}
 \sum_{j=1}^{n_{\rm f}+S-1} {\mathscr B}(\overline\psi_k,\overline \psi_j) \alpha_j =  \langle \overline\psi_k, b_I\rangle, \quad 1\leq k\leq n_{\rm f}+S-1.
 \end{equation}
Further, to parametrize the conductivity, we define a discretization of $\Omega$ by triangular or tetrahedral elements, independent of the discretization above, with $K$ nodes, and denote by $\{\eta_j\}_{j=1}^K$ the nodal-based piecewise polynomial Lagrange basis functions. We then parametrize the conductivity by writing
\[
 \sigma(x) = \sigma_0{\rm exp}\left(\sum_{j=1}^K u_j \eta_j(x)\right),\quad x\in\Omega,
\]  
where $\sigma_0>0$ is a fixed background conductivity.  The matrix $[{\mathscr B}(\overline\psi_k,\overline \psi_j)]$ defining the system (\ref{eq for alpha}) is parametrized by the vector $u\in\R^K$, and we write the equation in matrix form concisely as
\[
 \mA^{n_{\rm f}}_u \alpha = b(I),
\]
where we have indicated explicitly the dependency on the discretization by the number $n_{\rm f}$ of nodes. Solving this system for $\alpha$, extracting the last $S-1$ components $\alpha_{n_{\rm f} + s}$, $1\leq s\leq S-1$, and representing the voltage in terms of the basis functions $\phi_s$ defines the forward map
\[
  u \mapsto  V = \sum_{s=1}^{S-1}\alpha_{n_{\rm f}+s}\phi_s = \mR_u^{n_{\rm f}} I, \quad\mbox{where  $\alpha =(\mA_u^{n_{\rm f}})^{-1}b(I)$,}
\]
where $\mR^{n_{\rm f}}_u \in\R^{S\times S}$ is the resistance matrix.  We repeat the calculation for a full frame of $S-1$ linearly independent current patterns, $I^1,\ldots, I^{S-1}\in\R^S_0$, obtaining the full frame of voltage patterns $V^1,\ldots,V^{S-1}$. Finally, the voltage patterns are stacked together in a vector, constituting the forward model for the observation,
\[
 {\mathsf V} = \left[\begin{array}{c} V^1 \\ \vdots \\ V^{S-1}\end{array}\right] = F^{n_{\rm f}}(u), \quad F^{n_{\rm f}}:\R^K\to\R^{S(S-1)}.
\]
To guarantee satisfactory accuracy of the forward model, the discretization needs to be fine enough, in particular to capture the singularities of the voltage potential $v$ at the electrode edges. To demonstrate the modeling error effect, we construct a forward map defined over a coarser FEM mesh with $n_{\rm c}$ nodes ('c' for coarse), $n_{\rm c}<n_{\rm f}$, and denote the corresponding forward map by 
\[
 {\mathsf V} = F^{n_{\rm c}}(u), \quad F^{n_{\rm c}}:\R^K\to\R^{S(S-1)}, \quad n_{\rm c}<n_{\rm f}.
\]
Observe that the discretization of $u$ is independent of the FEM mesh, and is not changed when passing to a coarser computational mesh. In our computed examples, we use a piecewise linear Lagrange basis to represent both $u$ and $v$ over the different meshes.  The three meshes that we base our simulations, generated with the mesh generator described in \cite{perssonStrang}, on are shown in Figure~\ref{fig:meshes}. The number of electrodes is $S=16$.

\begin{figure}[h!]
\centerline{
\includegraphics[width=4cm]{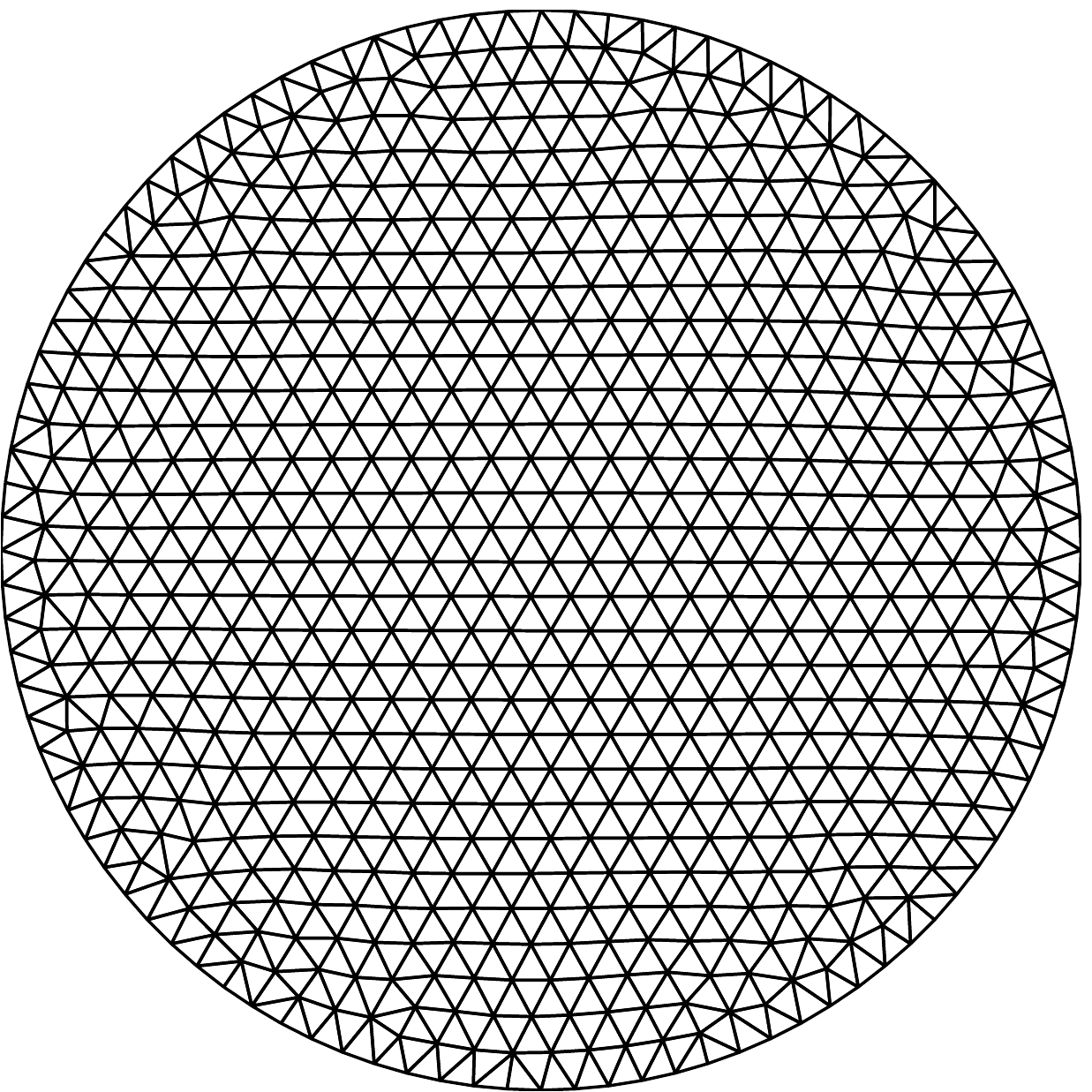}
\includegraphics[width=4cm]{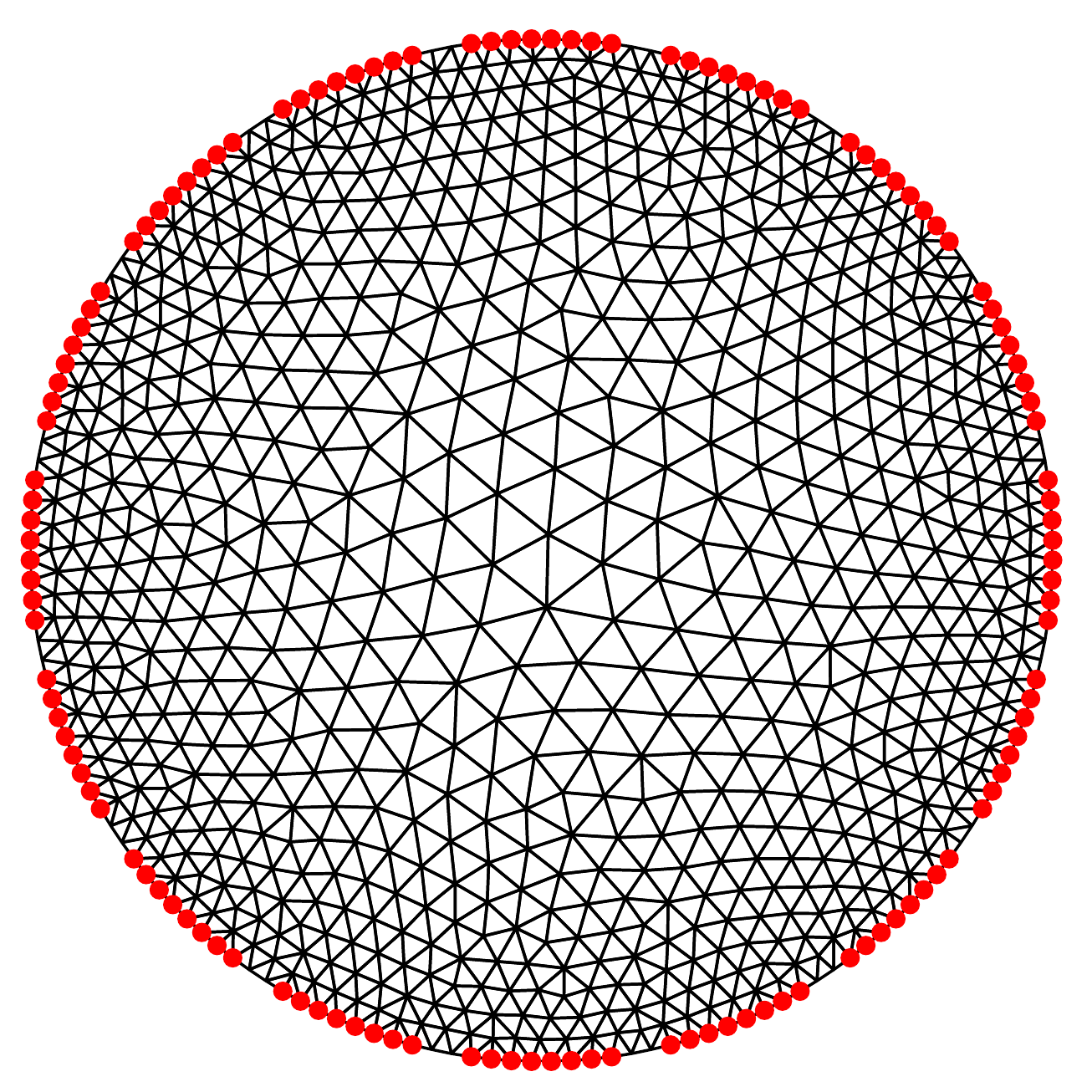}
\includegraphics[width=4cm]{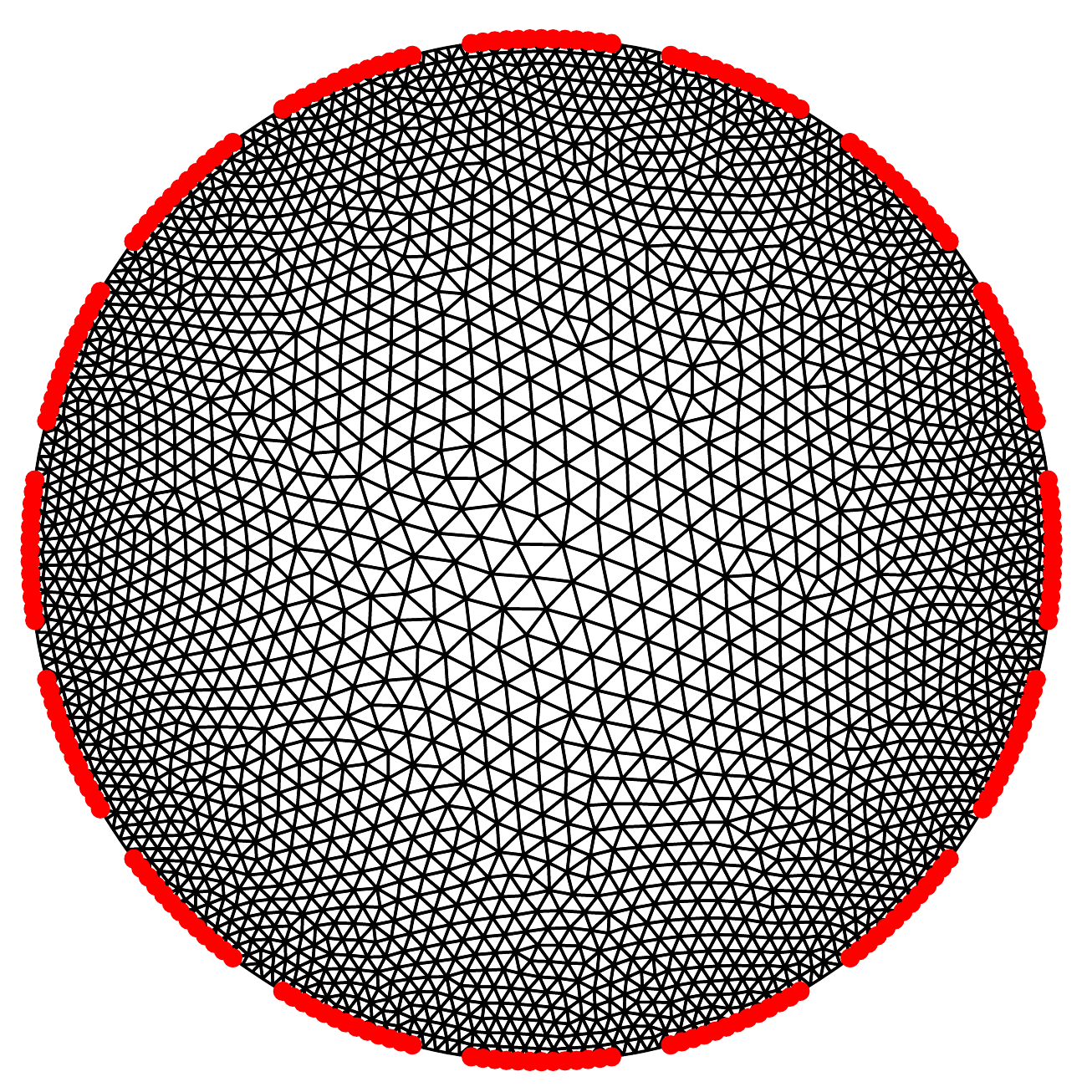}
}
\caption{\label{fig:meshes} Triangular meshes used in the numerical simulations. The number of electrodes is $L=16$, and they are indicated by red nodal points in the plot. The mesh for representing the conductivity distribution (left) has $K=733$ vertices and 1\,364 elements. The coarse mesh for the forward solver (middle) has $n_{\rm c}=877$ vertices and 1\,592 elements, and the fine scale mesh (right) consist of $n_{\rm f} =2\,418$ vertices and 4\,562 elements.
}
\end{figure}

We assign the Whittle-Mat\'{e}rn prior \cite{roininen2011, roininen2014} for the  vector $u$ defining the conductivity so that
\begin{eqnarray}
\label{eq:wmprior}
  \zeta \lambda^{-1} \left(-\lambda^2 \mL_g + \mI_K \right) u \sim {\mathcal N}(0, \mI_K),
\end{eqnarray}
where $\mL_g\in \R^{K \times  K}$ is the graph Laplacian defined on the conductivity mesh, $\lambda>0$ is a correlation length parameter, $\zeta >0$ is amplitude scaling, and $\mI_K$ is the identity matrix.  In Figure~\ref{fig:prior draws}, three independently drawn realizations of the conductivity distributions are shown. The values of the model parameters are indicated in the figure caption.

\begin{figure}[h!]
\centerline{
\includegraphics[width=4cm]{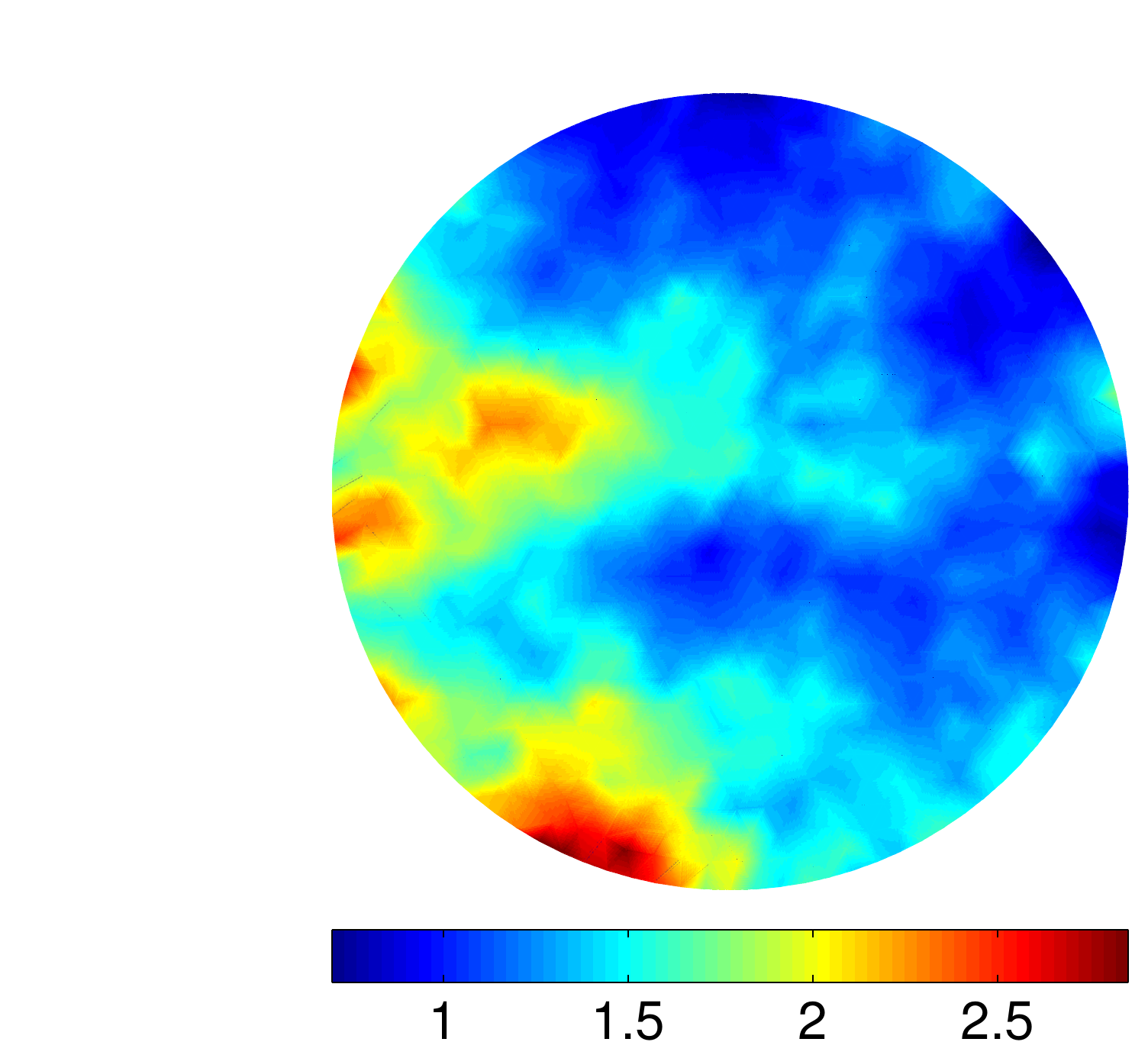}
\includegraphics[width=4cm]{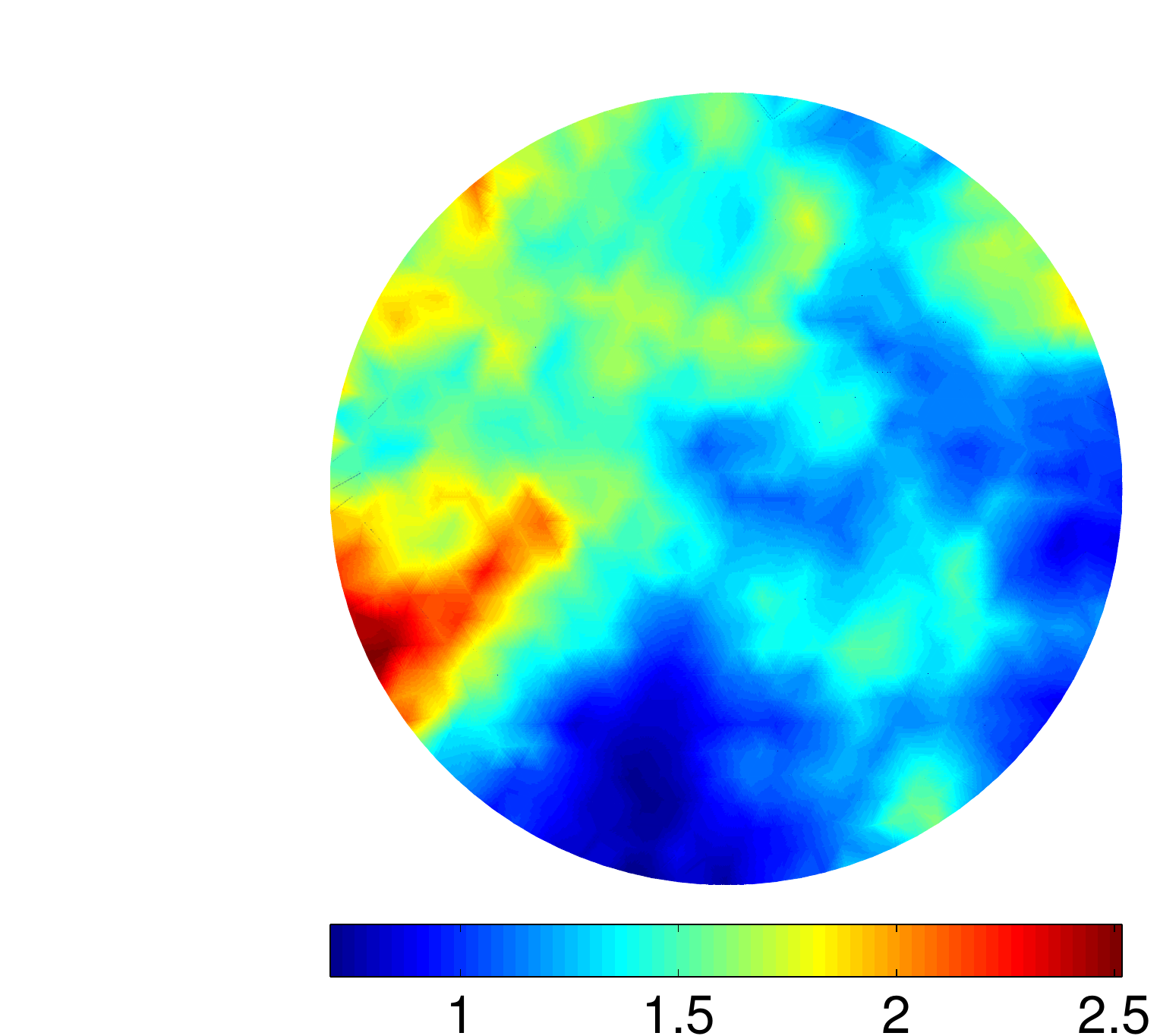}
\includegraphics[width=4cm]{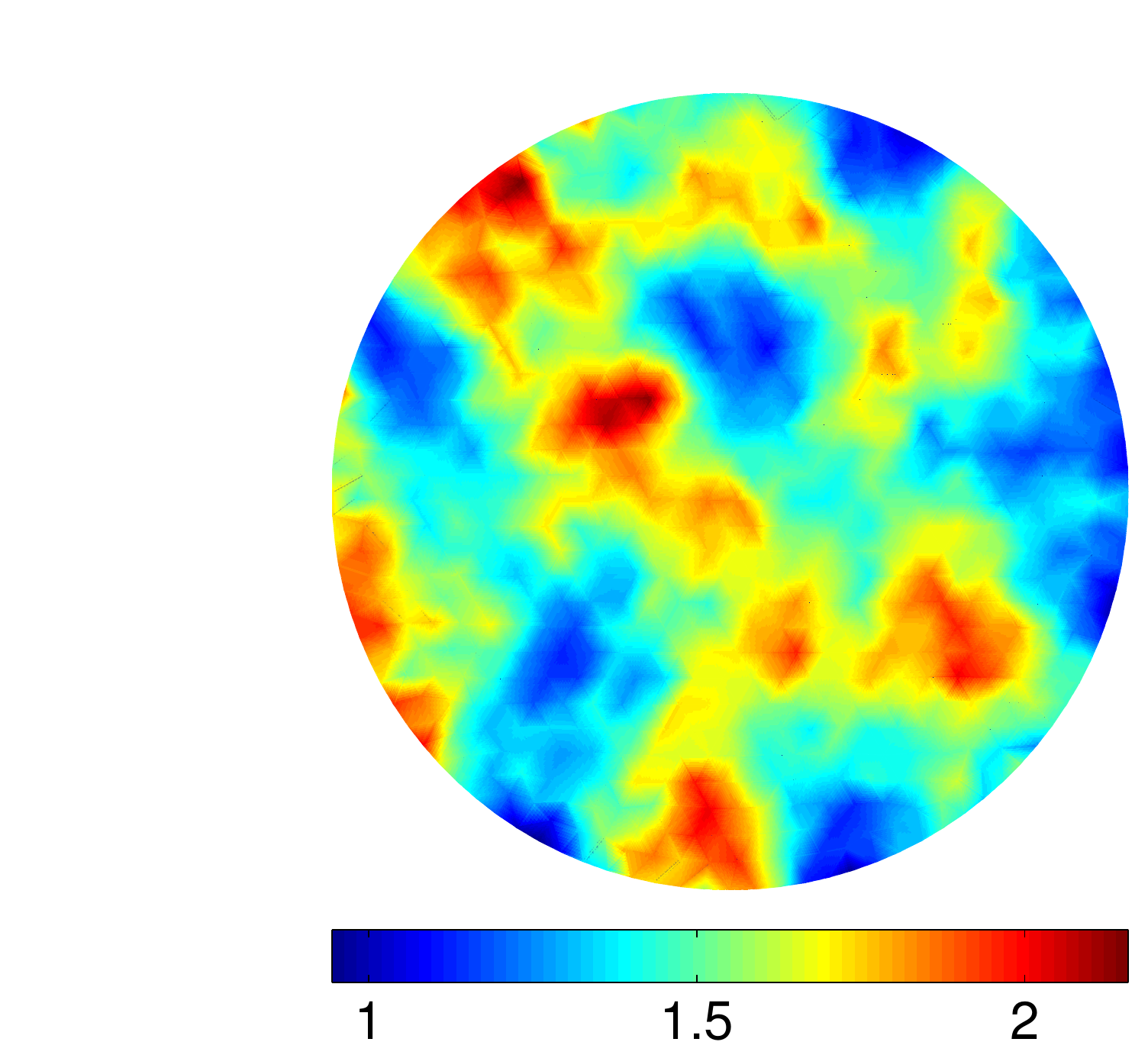}
}
\caption{\label{fig:prior draws} Conductivities $\sigma = \sigma_0{\rm exp} (u)$ corresponding to three independent draws of $u$ from the prior density. The parameter values used here are $\lambda = 0.2$ and $\zeta = 1/15$.  
The background conductivity is $\sigma_0 = 1.5$. The radius of the disc is unity, and the units are arbitrary.
}
\end{figure}

We generate the data using the fine scale model $F = F^{n_{\rm f}}$, and using the Conventional Error Model, i.e., ignoring the modeling error, compute a MAP estimate $u_{\rm MAP}$ using  the forward map $f = F^{n_{\rm c}}$ in the inverse solver. The estimate is based on a simple Gauss-Newton iteration. The additive noise covariance in this simulation is $\mGamma = \gamma^2\mI_{S(S-1)}$ with $\gamma= 10^{-3} V_{\rm max}$, where $V_{\rm max}$ is the maximum of all noiseless electrode voltages over the full frame of $S-1$ voltage patterns. The noise level is assumed to be low so that the modeling error is the predominant part in the uncertainty. In Figure~\ref{fig:true etc}, we show the conductivity distribution that was used to generate the synthetic data with the model $F$, the Conventional Error Model MAP estimate based on the coarse mesh model $f$, as well as the Enhanced Error Model estimate. In the latter, the modeling error mean and covariance are estimated from a sample of 1\,500 random draws from the prior of $u$.

\begin{figure}[h!]
\centerline{
\includegraphics[width=4cm]{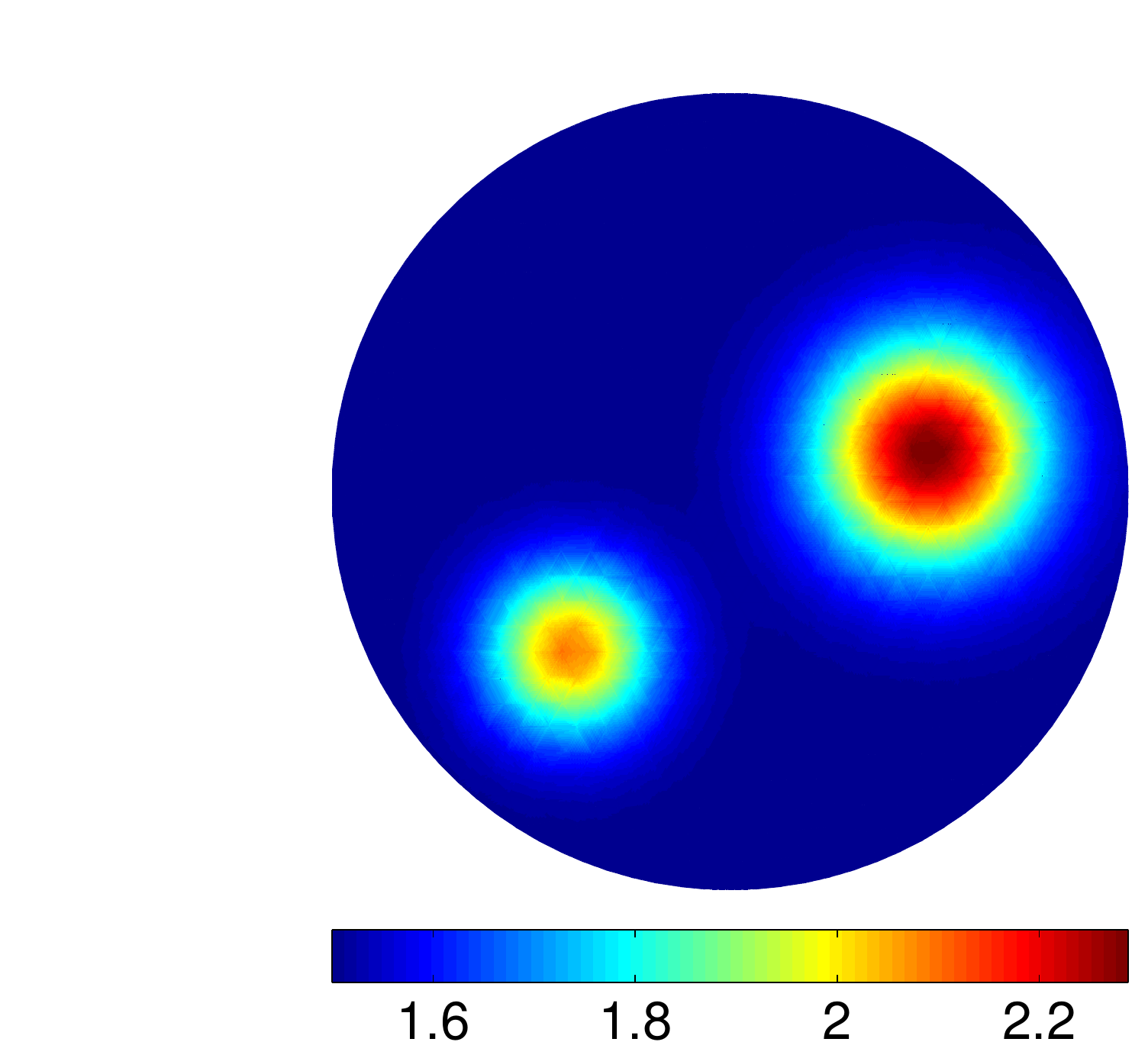}
\includegraphics[width=4cm]{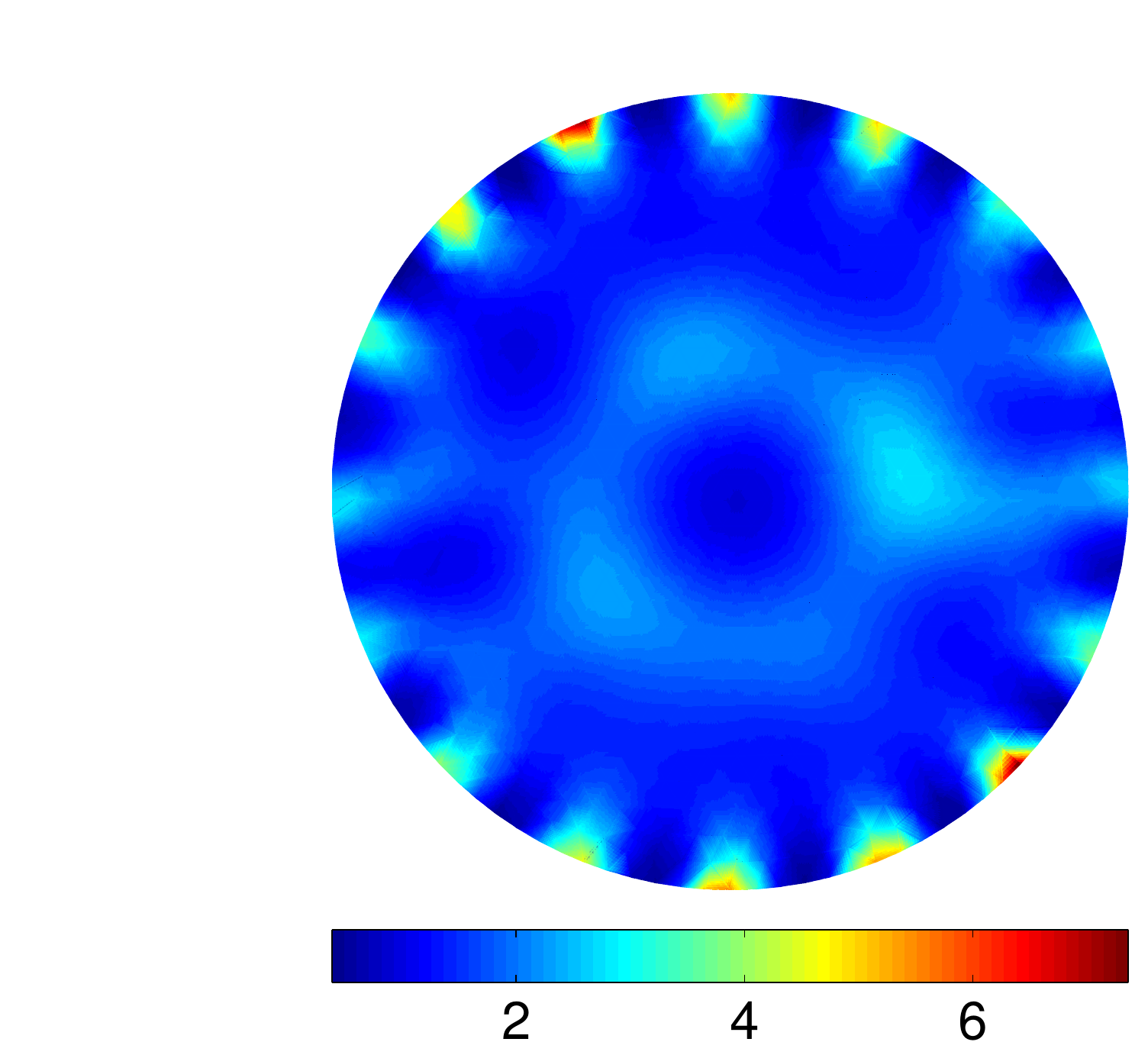}
\includegraphics[width=4cm]{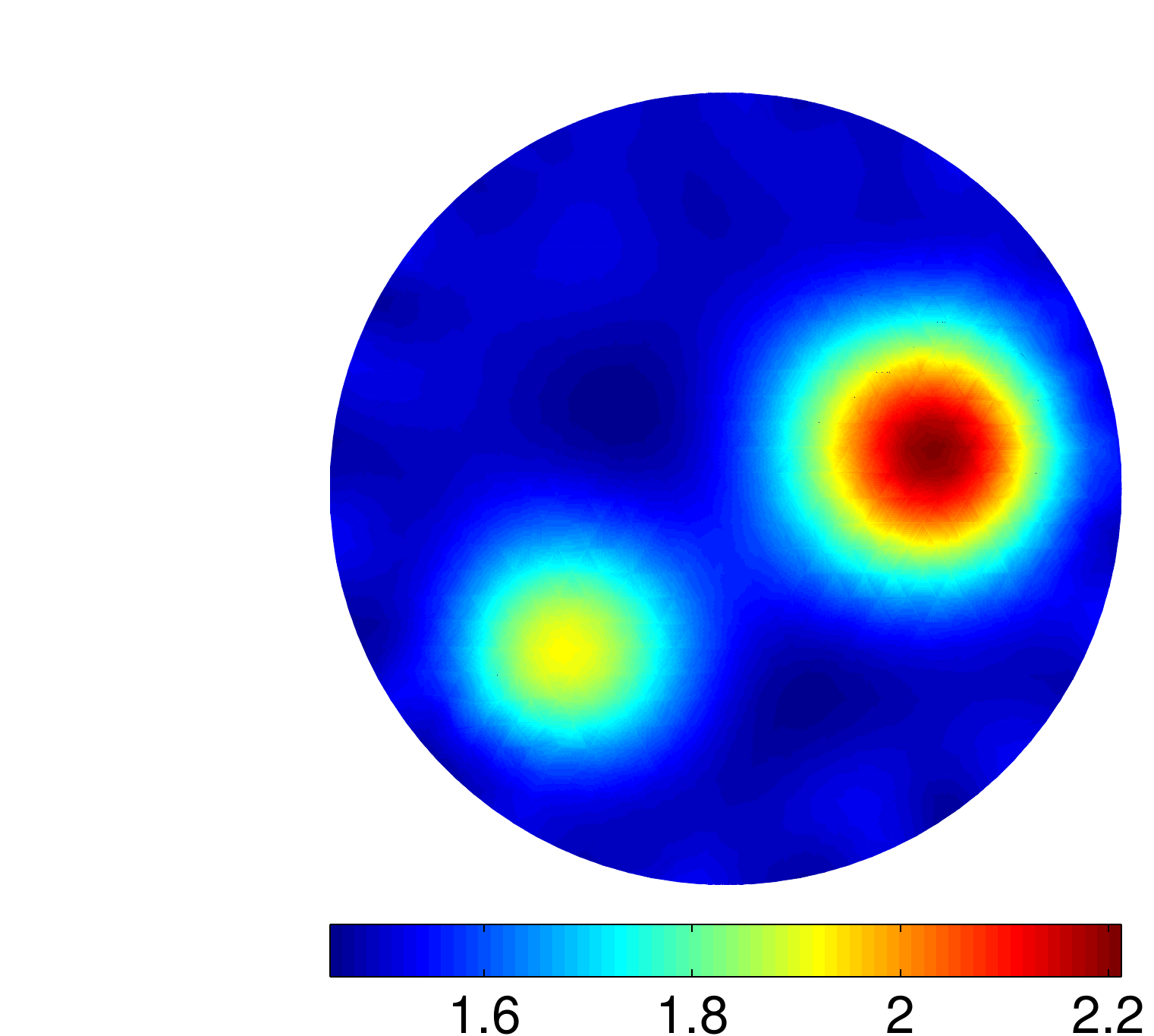}
}
\caption{\label{fig:true etc} Left: The true conductivity used to generate the test data using the finely discetized FEM forward model. Center: A Gauss-Newton-based MAP estimate based on the coarsely discretized FEM forward model, using the Conventional Error Model that ignores the modeling error.  Right: The MAP estimate computed by using the Enhanced Error Model, in which the modeling error mean and covariance are estimated from 1\,500 random draws from the prior. 
}
\end{figure}

Observe that in the reconstruction based on the Conventional Error Model, the true inclusions are completely overshadowed by the boundary artifacts that are concentrated around the edges of the electrodes. This is to be expected, since the basis functions in the coarse FEM mesh do not capture the voltage singularities at the electrode edges, and the inverse solution compensates the modeling error with elevated conductivity at the edges to mitigate the singularity. In agreement with previously published results, the Enhanced Error Model produces a solution without modeling error artifacts.

The computation of the MAP estimate, regardless of the error model, requires repeated linearization of the forward map. The re-evaluation of the Jacobian may be time consuming, and therefore it is tempting to replace the coarse mesh FEM model with a linearized approximation around the background conductivity $\sigma_0$ corresponding to $u=0$,
\[
 f(u) = F^{n_{\rm c}}(0) + {\mathsf D} F^{n_{\rm c}}(0) u. 
\]
The solution of the inverse problem with the linearized model and Gaussian prior is particularly straightforward, requiring a solution of a linear system. We iterate the posterior updating algorithm, generating samples ${\mathscr S}_\ell = \{u_\ell^1,\ldots,u_\ell^N\}$, $\ell=0,1,2,\cdots$ using the modeling error updating scheme. In Figure~\ref{fig:estimates, linear}, we plot the conductivities corresponding to the posterior means,
\[
 \overline\sigma_\ell = \sigma_0{\rm exp}(\overline u_\ell),\quad \overline u_\ell = \frac 1N\sum_{j=1}^N u_\ell^j,
\]
as well as the marginal variances of the parameters $u_\ell$, that is,
\[
 {\rm var}_\ell = {\rm diag}\left(\frac 1N\sum_{j=1}^N(u^j_\ell - \overline u_\ell)(u^j_\ell - \overline u_\ell)^\mT\right).
\]
The sample size here was $N = 5\,000$. 

\begin{figure}
\centerline{\includegraphics[width=3.2cm]{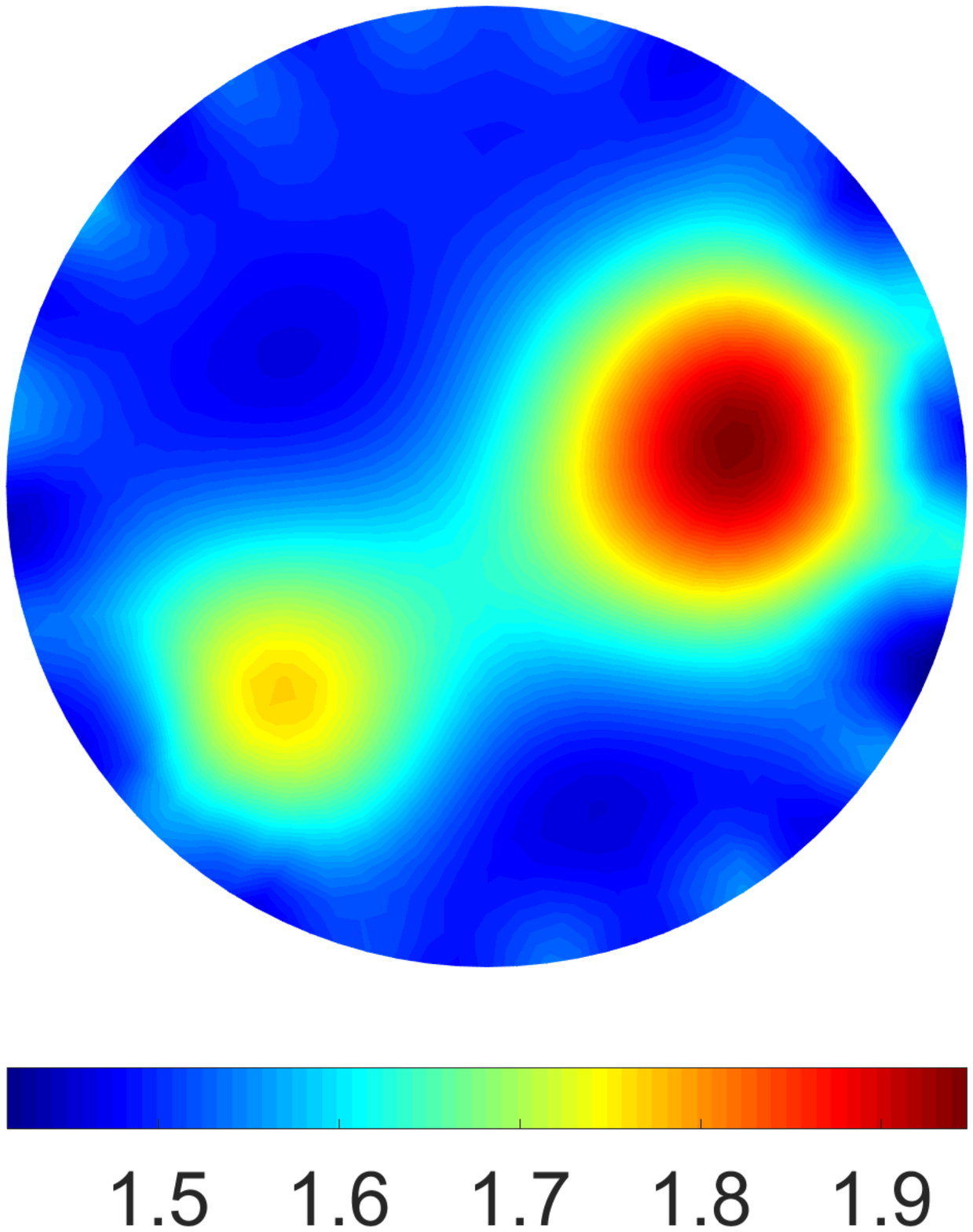}
\includegraphics[width=3.2cm]{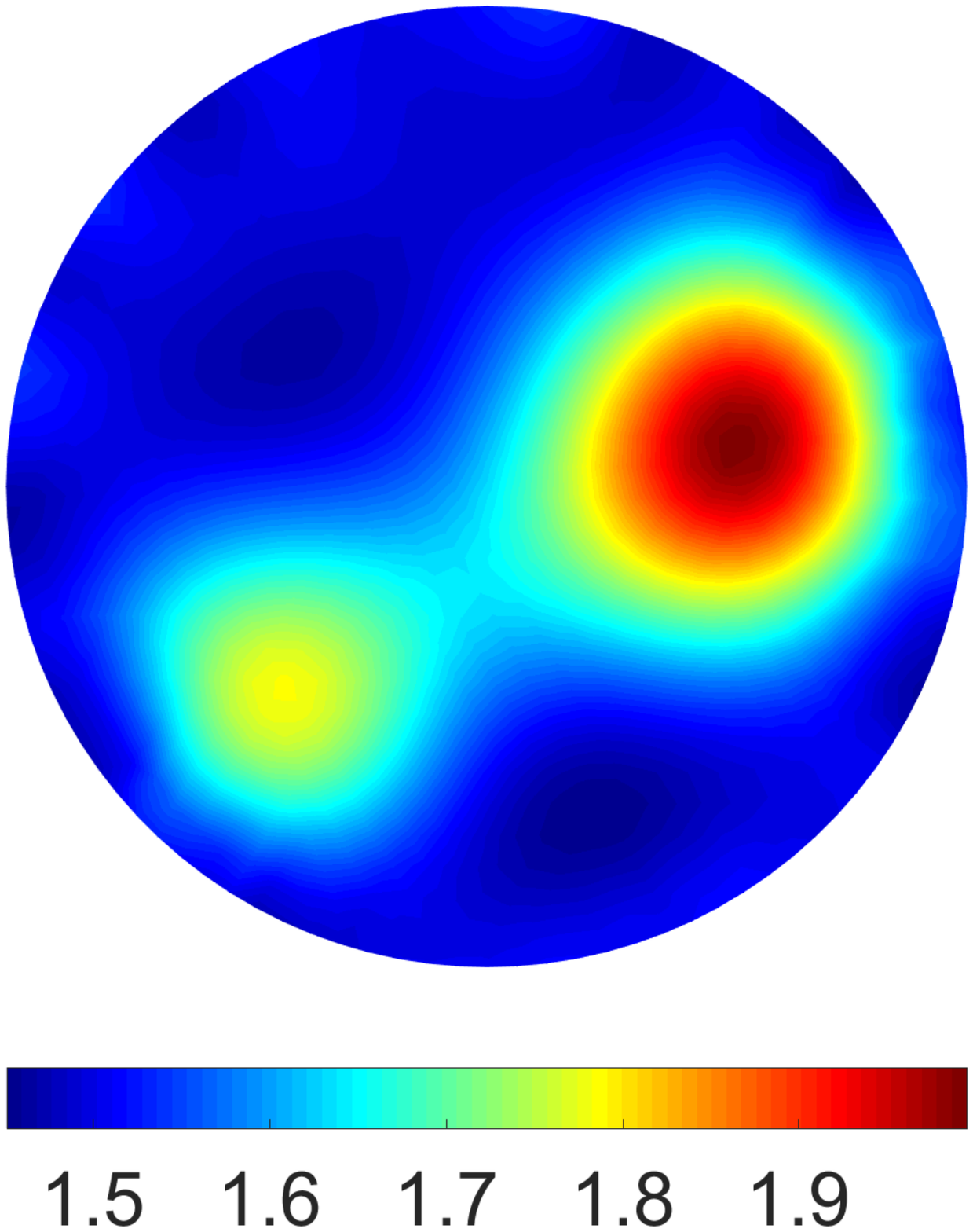}
\includegraphics[width=3.2cm]{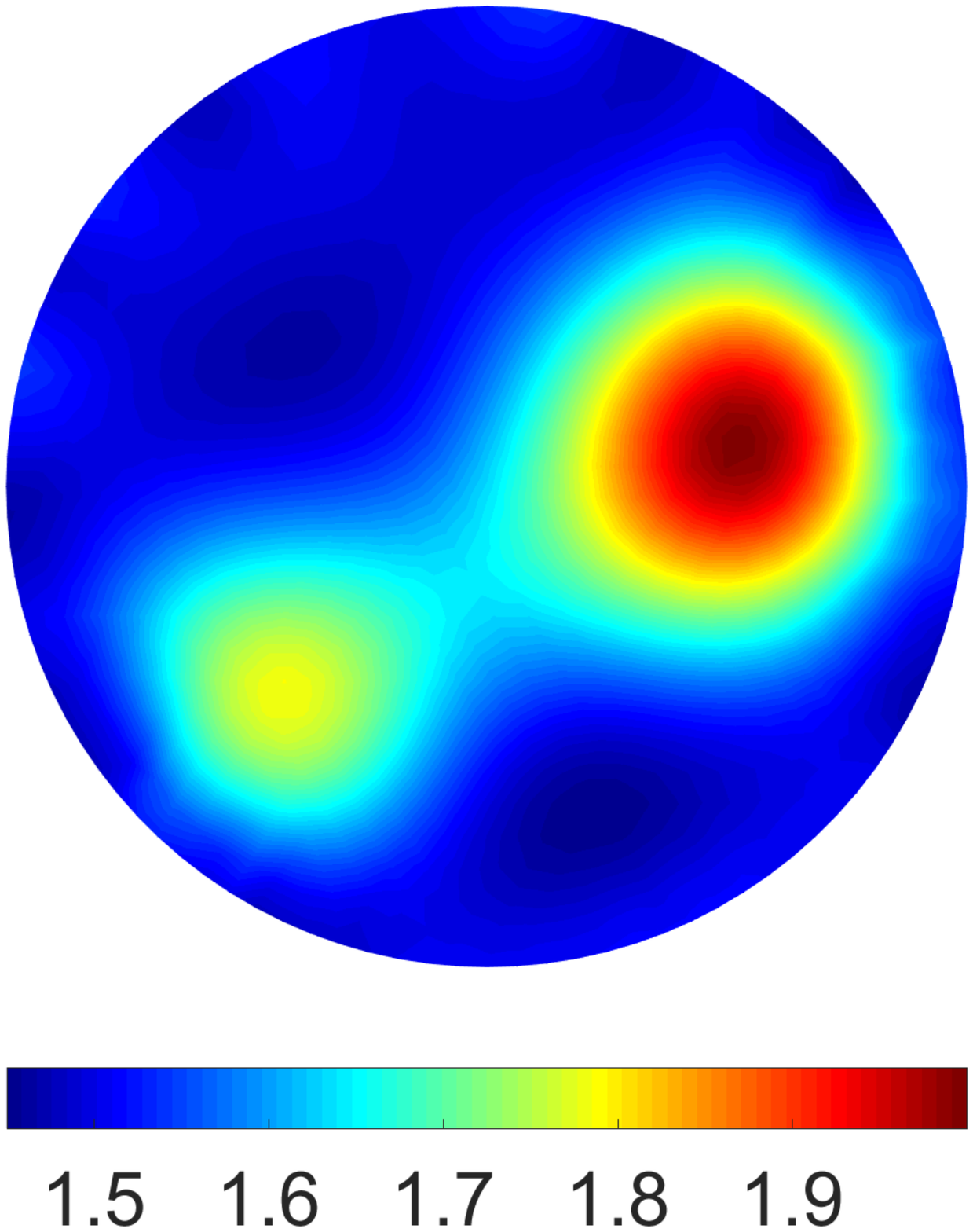}
\includegraphics[width=3.2cm]{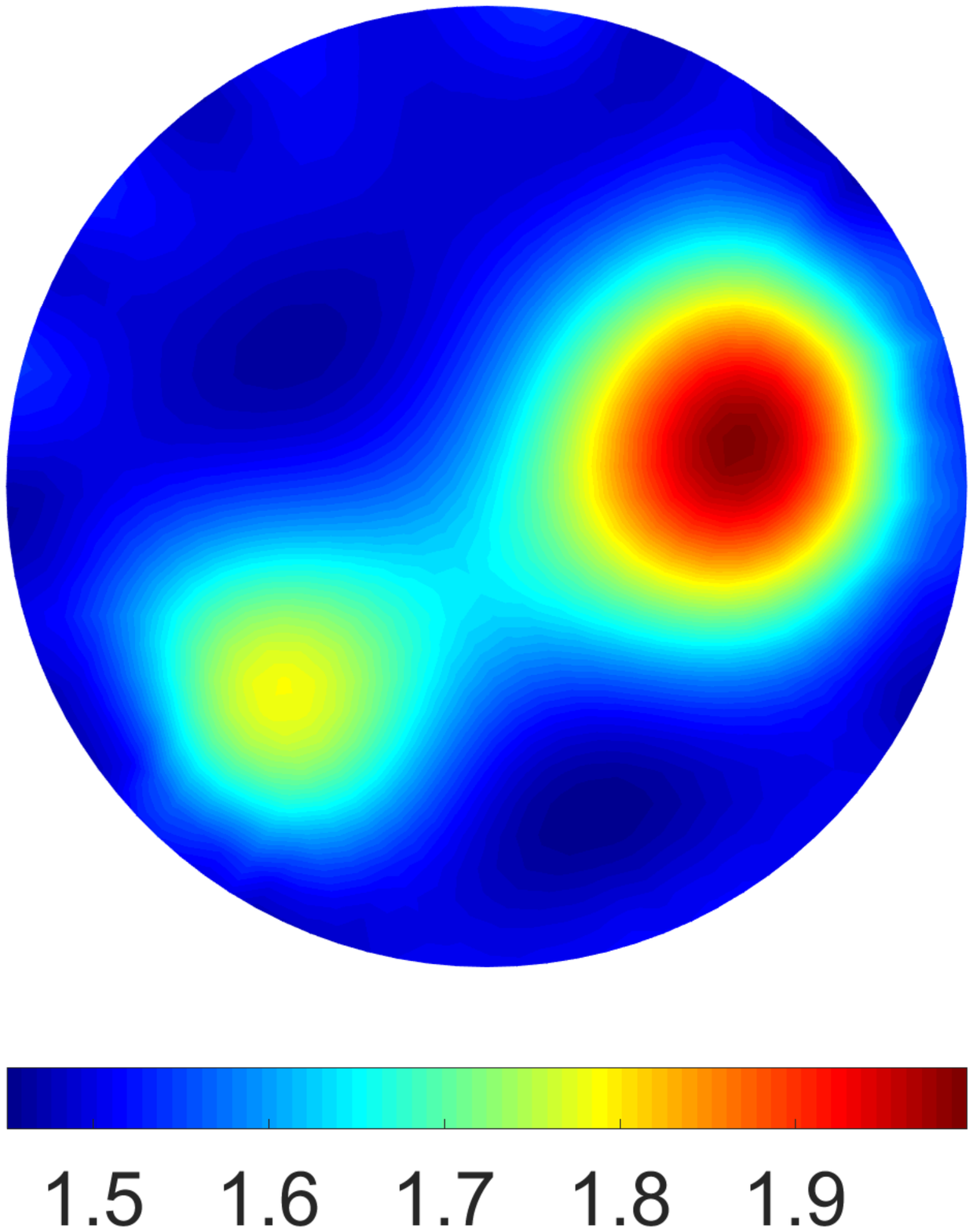}
\includegraphics[width=3.2cm]{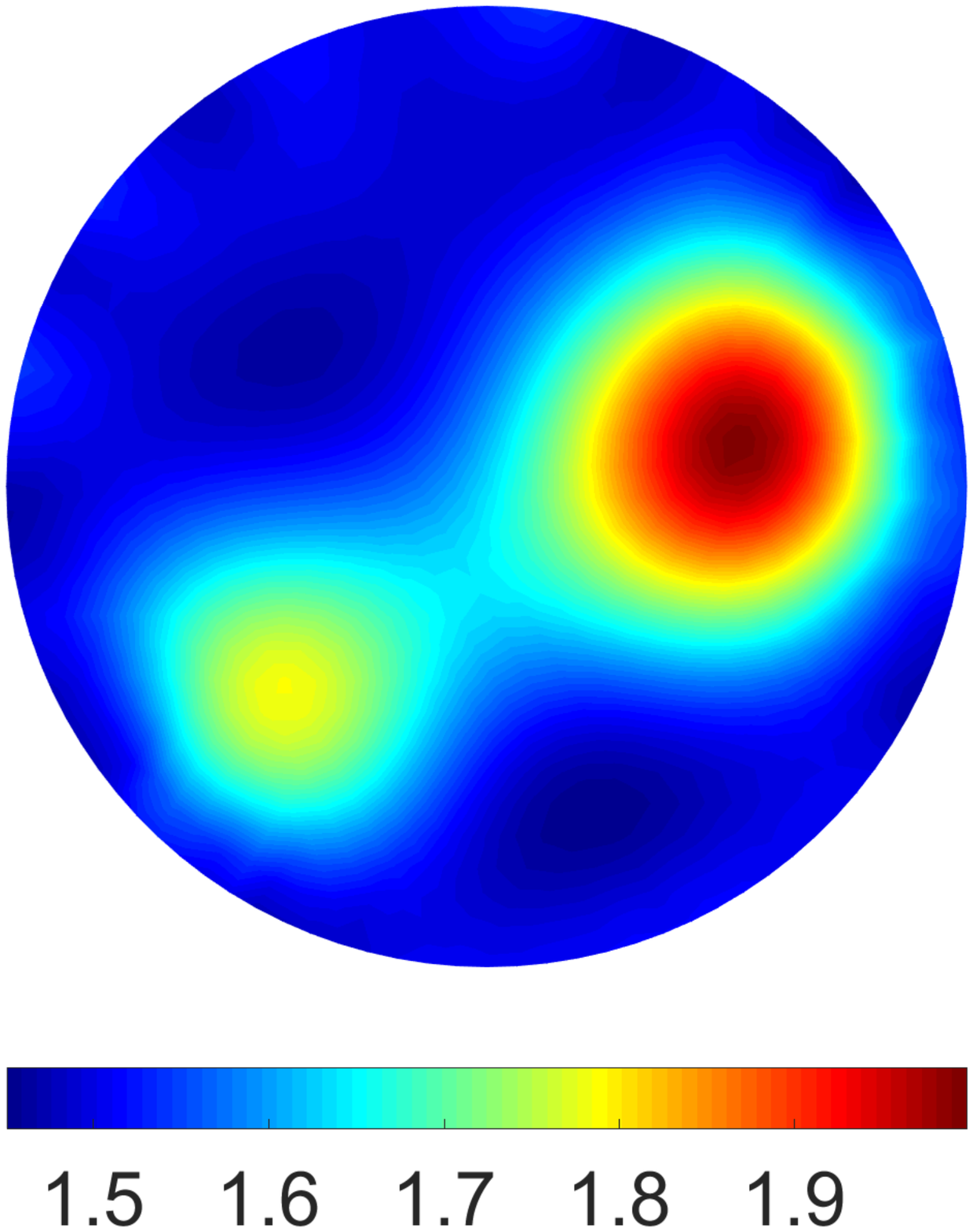}
}
\centerline{\includegraphics[width=3.2cm]{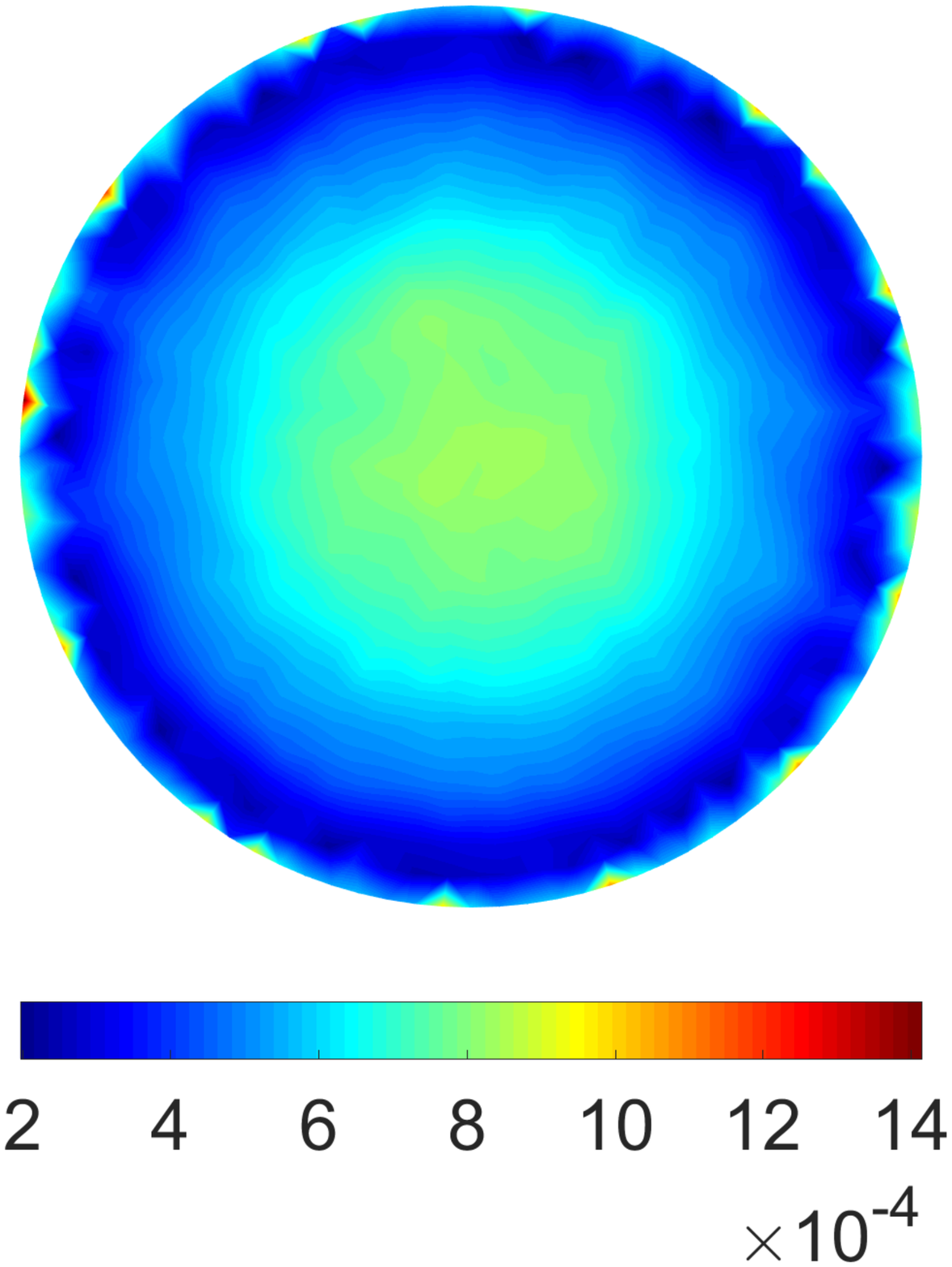}
\includegraphics[width=3.2cm]{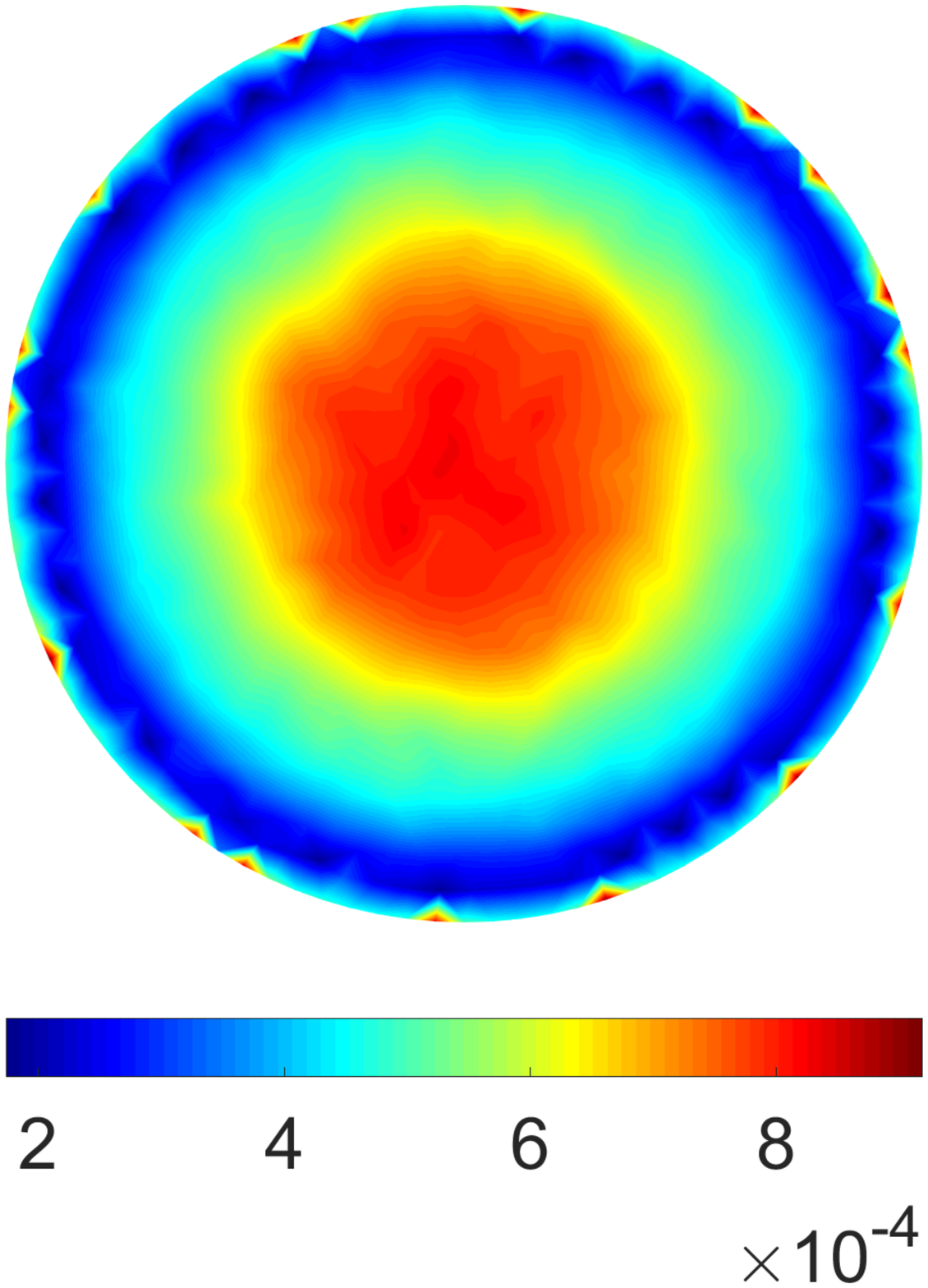}
\includegraphics[width=3.2cm]{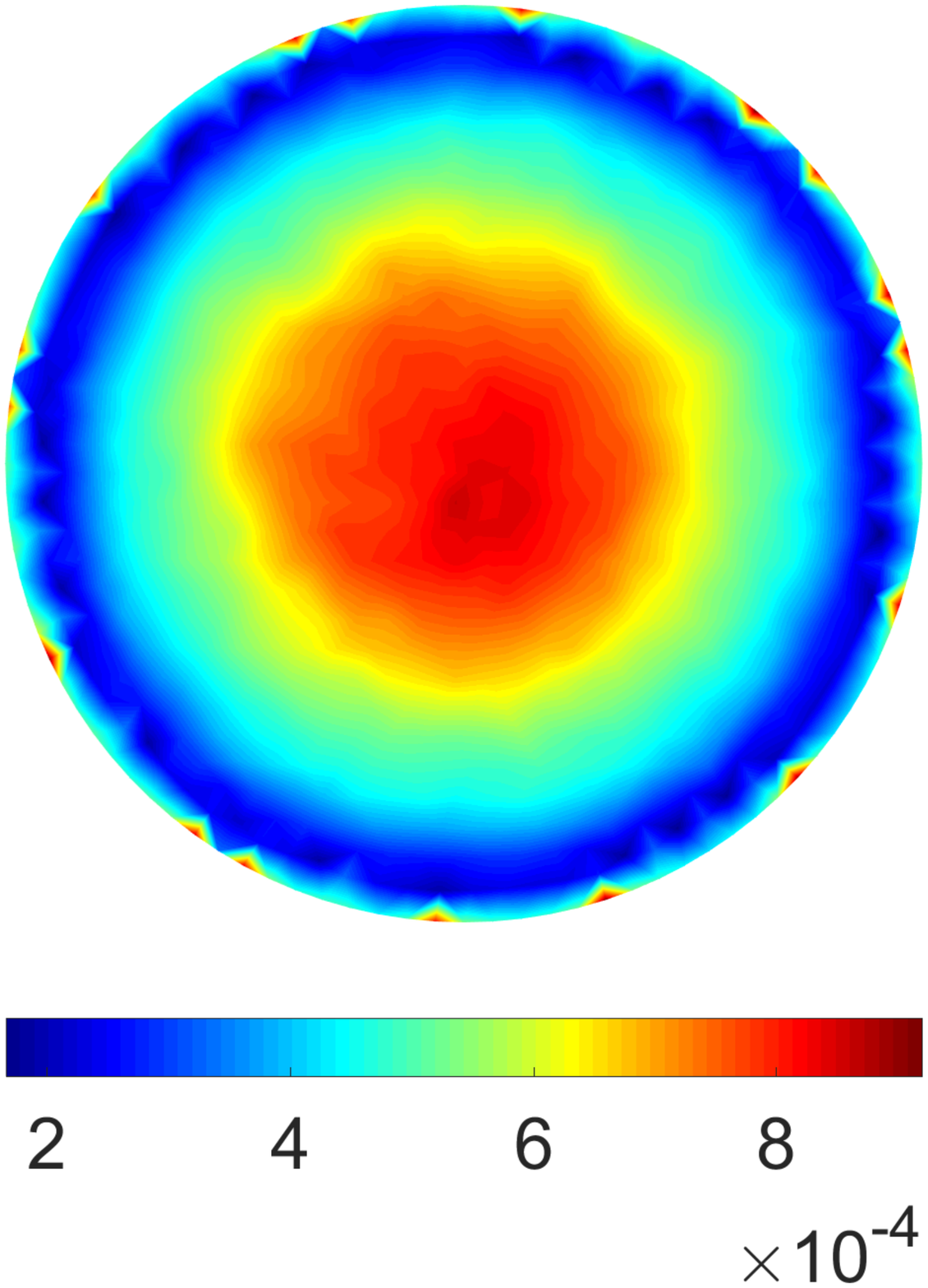}
\includegraphics[width=3.2cm]{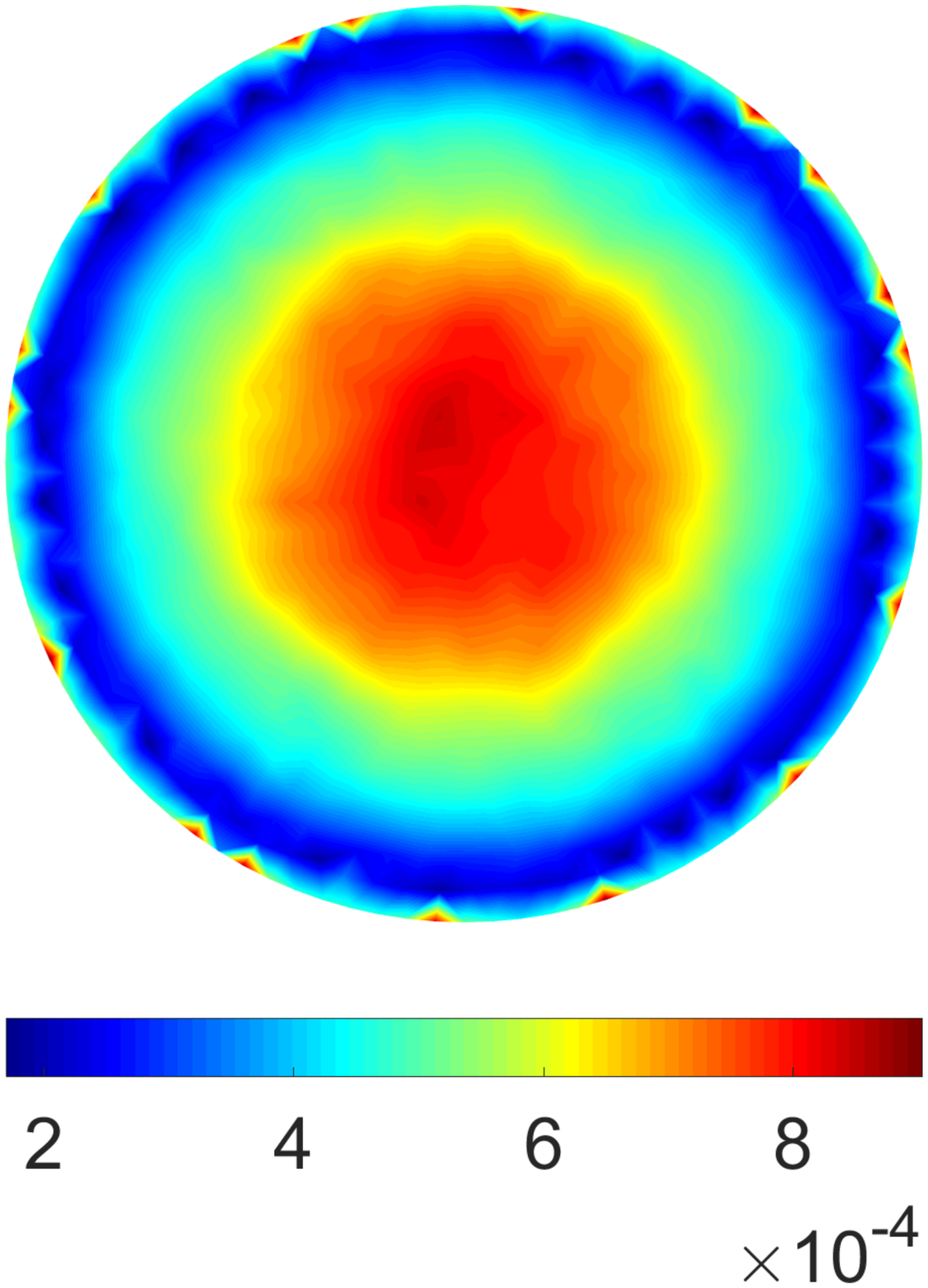}
\includegraphics[width=3.2cm]{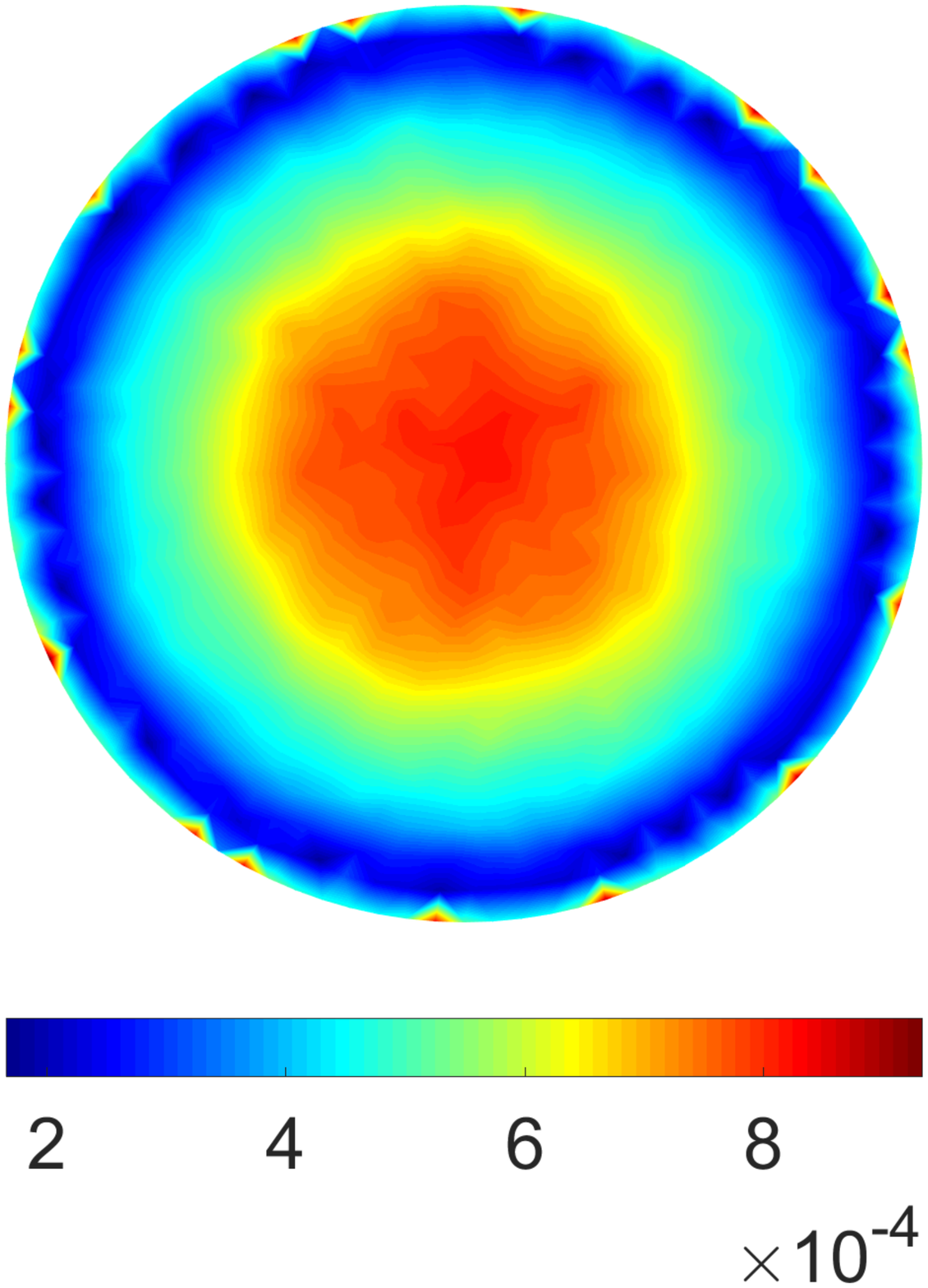}
}
\caption{\label{fig:estimates, linear} Sample means of the conductivity $\overline u_\ell$ for $\ell=1,2,3,4$ (upper row), and sample marginal variances of the components of the vectors $u$ at the same iterations.
}
\end{figure}

Finally, we consider the convergence of the iterated densities $\pi_\ell$ towards the posterior density by means of the Kullback-Leibler divergence, which we approximate using the particles drawn from $\pi_\ell$,
\begin{eqnarray*}
 D_{\rm KL}(\pi_\ell \| \pi_{\rm post}) &=& \int\pi_\ell(u) \log\left(\frac{\pi_\ell(u)}{\pi_{\rm post}(u)}\right)du \approx \frac 1N \sum_{j=1}^N  
  \log\left(\frac{\pi_\ell(u_\ell^j)}{\pi_{\rm post}(u_\ell^j)}\right) \\
  &=&\frac 1N \sum_{j=1}^N  
  \log\left(\frac{\pi_\ell(b\mid u_\ell^j)}{\pi(b\mid u_\ell^j)}\right)  -  \log\left(\frac{\pi(b)}{\pi_\ell(b)}\right),
\end{eqnarray*}
the second term corresponding to the normalization factors of the true and approximate posteriors.
Observe that to evaluate the posterior density, the fine mesh model needs not to be evaluated anew, since the fine mesh evaluations are already computed for the modeling error sample. The sample-based approximation of the KL divergence is straightforward to compute up to the normalizing constants.

Figure~\ref{fig:mean error,linear} shows the sample-based estimates of the Kullback-Leibler divergence for $\ell = 1,2,\ldots,\ell_{\rm max} = 5$. To subtract the unknown normalization offset, we plot the differences  
\begin{equation}\label{DeltaKL}
\Delta  D_{\rm KL}(\pi_\ell \| \pi_{\rm post})  = D_{\rm KL}(\pi_\ell \| \pi_{\rm post}) -  D_{\rm KL}(\pi_{\ell_{\rm max}} \| \pi_{\rm post}).
\end{equation}
The figure shows also the relative error of the sample mean approximating the true conductivity,
\[
 e_r(\overline u_\ell) = \frac{\| \sigma -\overline\sigma_\ell\|}{\|\sigma\|}.
\]
As in the previous subsection, the numerical results demonstrate that the
approach to model error advocated in this paper leads to improved estimates
of the true value used to generate the data.

\begin{figure}
\centerline{\includegraphics[height=5.3cm]{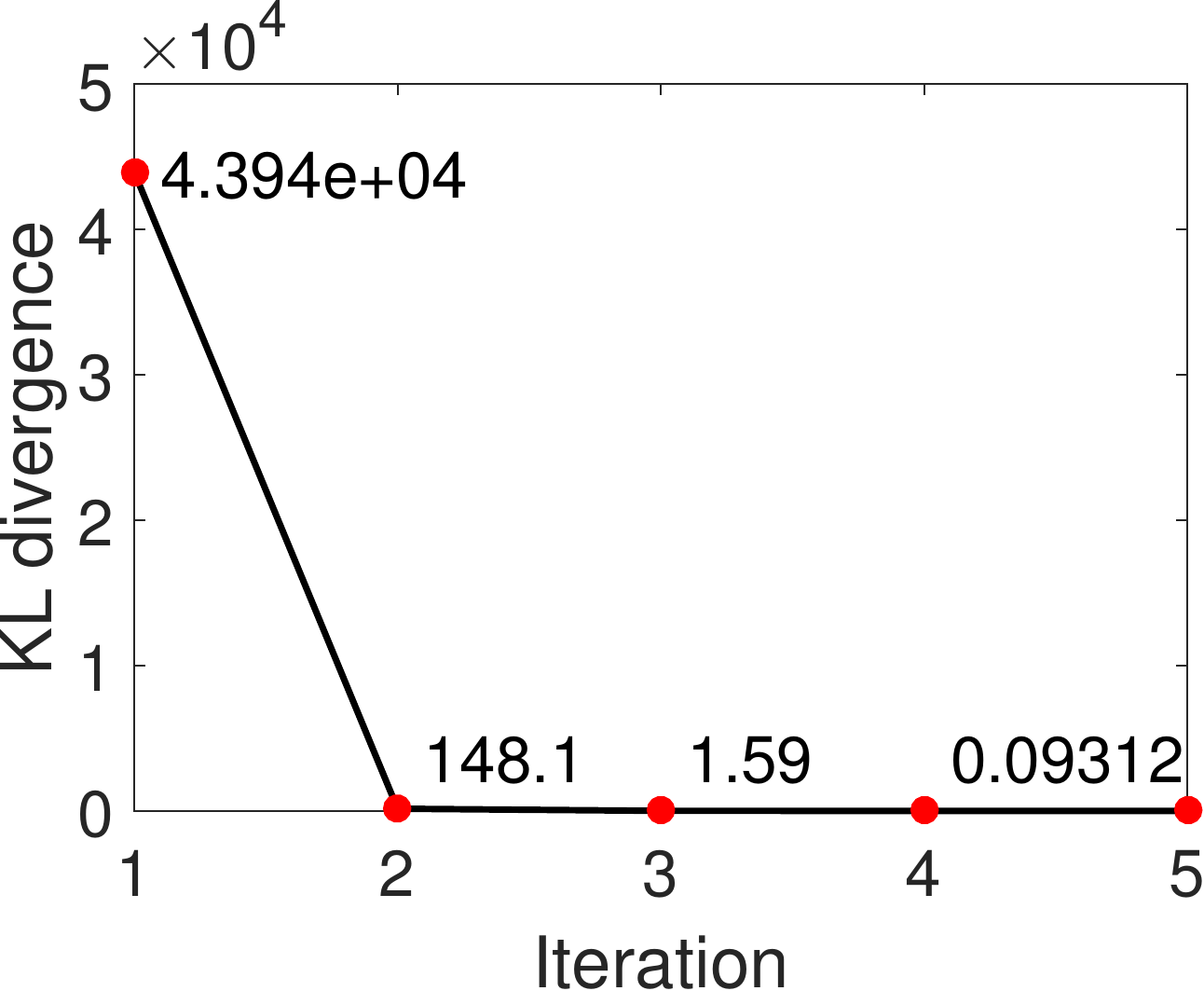}
\includegraphics[height=5cm]{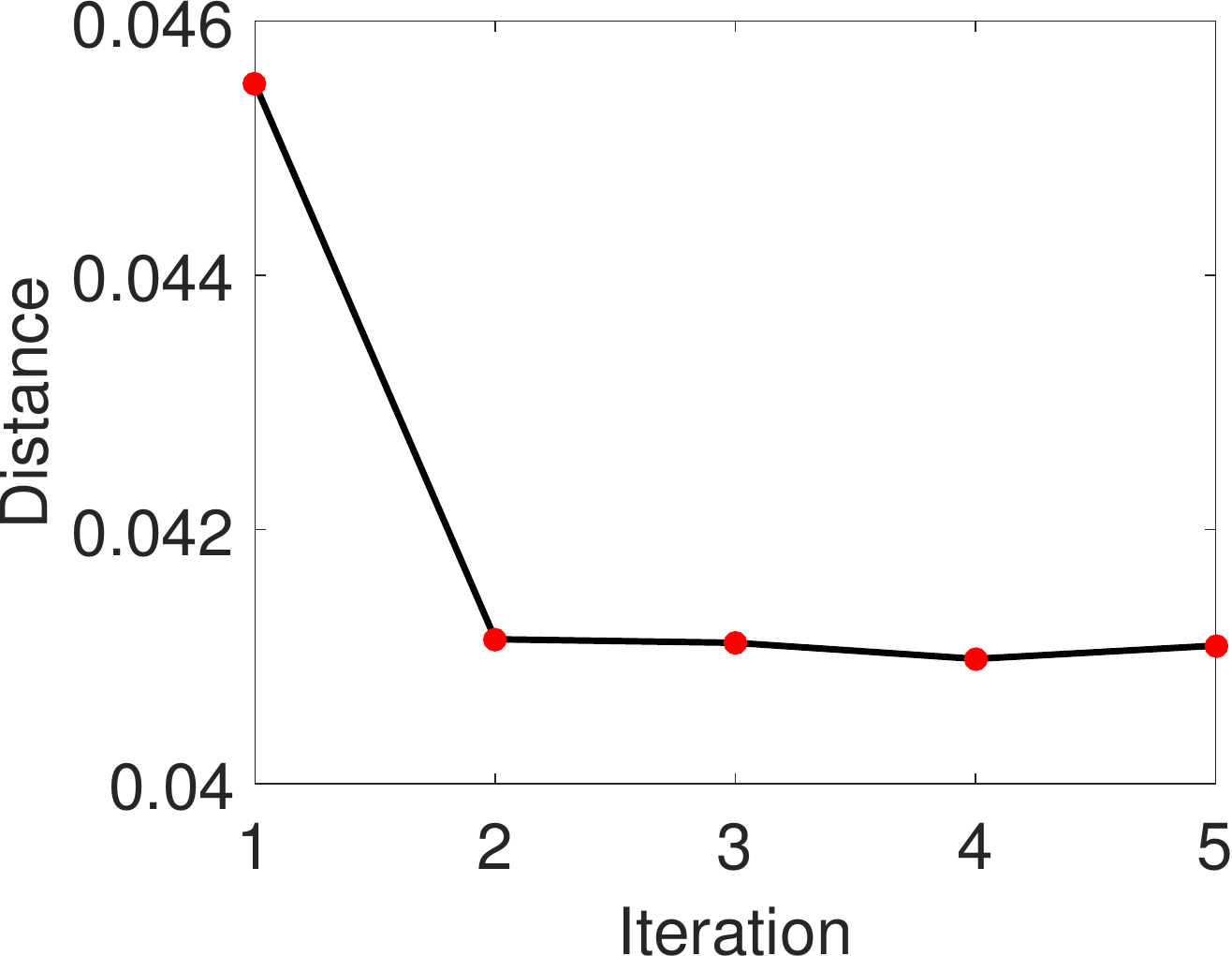}}
\caption{\label{fig:mean error,linear} Left: The Kullback-Leibler divergence (\ref{DeltaKL}) difference estimated by using a sample of $5\,000$  realizations drawn from the approximate posterior densities $\pi_\ell$. Right: The relative distance of the sample mean from the true conductivity over five iterations.}
\end{figure}

\subsection{Steady State Darcy Flow}
\label{ssec:gwf}

In the last computed example, we consider the inverse problem of estimating the permeability distribution in porous medium from a discrete set of pressure measurements. More precisely,  let the computational domain be $\Omega = (0,1)^2$, and define $X = L^\infty(\Omega)$. For a given $u \in X$, called {\em log-permeability} define the pressure filed $p = P(u) \in H^1_0(\Omega)$ to be the solution to the steady-state Darcy equation with Dirichlet boundary conditions,
\begin{eqnarray}
\label{eq:gwf_pde}
\cases{\pushright{-\nabla\cdot(e^u\nabla p) = g} &$x\in\Omega$\\
\pushright{\hfill p=0} &$x\in\partial\Omega$
}
\end{eqnarray}
for some fixed and presumably known source term $g \in H^{-1}(\Omega)$.

We define now the observation operator $\mathcal{O}:H^1_0(\Omega)\rightarrow\mathbb{R}^J$ by
\[
\mathcal{O}_j(p) = \frac{1}{2\pi\eps}\int_\Omega p(x)e^{-\frac{1}{2\eps^2}(x-q_j)^2}\,\dee x,\quad j=1,\ldots,J,
\]
for some set of points $\{q_1,\ldots,q_J\} \subseteq \Omega$. The observations are smoothed  observations at the points $\{q_1,\ldots,q_J\}$, converging to point observations as $\varepsilon\to 0$. Note that each $\mathcal{O}_j$ is a bounded linear functional on $H^1_0(\Omega)$. In what follows, we choose $\eps = 0.02$, $g(x_1,x_2) =  100\sin(\pi x_1)\sin(\pi x_2)$, and let $\{q_1,\ldots,q_{25}\}$ be a uniformly spaced grid of $25$ points in $\Omega$. The accurate model is then defined by the composition $F = \mathcal{O}\circ P$. As in the EIT example, the approximate model is defined through linearization,
\[
f(u) = F(u_0) + \mathsf{D}F(u_0)u.
\]
for some fixed $u_0 \in X$; in this example we choose $u_0 = 0$. To construct the linear model, the derivative may be computed inexpensively using an adjoint method. The computation of the full Jacobian $\mathsf{D}F(u_0)$ requires $J+1$ numerical solutions of a PDE of the form (\ref{eq:gwf_pde}), and needs to be performed only once.

In this example, we generate three different data sets 
corresponding to different noise levels: The noiseless data generated by using the non-linear model is
perturbed by additive observational noise drawn from  normal distribution 
$\mathcal{N}(0,\mGamma_i)$, where $\mGamma_i = 10^{-i-1}\mI$, and $i=1,2,3$. 
The true log-permeability $u^\dagger$, defined as the sum of two unnormalized Gaussian densities, is shown in Figure \ref{fig:gwf_truth}. In the same figure, the computed pressure field is shown, with the observation points indicated by black dots. Each data set is generated using a uniform mesh of $128\times 128$ points, while in the inverse computations, we use a reduced model with a uniform mesh of $64\times 64$ points. We perform $10$ iterations of the posterior updating algorithm, using as few as $100$ particles in the particle approximation model.

As in the EIT simulations, we choose a Whittle-Mat\'ern prior distribution for the vector $u$ defining the permeability as given in (\ref{eq:wmprior}). We make the choices $\lambda=0.1$ and $\zeta=1$, and note that here $\mL_g$ corresponds to the the finite-difference Laplacian on the reduced mesh using the standard 5-point stencil.

In Figure \ref{fig:gwf_means} the conditional means arising from the different error models and data sets are shown. In this example the conventional and enhanced error models have very similar performance in terms of inferring the conditional mean. They are both able to infer the geometry of the log-permeability field, particularly when the observational noise is small, however they fail to obtain the magnitude. The iterative algorithm proposed in this article is able to obtain both the geometry and magnitude with good accuracy in a small number of iterations. Figure \ref{fig:gwf_conv} shows the evolution of the size of the error between the conditional mean at iteration $\ell$ and the true log-permeability field. In all cases the error has converged in 4 or 5 iterations, similarly to what was observed in the EIT experiments.
As in the previous two subsections, the numerical results demonstrate that the approach to model error advocated in this paper leads to improved accuracy of the point estimates of the true parameter underlying the data.

\begin{figure}
\includegraphics[width=\linewidth, trim=1cm 1cm 1cm 1cm, clip]{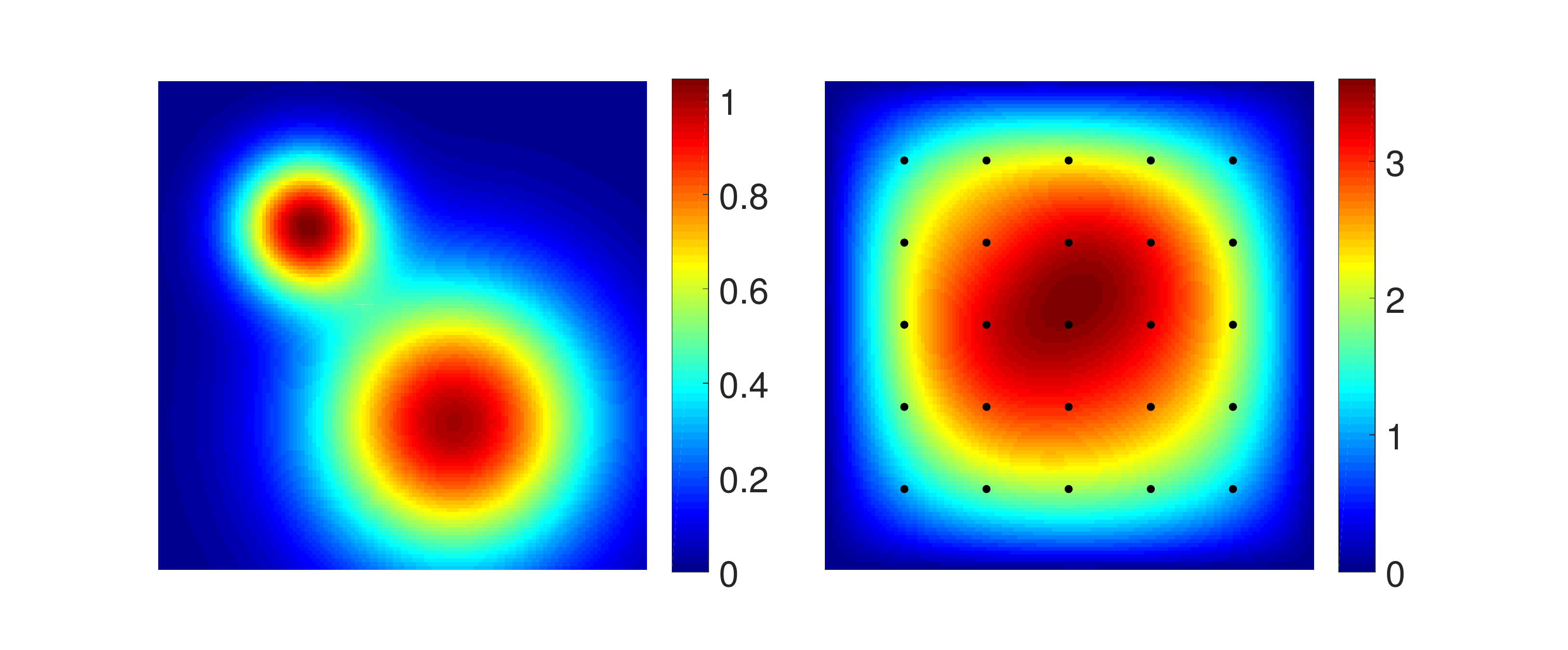}
\caption{(Left) The true log-permeability used to generate the data. (Right) The true pressure and the observation points $\{q_j\}_{j=1}^{25}$.}
\label{fig:gwf_truth}
\end{figure}

\begin{figure}
\centerline{
\includegraphics[width = 15cm]{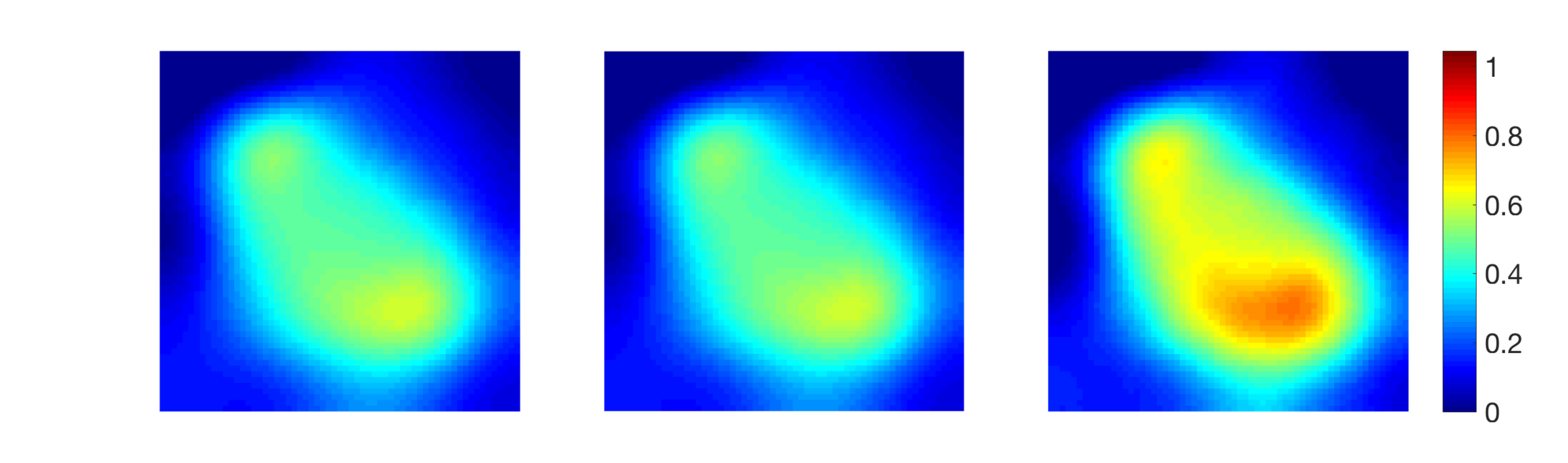}
}
\centerline{
\includegraphics[width = 15cm]{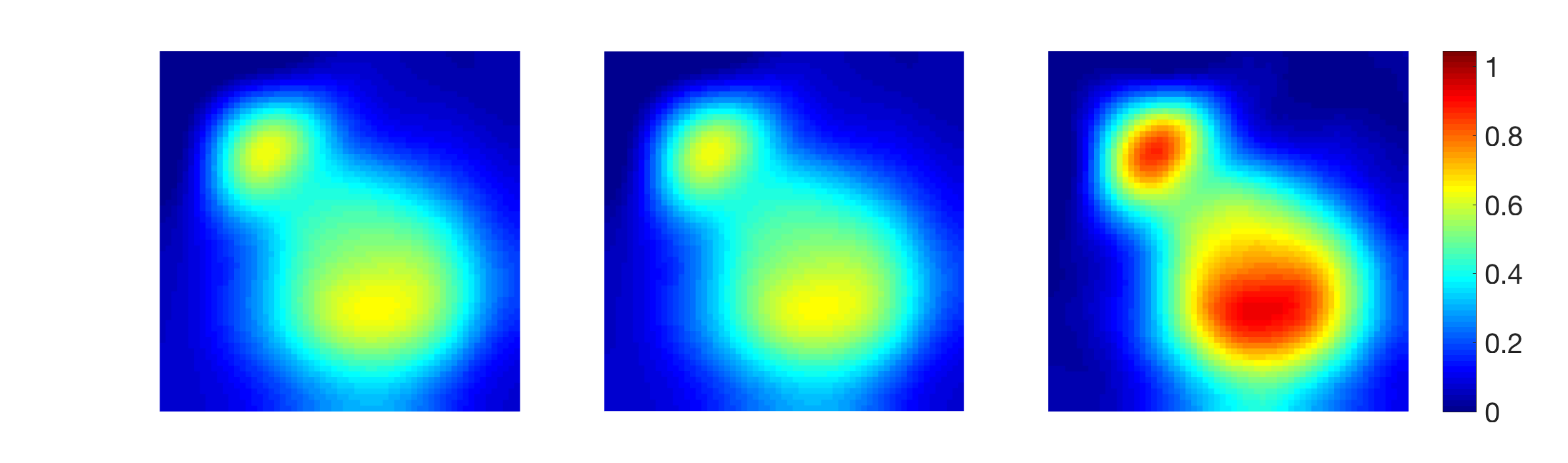}
}
\centerline{
\includegraphics[width = 15cm]{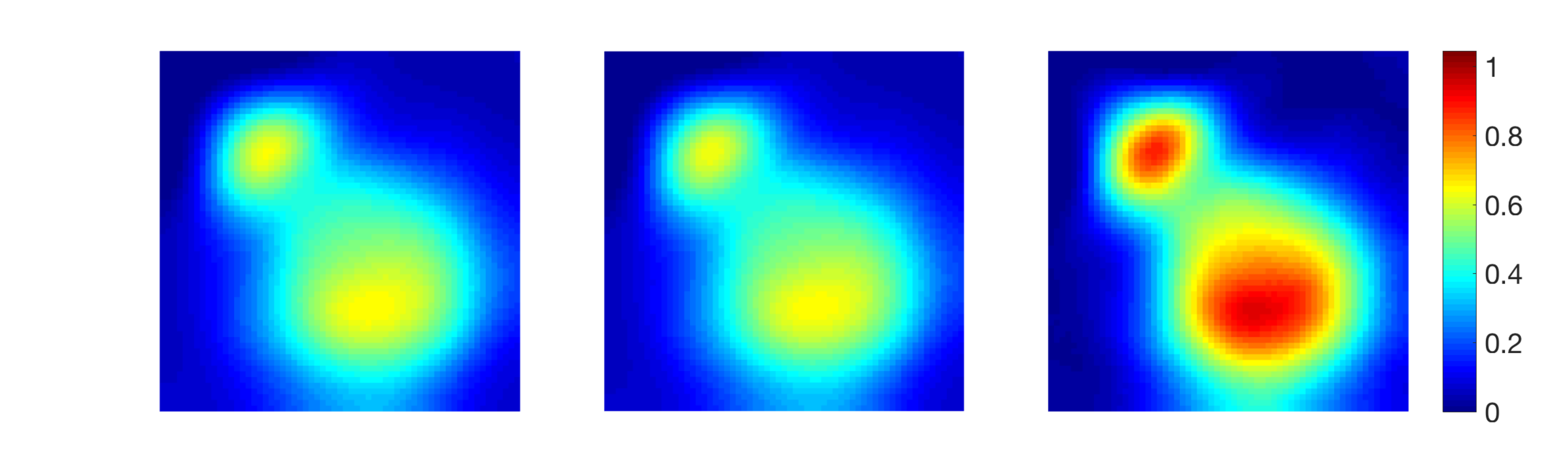}
}
\caption{(Left column) Conditional mean arising from conventional error model. (Middle column) Conditional mean arising from enhanced error model. (Right column) Conditional mean arising from iterative error model, iteration 10. From top to bottom, observational noise standard deviation is $10^{-2}$, $10^{-3}$, $10^{-4}$ respectively.}
\label{fig:gwf_means}
\end{figure}

\begin{figure}
\centerline{
\includegraphics[width=5.5cm]{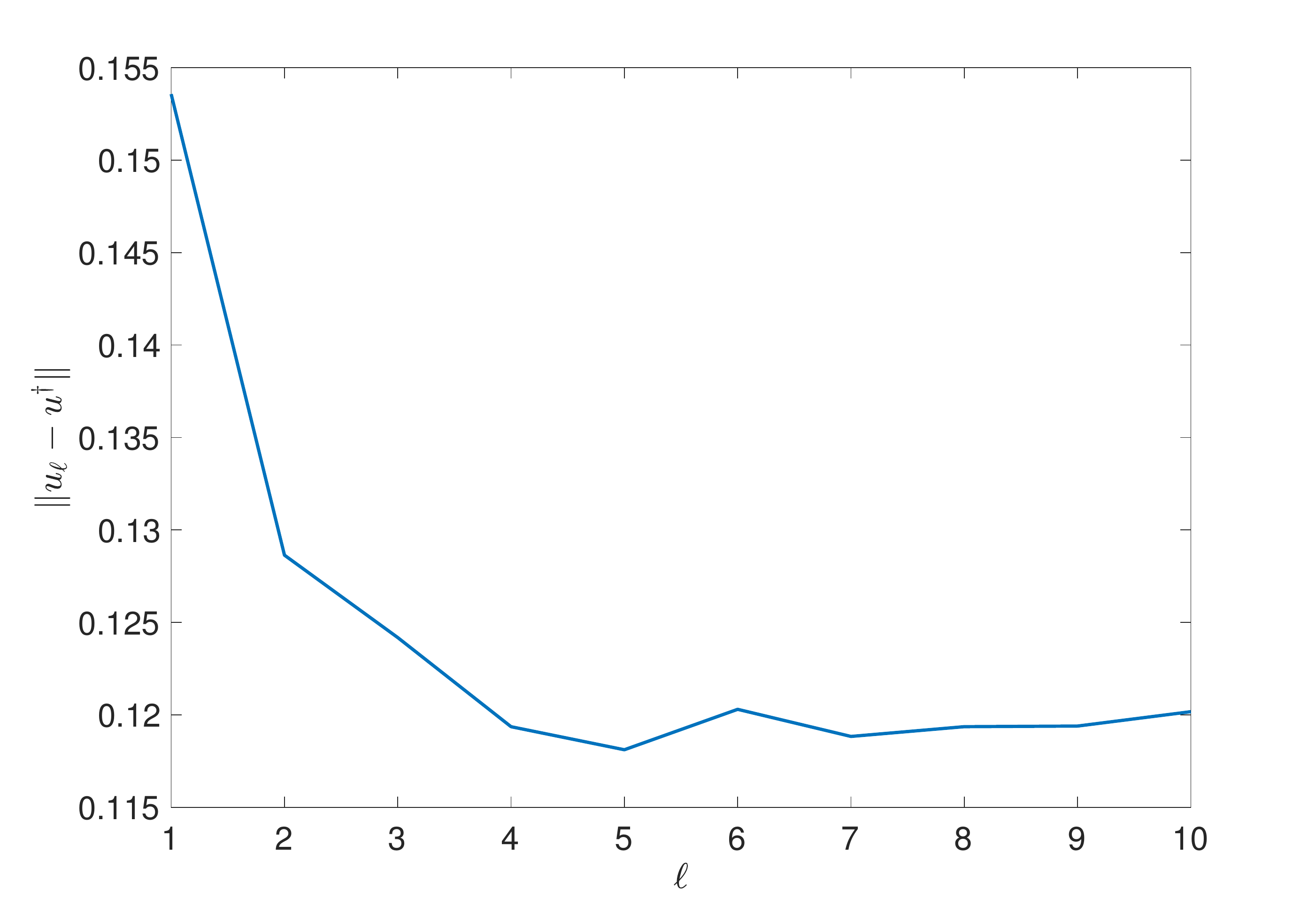}
\includegraphics[width=5.5cm]{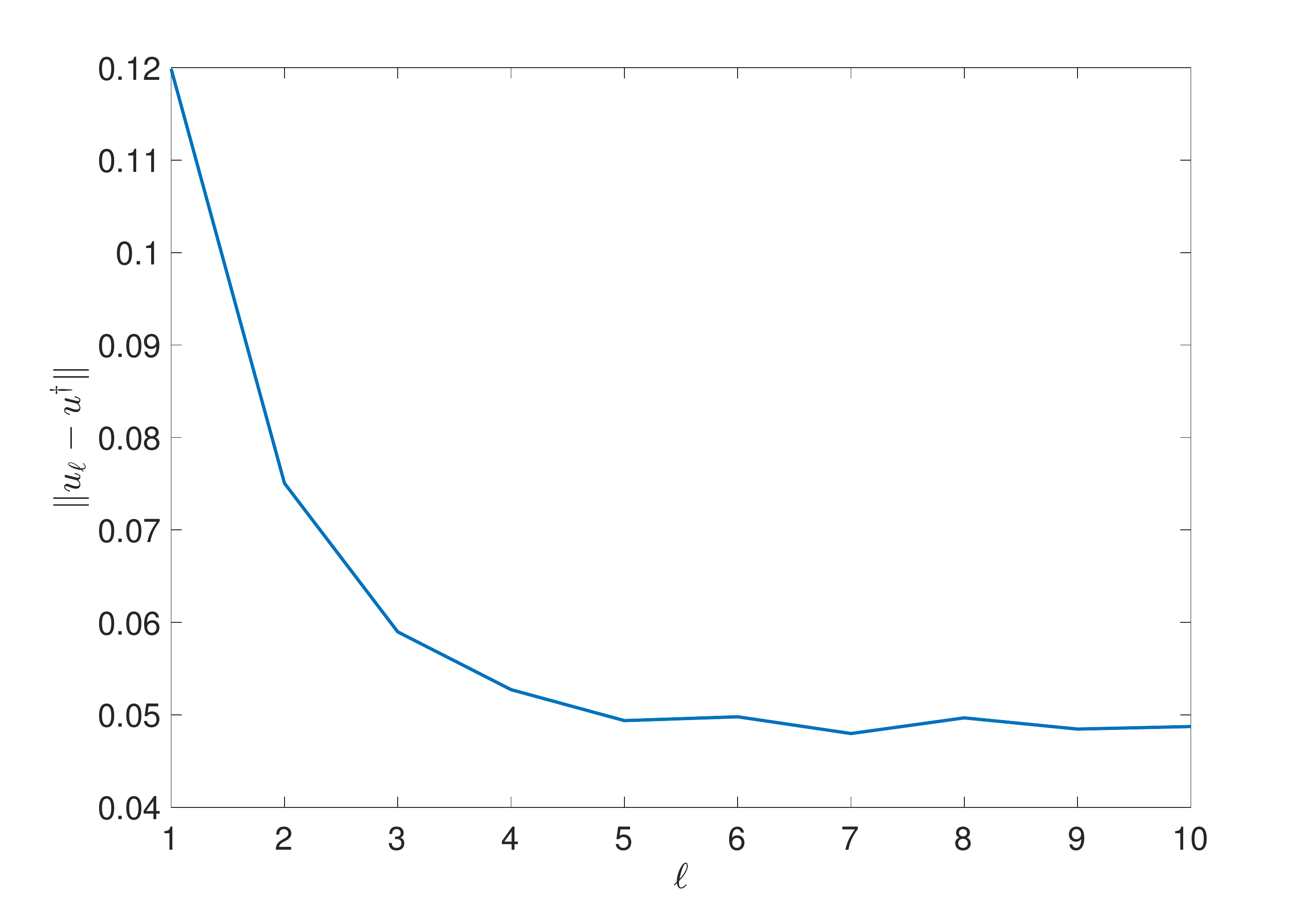}
\includegraphics[width=5.5cm]{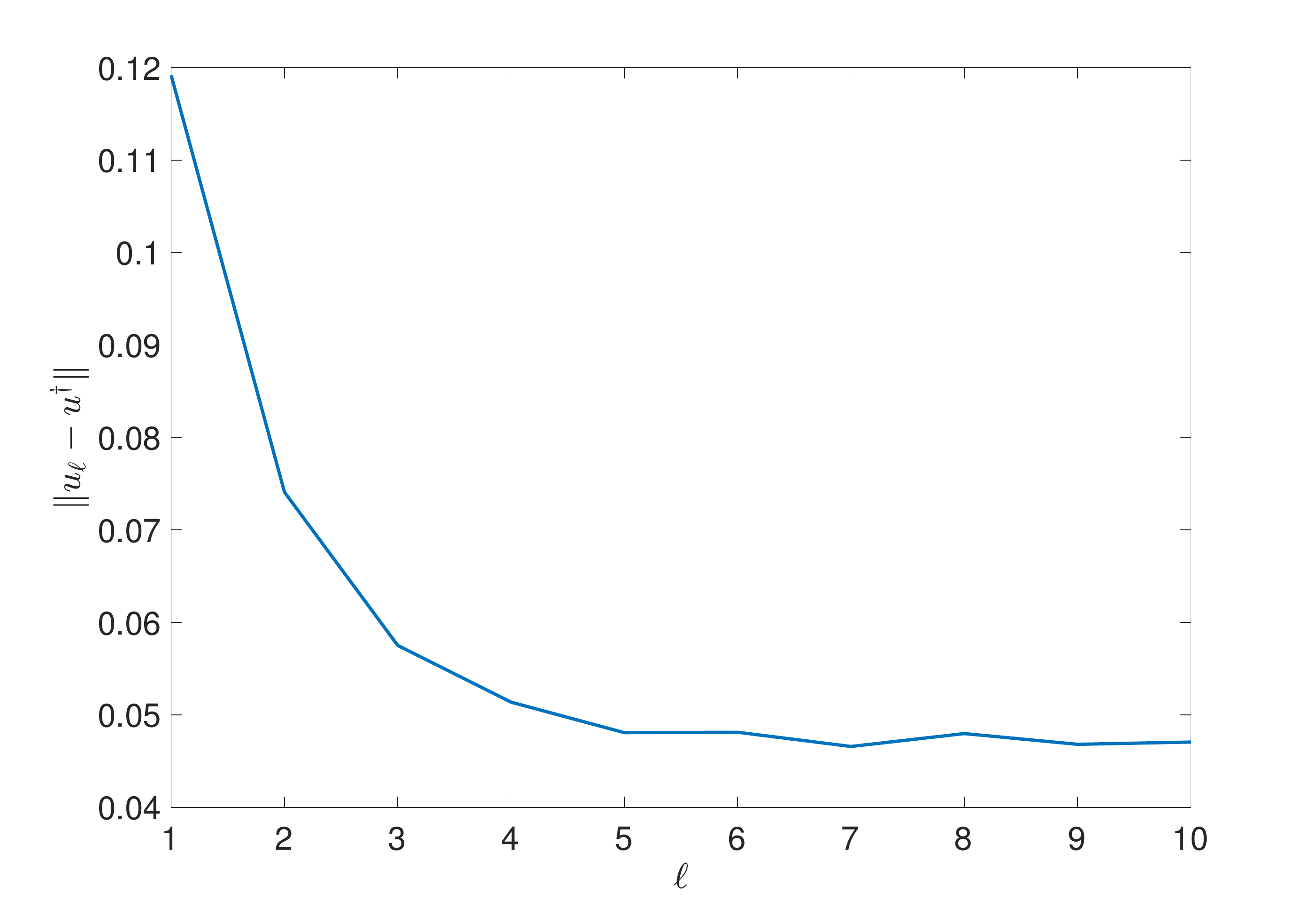}
}
\caption{Convergence of the error $\|u_\ell - u^\dagger\|$ between the conditional mean and the truth as the number of iterations increases. From left to to right observational noise standard deviation is $10^{-2}$, $10^{-3}$, $10^{-4}$ respectively.}
\label{fig:gwf_conv}
\end{figure}

\section{Conclusions}
\label{sec:conc}

Ill-posedness is a characteristic feature of inverse problems, and therefore, special attention needs to be paid to model uncertainties and model discrepancies that manifest themselves as highly correlated noise, deviating the measured data from the value predicted by the forward model. The modeling error is particularly detrimental when the quality of the data is good, and the exogenous noise does not mask the modeling errors that may become the predominant component of the noise. Quantification of the uncertainty due to the modeling errors is therefore an important part of successfully solving
the inverse problem. Modeling error depends on the unknown that is the target of the inverse problem, and therefore, the Bayesian framework provides a natural basis for attacking the problem: the unknown of interest, modeled as a random variable, can be used in a natural way to define the modeling error as a random variable, thus allowing a statistical interpretation of the modeling error. In this article we introduce,
and study the properties of, an iterative method of refining the statistical description of the modeling error as our information about the unknown increases.

From the implementational point of view, two cases in which the refinement of the modeling error distribution can be computed are identified. When the model is linear and the distributions are Gaussian, a fairly straightforward updating strategy of the posterior estimate is found, and convergence of this iteration can be shown. For non-linear inverse problems, a linearized approximate model leads to a tractable iterative algorithm based on particle approximations of the posterior, and as demonstrated in the numerical experiments, the computed point estimates can be very good, significantly improving on estimates which ignore model error. However, as pointed out in the article, the limiting approximate probability density obtained by the iterative algorithm is not identical to the Bayesian posterior density, although it may be close to it.  Regarding both the point estimate and the posterior it is important to
recognize that while the approximation error approach does requires a number of
evaluations of the expensive forward model, unlike traditional MCMC algorithms 
no rejections occur. Thus the methodology has potential to compute point estimates
more economically than conventional non-Bayesian approaches such as Tikhonov
regularization; and it also holds the potential to produce reasonable posterior
distributions at considerably lower cost than MCMC using the fully accurate
Bayesian posterior.  One of the future directions of research is to see how the approximation process proposed in this article can be effectively used to produce an estimate of the true posterior density.

\ack{The work of D Calvetti is partially supported by NSF grant DMS-1522334.
E Somersalo's work is partly supported by the NSF grant DMS-1312424. 
The research of AM Stuart was partially supported by the EPSRC programme grant EQUIP, by AFOSR Grant FA9550-17-1-0185 and ONR Grant N00014-17-1-2079. M Dunlop was partially supported by the EPSRC MASDOC Graduate Training Program. Both M Dunlop and AM Stuart
are supported by DARPA funded program Enabling Quantification of Uncertainty in Physical Systems (EQUiPS), contract W911NF-15-2-0121.}

\appendix
\section{Abstract Formulation of Algorithm}

Let $(Z,\mathcal{Z})$ be a measurable space, and given $A \in \mathcal{Z}$ define the indicator function $\mathbb{I}_A:Z\to\R$ by
\[
\mathbb{I}_A(z) =
\left\{\begin{array}{ll}
1 & z \in A\\
0 & z \notin A.
\end{array}
\right.
\]
Given two measures $\mu, \nu$ on $(Z,\mathcal{Z})$, let $\mu * \nu$ denote their convolution, i.e. the measure on $(Z,\mathcal{Z})$ given by
\[
(\mu*\nu)(A) = \int_{Z\times Z} \mathbb{I}_A(u+v)\mu(\dee u)\nu(\dee v)
\]
for any $A \in \mathcal{Z}$. Note that if we have $u \sim \mu$ and $v \sim \nu$ independently, then $u+v \sim \mu*\nu$.

{\bf Algorithm (General).} Let $\mu_0$ denote the prior distribution on $u$ and $\mathbb{Q}_0$ the distribution of the noise $\eps$. Given $v \in Y$, define $T_v:Y\rightarrow Y$ to be the translation operator $T_v(y) = y + v$. Set $\ell = 0$.
\begin{enumerate}[\hspace{0.2cm}1.]
\item Given $\mu_\ell$, assume $m \sim M^\#\mu_\ell$ independently of $\eps$, so $m + \eps \sim \mathbb{Q}_0^{(\ell+1)} := M^\#\mu_\ell *\mathbb{Q}_0$. The likelihood is given by
\[
b\mid u \sim \mathbb{Q}^{(\ell+1)}_u := T_{f(u)}^\#\mathbb{Q}_0^{(\ell+1)}.
\]
Assume that $\mathbb{Q}_u^{(\ell+1)} \ll \mathbb{Q}_0^{(\ell+1)}$, so that we have Radon-Nikodym density
\[
\frac{\dee \mathbb{Q}^{(\ell+1)}_u}{\dee \mathbb{Q}^{(\ell+1)}_0}(b) = \exp\big(-\Phi^{(\ell+1)}(u;b)\big).
\]
Bayes' Theorem gives the posterior distribution
\begin{eqnarray}
\label{eq:update_general}
\mu_{\ell+1}(\dee u) \propto \exp\big(-\Phi^{(\ell+1)}(u;b)\big)\mu_0(\dee u).
\end{eqnarray}
\item Set $\ell \mapsto \ell+1$ and go to 1.\qed
\end{enumerate}

The above iteration could be written more directly as
\[
\mu_{\ell+1}(\dee u) \propto \frac{\dee [T^\#_{f(u)}(M^\#\mu_\ell * \mathbb{Q}_0)]}{\dee[M^\#\mu_\ell * \mathbb{Q}_0]}(b)\,\mu_0(\dee u).
\]
though the expression (\ref{eq:update_general}) makes links with previous work on non-parametric Bayesian inverse problems clearer.

\end{document}